\documentclass[11pt]{article}
\usepackage[fleqn]{amsmath}
\usepackage{amssymb,amstext,mathrsfs}
\usepackage{bm}
\usepackage{rotating}
\usepackage{epstopdf}
\usepackage{booktabs}
\usepackage{amsthm}
\usepackage{float}
\usepackage{caption}
\usepackage{array}
\usepackage{fancyhdr}
\usepackage{graphicx}
\usepackage{tabularx}
\usepackage{multicol}
\usepackage{multirow}
\usepackage{longtable}
\usepackage{threeparttable}
\usepackage{titlesec}
\usepackage{titletoc}
\usepackage{indentfirst}
\usepackage{color}
\usepackage{subcaption}
\usepackage{natbib}
\usepackage{url}
\usepackage{stmaryrd}
\usepackage{verbatim}
\usepackage{accents}
\usepackage{statex}
\usepackage{psfrag,epsf}
\usepackage{enumerate}
\usepackage{enumitem}
\usepackage{setspace}
\usepackage{xcolor}
\usepackage{hyperref}
\usepackage{algorithm}
\usepackage{algorithmic}

\makeatletter
\def\widebar{\accentset{{\cc@style\underline{\mskip10mu}}}}
\def\Widebar{\accentset{{\cc@style\underline{\mskip8mu}}}}
\makeatother

\newtheorem{theorem}{Theorem}[section]

\newtheorem{corollary}{Corollary}[section]
\newtheorem{lemma}{Lemma}[section]
\newtheorem{definition}{Definition}[section]
\newtheorem{remark}{Remark}[section]
\newtheorem{assumption}{Assumption}[section]

\makeatletter

\newcommand{\Rmnum}[1]{\expandafter\@slowromancap\romannumeral #1@}
\makeatother

\hypersetup{CJKbookmarks,
	bookmarksnumbered,
	colorlinks,
	linkcolor=blue,
	citecolor=blue,
	plainpages=false,
	pdfstartview=FitH}

\allowdisplaybreaks[4]

\begin{document}
\title{\bf Testing for Serial Independence via Auto Hilbert-Schmidt Independence Criterion}
\author{Muyi Li$^{1,2},$
Yuqing Xu$^{2}$, Zhou Zhou$^{3}$ \vspace{10pt}\\\\
\small{
$^{1}$ Wang Yanan Institute for Studies in Economics (WISE), Xiamen University}\\
\small{
$^{2}$ Department of Statistics and Data Science, School of Economics, Xiamen University}\\
\small{$^{3}$ Department of Statistical Sciences, University of Toronto}
}
\date{}
\maketitle
\setstretch{1.25}
\begin{abstract}
We develop a Hilbert--Schmidt independence criterion (HSIC)-based framework for testing serial independence in strictly stationary time series. The proposed auto Hilbert--Schmidt independence criterion (AutoHSIC) measures dependence between an observation and its lagged counterpart, providing a kernel-based approach to detecting nonlinear serial dependence. The empirical AutoHSIC statistic is a lagged U-statistic constructed from overlapping observations, and hence inherits temporal dependence even under the i.i.d. null. Its asymptotic analysis therefore differs from standard i.i.d. HSIC theory and must account for degeneracy under the null. We establish the limiting behaviour of the resulting single-lag and portmanteau tests under the null and under fixed alternatives. Since the limiting null distribution is non-pivotal, we develop a wild bootstrap procedure for critical value approximation and prove its asymptotic validity. The framework is further extended to residual-based model diagnostics, where parameter estimation affects the null distribution. Simulations and empirical applications illustrate its ability to detect nonlinear serial dependence in multivariate, functional and matrix time series.
\end{abstract}

\newpage
\section{Introduction}\label{sec1}

Testing serial dependence is a fundamental problem in time series analysis. Before modelling, dependence tests help determine whether a stochastic dynamic model is needed. After modelling, they provide diagnostic tools for assessing whether the fitted model has adequately captured the serial structure of the data. Classical diagnostics for serial uncorrelatedness include time-domain
portmanteau tests based on sample autocorrelations, such as the Box--Pierce and Ljung--Box tests \citep{box1970, ljung1978}, and spectral-domain tests for
serial correlation \citep{hong1996}. However, uncorrelation does not characterise general independence. A time series may be serially uncorrelated while still exhibiting nonlinear dependence, as in many ARCH-type processes \citep{engle1982arch, bollerslev1986garch}. This limitation has motivated a large literature on nonlinear dependence testing.

Several approaches have been developed for this purpose. Rank-based measures, such as Spearman's $\rho$ and Kendall's $\tau$, are useful for monotone dependence but are not designed to detect all alternatives. 
Information theoretic approaches offer another route to nonlinear serial dependence testing. Mutual information characterises independence in principle \citep{shannon1948mathematical}, while entropy-based tests have been developed for general linear dependence \citep{giannerini2015entropy}. These methods typically require density estimation, optimisation, or bootstrap calibration.

Generalised spectral methods \citep{hong1999hypothesis} characterise serial dependence through Fourier transforms of generalised covariance functions, although their extension to high-dimensional settings is challenging. Distance correlation \citep{szekely2007} provides a prominent distance-based measure of nonlinear dependence by comparing the joint characteristic function with the product of the marginals. For time series, \citet{zhou2012measuring} introduced the auto distance correlation function, while \citet{davis2018applications} developed a systematic theory for empirical auto- and cross-distance correlation functions of stationary time series under ergodicity and strong mixing. Subsequent work has developed related tests for serial independence, goodness of fit assessment, and object-valued or metric space data \citep{fokianos2017, wan2022goodness, jiang2024testing}.

Kernel methods provide a complementary and highly flexible route to nonlinear dependence testing. The Hilbert--Schmidt independence criterion (HSIC), introduced by \citet{gretton2005}, measures the discrepancy between the embedding of a joint distribution and the tensor product of its marginal embeddings in reproducing kernel Hilbert spaces. With characteristic kernels, HSIC is zero if and only if the two variables are independent \citep{Fukumizu2008characteristic}. This property gives HSIC a clear population interpretation as a dependence measure, while its empirical version has a simple U-statistic form and avoids explicit density estimation or optimisation. The choice of kernel also makes HSIC adaptable to different data structures. For Euclidean data, Gaussian and Laplace kernels provide flexible nonlinear features; for non-Euclidean objects, suitable positive definite kernels allow dependence to be measured without reducing the data to a small number of coordinates. Moreover, bounded kernels can substantially weaken moment requirements, which is useful for heavy-tailed data. The connection between HSIC and distance covariance established by \citet{sejdinovic2013} shows that kernel-based and distance-based dependence measures share a common RKHS foundation, providing a unified framework for constructing dependence measures for complex time series objects.

Despite these advantages, HSIC has not been fully developed as an inferential tool for testing serial independence within a single time series. The original HSIC theory was mainly developed for independent observations \citep{gretton2005, gretton2007}. Existing extensions of kernel dependence methods address related but different problems. For example, \citet{fukumizu2007kernel} developed RKHS-based measures of conditional dependence, and \citet{zhang2011kernel} proposed kernel methods for conditional independence testing. \citet{chwialkowski2014kernel} considered independence testing between two stochastic processes, allowing temporal dependence within each process. \citet{pfister2018kernel} developed kernel tests for joint independence. More recent work has studied conditional mean independence \citep{lai2021kernel}, independence testing between two stationary time series through model-based residuals \citep{wang2021new}, high-dimensional independent data \citep{zhang2023statistical}, and factor models \citep{xu2024hsic}. 
Recently, \citet{Ghoshal2026resampling} developed a resampling-free RKHS embedding approach for several time series inference problems, including goodness-of-fit, change point testing, and independence testing between two time series. These contributions demonstrate the broad usefulness of RKHS methods, but they do not directly address serial independence testing within a single process.

The distinction is important. A natural HSIC-based serial dependence statistic is constructed from lagged pairs $(X_t,X_{t-m})$. These pairs overlap, and the resulting U-statistic inherits temporal dependence even under the i.i.d. null. Moreover, the centred HSIC kernel is degenerate under independence. Hence standard i.i.d. HSIC theory is not directly applicable: valid inference requires asymptotic arguments for degenerate U-statistics under temporal dependence, and the limiting null distribution is non-pivotal.

This paper develops an HSIC-based framework for testing serial independence in strictly stationary time series. For a given lag $m$, we define the auto Hilbert--Schmidt independence criterion (AutoHSIC) to measure dependence between $X_t$ and $X_{t-m}$. Based on this measure, we construct single-lag and portmanteau tests for serial independence. The framework is designed to complement classical autocorrelation-based diagnostics by targeting nonlinear forms of lagged dependence.

The contributions of the paper are fourfold. First, we formulate AutoHSIC as a lag-specific kernel measure of serial dependence for a strictly stationary process, and develop single-lag and portmanteau statistics for testing serial independence.
Second, we establish the asymptotic behaviour of the resulting U-statistics, explicitly accounting for overlapping lagged pairs and degeneracy under the null. Third, since the limiting null distribution is nonstandard and non-pivotal, we develop a wild bootstrap procedure that consistently approximates the null distribution and yields feasible critical values for the proposed tests. Fourth, we extend the framework to residual diagnostics for fitted time series models, and show how parameter estimation affects the asymptotic behaviour of the test statistics.

The rest of the paper is organised as follows. Section~\ref{sec2_def_ahsic} reviews HSIC and introduces AutoHSIC. Section~\ref{sec3_test} develops the single-lag and portmanteau tests and establishes their asymptotic properties. Section~\ref{sec_wildboot} presents the wild bootstrap. Section~\ref{sec4_simulation} reports the simulation study. Section~\ref{sec_modelcheck} develops the residual-based diagnostic extension. Section~\ref{sec_realdata} gives the empirical applications. Section~\ref{sec_conclusion} concludes. Proofs and additional numerical results are given in the supplementary material.

\section{HSIC and AutoHSIC}\label{sec2_def_ahsic}

We use the following notation. Let $\mathbb{R}=(-\infty,\infty)$ and let $\mathbb{Z}^+$ be the set of positive integers. For a vector or matrix $A$, $A^{\top}$ denotes its transpose and $\Vert A\Vert$ its Euclidean or Frobenius norm. Let $I_d$ be the $d\times d$ identity matrix and, for a function $h(x,y)$, let $\partial_xh$ denote the partial derivative with respect to $x$. 
$o_p(1)$ and $O_p(1)$ denote convergence to zero in probability and boundedness in probability, respectively, while $\xrightarrow{\mathrm{p}}$ and $\xrightarrow{\mathrm{d}}$ denote convergence in probability and in distribution. 
Let $\mathcal{H}$ be a Hilbert space of real-valued functions on a metric space $\mathcal{X}$, with inner product $\langle\cdot,\cdot\rangle_{\mathcal H}$ and norm $\Vert \cdot\Vert _{\mathcal H}$. 
Let $L^2[0,1]$ denote the Hilbert space of square-integrable functions on $[0,1]$.

\subsection{The Hilbert--Schmidt Independence Criterion}

We briefly review the Hilbert--Schmidt independence criterion (HSIC), retaining only the elements needed for the construction of AutoHSIC; see \cite{gretton2005} and \cite{sejdinovic2013} for further details.

Let $k$ and $l$ be symmetric positive-definite kernels on metric spaces $\mathcal{X}$ and $\mathcal{Y}$, with associated reproducing kernel Hilbert spaces (RKHSs) $\mathcal{H}_k$ and $\mathcal{H}_l$. For random objects $(X,Y)\sim P_{XY}$ with marginal distributions $P_X$ and $P_Y$, define the product kernel
\( \tilde{k}\big((x,y),(x',y')\big)=k(x,x')l(y,y') \).
The HSIC between $X$ and $Y$ is defined by
\begin{equation}\nonumber
\operatorname{HSIC}(X,Y)
=
\Vert  \mu_{\tilde{k}}(P_{XY}) - \mu_{\tilde{k}}(P_XP_Y) \Vert_{\mathcal{H}_{\tilde{k}}}^{2},
\end{equation}
where $\mu_{\tilde{k}}(\cdot)$ denotes the kernel mean embedding in the RKHS induced by $\tilde{k}$. Equivalently,
\begin{equation}\label{hsic_centered}
\operatorname{HSIC}(X,Y)
=
\mathbb{E}\{d_k(X,X')d_l(Y,Y')\},
\end{equation}
where $(X',Y')$ is an independent copy of $(X,Y)$, and
\begin{align}
d_k(x,x')
&=
k(x,x')
-\mathbb{E}\{k(x,X')\}
-\mathbb{E}\{k(X,x')\}
+\mathbb{E}\{k(X,X')\}, \label{d_k}\\
d_l(y,y')
&=
l(y,y')
-\mathbb{E}\{l(y,Y')\}
-\mathbb{E}\{l(Y,y')\}
+\mathbb{E}\{l(Y,Y')\}. \label{d_l}
\end{align}

For a kernel $k$, write
\begin{equation}\label{Mk_alpha}
\mathcal{M}_k^\alpha(\mathcal{X})
=
\left\{
P:\int k^\alpha(x,x)\, \mathrm{d}P(x)<\infty
\right\}.
\end{equation}
If $P_X\in\mathcal{M}_k^1(\mathcal{X})$, $P_Y\in\mathcal{M}_l^1(\mathcal{Y})$, and the RKHSs $\mathcal{H}_k$ and $\mathcal{H}_l$ are characteristic, then $\operatorname{HSIC}(X,Y)=0$ if and only if $P_{XY}=P_XP_Y$. Thus, under characteristic kernels, HSIC characterises independence.

\begin{remark}[Choice of kernel]\label{remark_kernel}
The choice of kernel affects both the sensitivity of HSIC and the moment conditions for inference. Bounded characteristic kernels, such as the Gaussian, Laplacian and Cauchy kernels, automatically satisfy the kernel moment conditions used below and hence do not require moment assumptions directly on the data distribution.
This is useful for heavy-tailed time series. By contrast, distance-induced kernels, such as the Brownian distance kernel
\(
k(x,x') = \Vert x\Vert _{\mathcal{X}} + \Vert x'\Vert _{\mathcal{X}} - \Vert x-x'\Vert _{\mathcal{X}},
\)
are unbounded and typically require additional moment assumptions. With this Brownian distance kernel, HSIC is equivalent to distance covariance.
\end{remark}

\subsection{AutoHSIC and Its Sample Analogue}

HSIC measures dependence between two random objects. To adapt this idea to serial dependence, we apply HSIC to the lagged pair $(X_t,X_{t-m})$. This gives the auto Hilbert--Schmidt independence criterion (AutoHSIC), analogous to the auto distance correlation function of \cite{zhou2012measuring}.

Consider $\mathcal{X}=\mathcal{Y}$. Let $\{X_t\}$ be a strictly stationary time series taking values in $\mathcal{X}$, with observations $\{X_t\}_{t=1}^T$. To measure serial dependence at lag $m$, we define the AutoHSIC between $X_t$ and $X_{t-m}$ as follows.

\begin{definition}[AutoHSIC]\label{autohsic_def}
For $m\in\mathbb{Z}^{+}$, the AutoHSIC of $\{X_t\}$ at lag $m$ is defined by
\begin{equation}\nonumber
V_m = \operatorname{HSIC}(X_t,X_{t-m}) = \mathbb{E}\{d_k(X_t,X_t')d_l(X_{t-m},X_{t-m}')\},
\end{equation}
where $Z_{t,m}'=(X_t',X_{t-m}')$ is an independent copy of $Z_{t,m}=(X_t,X_{t-m})$, and $d_k(\cdot,\cdot)$ and $d_l(\cdot,\cdot)$ are defined in \eqref{d_k} and \eqref{d_l}, respectively. Moreover, we define $V_0=\operatorname{HSIC}(X_t,X_t)\ge 0$ and set $V_m=V_{-m}$ for $m<0$.
\end{definition}

If the kernels are characteristic and the required moment conditions hold, then, for $m \neq 0$, $V_m=0$ if and only if $X_t$ and $X_{t-m}$ are independent. 
This property motivates the serial independence tests developed in the next section.
Inspired by \cite{szekely2014partial}, we estimate $V_m$ using the $\mathcal{U}$-centring approach. Let \(k_{ij}=k(X_i,X_j)\), for \(m+1\leq i,j\leq T\).
Its $\mathcal{U}$-centred version is
{\small
\begin{equation}\nonumber
a_{i j,m}= \begin{cases}k_{i j}-\dfrac{\sum_{t=m+1}^T k_{i t}}{T-m-2}-\dfrac{\sum_{t=m+1}^T k_{tj}}{T-m-2}+\dfrac{\sum_{m+1 \leq t \neq t^{\prime} \leq T} k_{t t^{\prime}}}{(T-m-1)(T-m-2)}, & i \neq j, \\ 0, & i=j .\end{cases}
\end{equation}}
Define $b_{ij,m}$ analogously for \( l_{ij,m}=l_{i-m,j-m}=l(X_{i-m},X_{j-m})\), \( m+1\leq i,j\leq T\).
The $\mathcal{U}$-centred quantities $a_{ij,m}$ and $b_{ij,m}$ approximate $d_k(X_i,X_j)$ and $d_l(X_{i-m},X_{j-m})$, respectively. We then estimate $V_m$ by
\begin{equation}\label{V_T(m)}
V_{T,m} = \frac{1}{(T-m)(T-m-3)} \sum_{i,j=m+1}^{T}a_{ij,m}b_{ij,m}.
\end{equation}
For the asymptotic analysis, it is useful to write $V_{T,m}$ as a fourth-order U-statistic. Following \citet{gretton2007} and \citet{zhang2018conditional},
\begin{equation}\label{VTm_Ustat}
V_{T,m} = \binom{T-m}{4}^{-1} \sum_{m+1\leq i<j<q<r\leq T}
h(Z_{i,m},Z_{j,m},Z_{q,m},Z_{r,m}),
\end{equation}
where $Z_{t,m}:=(X_t,X_{t-m})$, $t=m+1,\ldots,T$, and the symmetric kernel
$h:(\mathcal{X}\times\mathcal{X})^4\to\mathbb{R}$ is defined by
{\small
\begin{equation}\label{h_kernel}
\begin{aligned}
h(z_i,z_j,z_q,z_r) = \frac{1}{4!} \sum_{(i_1,i_2,i_3,i_4)} k(x_{i_1},x_{i_2}) \Big\{ 
&\,l(y_{i_3},y_{i_4}) +l(y_{i_1},y_{i_2}) -2l(y_{i_1},y_{i_3})
\Big\},
\end{aligned}
\end{equation}}
where $z_t=(x_t,y_t)$ and the summation is over all $4!$ permutations of the index set $\{i,j,q,r\}$.

To establish the consistency of $V_{T,m}$, we impose the following assumptions.

\begin{assumption}\label{ass5.1_stationary}
$\{X_t\}$ is strictly stationary.
\end{assumption}

\begin{assumption}\label{assumption_bound_2+2delta}
The marginal distribution $P$ of $X_t$ satisfies $ P\in\mathcal{M}_k^{2+r}(\mathcal{X})$ and $P\in\mathcal{M}_l^{2+r}(\mathcal{X})$ for some $r>0$, where $\mathcal{M}_k^{2+r}(\cdot)$ and $\mathcal{M}_l^{2+r}(\cdot)$ are defined in \eqref{Mk_alpha}.
\end{assumption}

\begin{assumption}\label{assumption1_beta}
$\{X_t\}$ is a $\beta$-mixing process, with mixing coefficient $\beta(n)$ satisfying
\(
\beta(n)=O\left(n^{-(2+r')/r'}\right),\quad
\text{for some }0<r'<r/2.
\)
\end{assumption}

Assumption~\ref{ass5.1_stationary} is standard in time series analysis and is maintained throughout the paper. Assumption~\ref{assumption_bound_2+2delta} imposes moment conditions on the kernels, which are automatically satisfied by bounded kernels.
Assumption~\ref{assumption1_beta} is used for the asymptotic analysis of U-statistics under weak dependence; see \cite{yoshihara1976limiting}. Conditions for verifying $\beta$-mixing in nonlinear time series models are discussed, for example, in \cite{Pham1985}, \cite{fan2003nonlinear} and \cite{Carrasco2002mixing}.

The following theorem establishes the consistency of the empirical AutoHSIC \(V_{T,m}\) for \(V_m\) at each fixed lag. Since \(V_{T,m}\) is a U-statistic based on the dependent lagged pairs \(Z_{t,m}=(X_t,X_{t-m})\), the proof relies on U-statistic arguments under weak dependence rather than standard i.i.d. HSIC theory.

\begin{theorem}\label{thm1_VTm_consis}
Suppose that Assumptions~\ref{ass5.1_stationary}--\ref{assumption1_beta} hold. Then, for any fixed lag $m\in\mathbb{Z}^{+}$,
\( V_{T,m}\xrightarrow{\mathrm{p}}V_m \), as \(T\to\infty\).
\end{theorem}

\section{Test Statistics and Asymptotic Theory}\label{sec3_test}

We now develop AutoHSIC-based tests for serial independence. For a fixed maximum lag $M\in\mathbb{Z}^{+}$, the aim is to detect whether the process exhibits lag-specific dependence at any lag up to $M$. We formulate the hypotheses as
\begin{equation}\nonumber
H_0:\ \{X_t\}\ \text{is i.i.d.},
\qquad
H_{a,M}:\ V_m>0\ \text{for some }m\in\{1,\ldots,M\}.
\end{equation}
The null hypothesis is stated as an i.i.d. condition, while the alternative is expressed through the population AutoHSIC values. Thus the test targets departures from serial independence that are visible through pairwise lagged dependence within the first $M$ lags. This formulation is analogous in spirit to classical portmanteau tests, but replaces autocorrelation by a kernel-based dependence measure.

Based on the empirical AutoHSIC values $V_{T,m}$, define the portmanteau statistic
\begin{equation}\nonumber 
P_{T,M} = \sum_{m=1}^M V_{T,m}.
\end{equation}
Large values of $P_{T,M}$ provide evidence against the null, indicating the presence of nonlinear serial dependence at one or more of the specified lags.

We first derive the null distribution. The following moment condition is imposed in addition to strict stationarity.

\begin{assumption}\label{assumption2_H0}
The marginal distribution $P$ of $X_t$ satisfies
$P\in\mathcal{M}_k^4(\mathcal{X})$
and $P\in\mathcal{M}_l^4(\mathcal{X})$,
where $\mathcal{M}_k^4(\cdot)$ and $\mathcal{M}_l^4(\cdot)$ are defined in
\eqref{Mk_alpha}.
\end{assumption}

Assumption~\ref{assumption2_H0} controls the fourth-order kernel moments that enter the degenerate U-statistic limit under the null. It is automatically satisfied when bounded kernels are used. This is one reason why bounded characteristic kernels are particularly convenient for heavy-tailed or non-Euclidean time series.

The asymptotic analysis under \(H_0\) differs from standard i.i.d. HSIC theory because the lagged pairs \( \{ Z_{t,m}=(X_t,X_{t-m}) \}_{t=m+1}^T \) overlap and are \(m\)-dependent, even when \(\{X_t\}\) is i.i.d.  Hence \(V_{T,m}\) is a degenerate U-statistic based on dependent lagged pairs, rather than independent observations.

\begin{theorem}\label{thm2_nulldis}
Under $H_0$, suppose that Assumptions~\ref{ass5.1_stationary} and \ref{assumption2_H0} hold. Then, for any fixed $M\in\mathbb{Z}^{+}$,
\begin{equation}\nonumber
T\left(V_{T,1},V_{T,2},\ldots,V_{T,M}\right)^\top
\xrightarrow{\mathrm{d}}
(\xi_1,\xi_2,\ldots,\xi_M)^\top,
\qquad T\to\infty,
\end{equation}
where
\(
\xi_m \stackrel{\mathrm{d}}{=} \sum_{\ell=1}^{\infty}\lambda_\ell\left(G_{\ell,m}^2-1\right)
\), for \(m=1,\ldots,M\).
Here $\{G_{\ell,m}:\ell\ge1,\ m=1,\ldots,M\}$ is a sequence of i.i.d. standard normal random variables, and $\{\lambda_\ell\}_{\ell\ge1}$ and $\{\Phi_\ell\}_{\ell\ge1}$ are the non-zero eigenvalues and orthonormal eigenfunctions satisfying
\begin{equation}\nonumber
\mathbb{E}\{\mathcal{K}(z,Z_{t,m})\Phi_\ell(Z_{t,m})\}
=
\lambda_\ell\Phi_\ell(z),
\end{equation}
with
\(
\mathcal{K}(z_1,z_2)
=
d_k(x_1,x_2)d_l(y_1,y_2)\),
\(z_i=(x_i,y_i)\in\mathcal{X}\times\mathcal{X}\), \(i=1,2\).
\end{theorem}

Theorem~\ref{thm2_nulldis} shows that the null limit of the sample AutoHSIC is a weighted sum of centred chi-squared variables.  This is the characteristic limit of a first-order degenerate U-statistic, but here it arises in the presence of temporal dependence induced by the overlapping lag construction.
The distribution is non-pivotal because the eigenvalues \(\{\lambda_\ell\}\) depend on the unknown marginal distribution and the chosen kernels, which motivates the bootstrap calibration in the next section. With the Brownian distance kernel, the result reduces to the corresponding limit for distance covariance; see, for example, Theorem~1 of \cite{jiang2024testing}.

The following corollary follows from the continuous mapping theorem.

\begin{corollary}\label{corollary_singledis}
Under $H_0$, suppose that Assumptions~\ref{ass5.1_stationary} and
\ref{assumption2_H0} hold. 
Then,
\( T V_{T,m} \xrightarrow{\mathrm{d}} \xi_m \), as \(T \to \infty\), for any fixed $m \in \{1,\ldots,M\}$; 
and \( T P_{T,M} \xrightarrow{\mathrm{d}} \sum_{m=1}^M \xi_m\), as \(T \to \infty\), 
where $\xi_m$ is defined in Theorem~\ref{thm2_nulldis}.
\end{corollary}

Corollary~\ref{corollary_singledis} provides the limiting null distributions of both the single-lag and portmanteau statistics. These limits justify rejecting the null for large values of $T V_{T,m}$ or $T P_{T,M}$. Since the limiting distributions depend on unknown eigenvalues, feasible inference requires a data-driven approximation.

We next study the behaviour of the statistic under fixed alternatives. The following stronger moment condition is used to derive the asymptotic distribution of $V_{T,m}$ when the corresponding lag dependence is present.

\begin{assumption}\label{assumption_bound_4+q}
The marginal distribution $P$ of $X_t$ satisfies
$P\in\mathcal{M}_k^{4+r}(\mathcal{X})$
and $P\in\mathcal{M}_l^{4+r}(\mathcal{X})$
for some $r>0$, where $\mathcal{M}_k^{4+r}(\cdot)$ and
$\mathcal{M}_l^{4+r}(\cdot)$ are defined in \eqref{Mk_alpha}.
\end{assumption}

Assumption~\ref{assumption_bound_4+q} strengthens the null moment condition to ensure a well-defined long-run variance for the first-order projection under alternatives. For bounded kernels, this condition is again automatically satisfied.

\begin{theorem}\label{thm_dis_Ha}
Under $H_{a,M}$, suppose that Assumptions~\ref{ass5.1_stationary}, \ref{assumption1_beta} and \ref{assumption_bound_4+q} hold. Then, for any lag $m\in\{1,\ldots,M\}$ such that $V_m>0$,
\begin{equation}\nonumber
\sqrt{T}(V_{T,m}-V_m) \xrightarrow{\mathrm{d}} N(0,16\sigma^2), \qquad T\to\infty,
\end{equation}
where \( \sigma^2 = \sigma_1^2+2\sum_{i=1}^{\infty}\sigma_{1,i}^2 \),
\( \sigma_1^2=\operatorname{Var}\{h_1(Z_{t,m})\} \),
\( \sigma_{1,i}^2=\operatorname{Cov}\{h_1(Z_{t,m}),h_1(Z_{t+i,m})\} \),
and
\( h_1(z) = \mathbb{E}\{h(z,Z,Z^\prime,Z^{\prime \prime})\} \)
with $Z,Z^\prime,Z^{\prime \prime}$ being i.i.d. copies of $Z_{t,m}$.
\end{theorem}

Theorem~\ref{thm_dis_Ha} gives the root-\(T\) limit of \(V_{T,m}\) under fixed alternatives, where the first-order projection is non-degenerate and its long-run variance determines the limiting variance.
It also implies consistency. Under \(H_{a,M}\), there exists some \(m \leq M\) such that \(V_m>0\). Since \(V_{T,m}\xrightarrow{\mathrm p}V_m\), it follows that \(T V_{T,m}\to\infty\) and \(T P_{T,M}\to\infty\) in probability. Thus the single-lag and portmanteau tests have asymptotic power one. 

\section{Wild Bootstrap}\label{sec_wildboot}

The limiting null distributions of \(V_{T,m}\) and \(P_{T,M}\) are non-pivotal, so we use a wild bootstrap to approximate the null distributions.

The bootstrap mimics the degenerate U-statistic structure of \(V_{T,m}\) under the null by keeping the centred kernel quantities \(a_{ij,m}\) and \(b_{ij,m}\) fixed and applying random weights to both indices. This double weighting reproduces the leading quadratic component without estimating the eigenvalues in Theorem~\ref{thm2_nulldis}. The procedure is summarised in Algorithm~\ref{algo1_boot}.

\begin{algorithm}
\caption{Wild bootstrap}\label{algo1_boot}
\small
\begin{algorithmic}

\STATE {\bf Step~1}: Given the maximum lag $M$, the original sample $\{X_t\}_{t=1}^T$, and the number
of bootstrap replications $B$.

\STATE {\bf Step~2}: Compute $V_{T,m}$, $m=1,\ldots,M$, and $P_{T,M}$ from $\{X_t\}_{t=1}^T$.

\STATE {\bf Step~3}:
\FOR{$b=1,\ldots,B$}

\STATE \parbox[t]{0.95\linewidth}{%
{\bf Step~3.1}: For each lag $m=1,\ldots,M$, generate an i.i.d. weight sequence
    $\{w_{t,m}^{(b)}\}_{t=1}^T$ satisfying
    $\mathbb{E}w_{t,m}^{(b)}=0$,
    $\mathbb{E}\{(w_{t,m}^{(b)})^2\}=1$, and
    $\mathbb{E}\{(w_{t,m}^{(b)})^4\}<\infty$.
For distinct lags $m\ne m'$, the weight vectors
    $\mathbf{w}_m^{(b)}=(w_{1,m}^{(b)},\ldots,w_{T,m}^{(b)})^\top$ and
    $\mathbf{w}_{m'}^{(b)}$ are generated independently. }
\STATE \parbox[t]{0.95\linewidth}{%
{\bf Step~3.2}: For $m=1,\ldots,M$, compute the single-lag statistic
    \[
    V_{T,m}^{*(b)}
    =
    \frac{1}{(T-m)(T-m-3)}
    \sum_{m+1\leq i\ne j\leq T}
    w_{i,m}^{(b)}a_{ij,m}b_{ij,m}w_{j,m}^{(b)},
    \]
and the bootstrap portmanteau statistic
    \[
    P_{T,M}^{*(b)}
    =
    \sum_{m=1}^M V_{T,m}^{*(b)} .
    \] }

\ENDFOR

\STATE {\bf Step~4}: Let $V_{T,m;1-\alpha}^*$ and $P_{T,M;1-\alpha}^*$ denote the empirical 
$1-\alpha$ quantiles of
$\{V_{T,m}^{*(b)}:b=1,\ldots,B\}$ and
$\{P_{T,M}^{*(b)}:b=1,\ldots,B\}$, respectively.
Reject $H_0$ in the single-lag test at lag $m$ if
$TV_{T,m}> TV_{T,m;1-\alpha}^*$, and reject $H_0$ in the portmanteau test if
$TP_{T,M}> TP_{T,M;1-\alpha}^*$.

\end{algorithmic}
\end{algorithm}

Let $V_{T,m}^*$ and $P_{T,M}^*$ denote generic bootstrap statistics generated by Algorithm~\ref{algo1_boot}, conditional on the original sample $\{X_t\}_{t=1}^T$. To state the bootstrap validity result, we use convergence in distribution in probability in the sense of \cite{li2003consistent}. Specifically, conditional on the original sample, we write ``$\xrightarrow{\mathrm{d}^*}$ in probability'' for convergence in distribution in probability, and write ``$O_p^*(1)$ in probability'' for boundedness in probability under the bootstrap law.

The following theorem establishes the validity of the wild bootstrap.

\begin{theorem}\label{thm_boot_null}
Suppose that Assumptions~\ref{ass5.1_stationary} and \ref{assumption2_H0} hold.
Assume that the bootstrap weights satisfy
\(
\mathbb{E}w_t=0,\;
\mathbb{E}w_t^2=1,\;
\mathbb{E}w_t^4<\infty
\).
Then the following statements hold.

\noindent $(a)$ Under $H_0$, for any fixed maximum lag $M\in\mathbb{Z}^{+}$,
\[
T\left(V_{T,1}^*,V_{T,2}^*,\ldots,V_{T,M}^*\right)^\top
\xrightarrow{\mathrm{d}^*}
(\xi_1,\xi_2,\ldots,\xi_M)^\top
\quad\text{in probability},
\qquad T\to\infty,
\]
where $(\xi_1,\ldots,\xi_M)^\top$ is the limiting vector in Theorem~\ref{thm2_nulldis}.

\noindent $(b)$ Under $H_{a,M}$ and Assumption~\ref{assumption1_beta},
\(
T V_{T,m}^*=O_p^*(1)
\) in probability,
for each fixed $m\in\{1,\ldots,M\}$, 
and
\(
T P_{T,M}^*=O_p^*(1)
\) in probability.

\end{theorem}

Theorem~\ref{thm_boot_null} shows that the wild bootstrap consistently approximates the limiting null distributions of \(T V_{T,m}\) and \(T P_{T,M}\), yielding feasible critical values. Under \(H_{a,M}\), the original statistics diverge while their bootstrap counterparts remain bounded, so the bootstrap tests have asymptotic power one against the alternatives. 

\section{Simulations}\label{sec4_simulation}

This section examines the finite-sample performance of the proposed AutoHSIC-based tests for high-dimensional, functional and matrix-valued time series. 
The simulations assess size under serial independence and power against nonlinear serial dependence not captured by linear autocorrelation.

Throughout the simulations, we use the Gaussian kernel (GK), the Laplacian kernel (LK) and the Brownian distance kernel (BDK).
The Gaussian and Laplacian kernels are defined as follows:

    \noindent $(a)$ Gaussian kernel:
    \(k(x,x')=\exp \left(-\Vert x-x'\Vert _{\mathcal{X}}^2/(2\gamma^2)\right)\),
    for some $\gamma>0$;

    \noindent $(b)$ Laplacian kernel:
    \(k(x,x')=\exp\left(-\Vert x-x'\Vert _{\mathcal{X}}/\gamma\right)\),
    for some $\gamma>0$.\\
Following \cite{zhang2023statistical}, we set the bandwidth \(\gamma\) to the median of \(\{\Vert X_i-X_j\Vert _{\mathcal X}:1\leq i<j\leq T\}\). 
Empirical rejection rates are based on 1000 replications, with critical values computed from 500 bootstrap replications using Rademacher weights,
\(\Pr(w_{t,m}^{(b)}=1)=\Pr(w_{t,m}^{(b)}=-1)=1/2\).
We consider sample sizes \(T=100,200\), single-lag tests with \(m=1,3\), and portmanteau tests with maximum lag \(M=3,6\).

\subsection{Multivariate and High-Dimensional Time Series}\label{subsec_simulation}

We first consider multivariate and high-dimensional time series with \(\mathcal X=\mathbb R^d\) and the Euclidean norm, where \(d=1,5,10,20,40,80,160\). Let \(\odot\) denote the Hadamard product, and let
\(\Sigma=(\Sigma_{ij})_{1\leq i,j\leq d}\) denote the matrix with \(\Sigma_{ij}=0.5^{|i-j|}\). 
We write \(N_d(\mathbf 0,\Sigma)\) for the \(d\)-dimensional normal distribution with mean zero and covariance matrix \(\Sigma\), and \(t_\nu(\Sigma)\) for the \(d\)-dimensional Student \(t\) distribution with degrees of freedom \(\nu\) and scale matrix \(\Sigma\).

We first evaluate size under three serially independent data-generating processes:

\noindent DGP~1: $X_t\stackrel{i.i.d.}{\sim}N_d(\mathbf 0,\Sigma)$; \quad
DGP~2: $X_t\stackrel{i.i.d.}{\sim}t_2(\Sigma)$;\quad
DGP~3: $X_t\stackrel{i.i.d.}{\sim}t_1(\Sigma)$.

DGP~1 provides a Gaussian benchmark, whereas DGPs~2 and 3 examine the robustness of the tests under heavy-tailed marginal distributions.

To evaluate power, we consider nonlinear moving-average processes of the form:

\noindent DGP~4: $X_t=\eta_t\odot\eta_{t-1}\odot\eta_{t-2}$, $\eta_t\stackrel{i.i.d.}{\sim}N_d(\mathbf 0,\Sigma)$.

\noindent DGP~5: $X_t=\eta_t\odot\eta_{t-1}\odot\eta_{t-2}$, $\eta_t\stackrel{i.i.d.}{\sim}t_2(\Sigma)$.

These processes are serially uncorrelated but not serially independent, and thus provide alternatives under which autocorrelation-based tests may have limited power while nonlinear serial dependence remains present.

The proposed single-lag and portmanteau tests are denoted by $V_{T,m}$ and $P_{T,M}$, respectively. We also report results for tests based on auto distance covariance (ADCV), introduced by \cite{zhou2012measuring} and further developed by \cite{jiang2024testing}. 
As discussed in Remark~\ref{remark_kernel}, ADCV corresponds to AutoHSIC with the Brownian distance kernel, up to the usual centring and normalisation conventions. 
Accordingly, the ADCV-based tests are reported as the BDK versions of \(V_{T,m}\) and \(P_{T,M}\).

Figures~\ref{figure_rejections_size} and~\ref{figure_rejections_power} report the empirical rejection rates at the $5\%$ nominal level for DGPs~1--3 and DGPs~4--5, respectively.
For the null DGPs~1--3, most procedures have rejection rates close to the nominal level as $T$ increases, indicating satisfactory size control. An exception is observed for the ADCV-based tests, such as $V_{T,1}$ with the BDK, which tend to overreject under DGP~3 when $d \geq 10$. This suggests that the ADCV-based tests may be sensitive to heavy tails.

For DGP~4 with Gaussian innovations, empirical power decreases monotonically as \(d\) increases. 
This pattern accords with  the high-dimensional behaviour of distance covariance and HSIC-type statistics \citep{zhu2020distance,
zhang2023statistical}, which become mainly sensitive to componentwise linear dependence in finite-moment regimes. 
Under DGP~4, however, the process is serially uncorrelated and the lagged dependence is purely nonlinear. 
As \(d\) grows, Euclidean distances and kernel values become increasingly concentrated, which in turn attenuates the contrast in the kernel matrices and renders the nonlinear lagged signal progressively harder to detect.

The behaviour under DGP~5 is different. Here the innovations follow a heavy-tailed \(t_2\) distribution, so the finite-moment high-dimensional theory does not directly apply.
The multiplicative structure
can make extreme innovations generate large neighbouring lagged products, creating strong nonlinear tail dependence despite vanishing linear autocorrelations. 
As \(d\) increases, extreme coordinates become more likely, potentially strengthening the detectable tail-driven signal. 
We do not claim that heavy tails universally enhance high-dimensional power; the phenomenon depends jointly on the tail index, the dependence structure, the choice of kernel, and the bandwidth selection rule. 
Nonetheless, the simulations indicate that bounded kernels can be particularly effective in heavy-tailed settings, as they capture tail-induced nonlinear dependence while mitigating the size distortions observed with the Brownian distance kernel.

Additional results, including those for the Laplacian kernel and other data-generating processes, are reported in Tables~6--10 of the supplementary material. These include results for a vector autoregressive process (Table~10), where both the AutoHSIC-based and ADCV-based tests show high power, confirming their ability to detect linear serial dependence.

\begin{figure}[!htbp]
\centering
\includegraphics[width=0.68\linewidth]{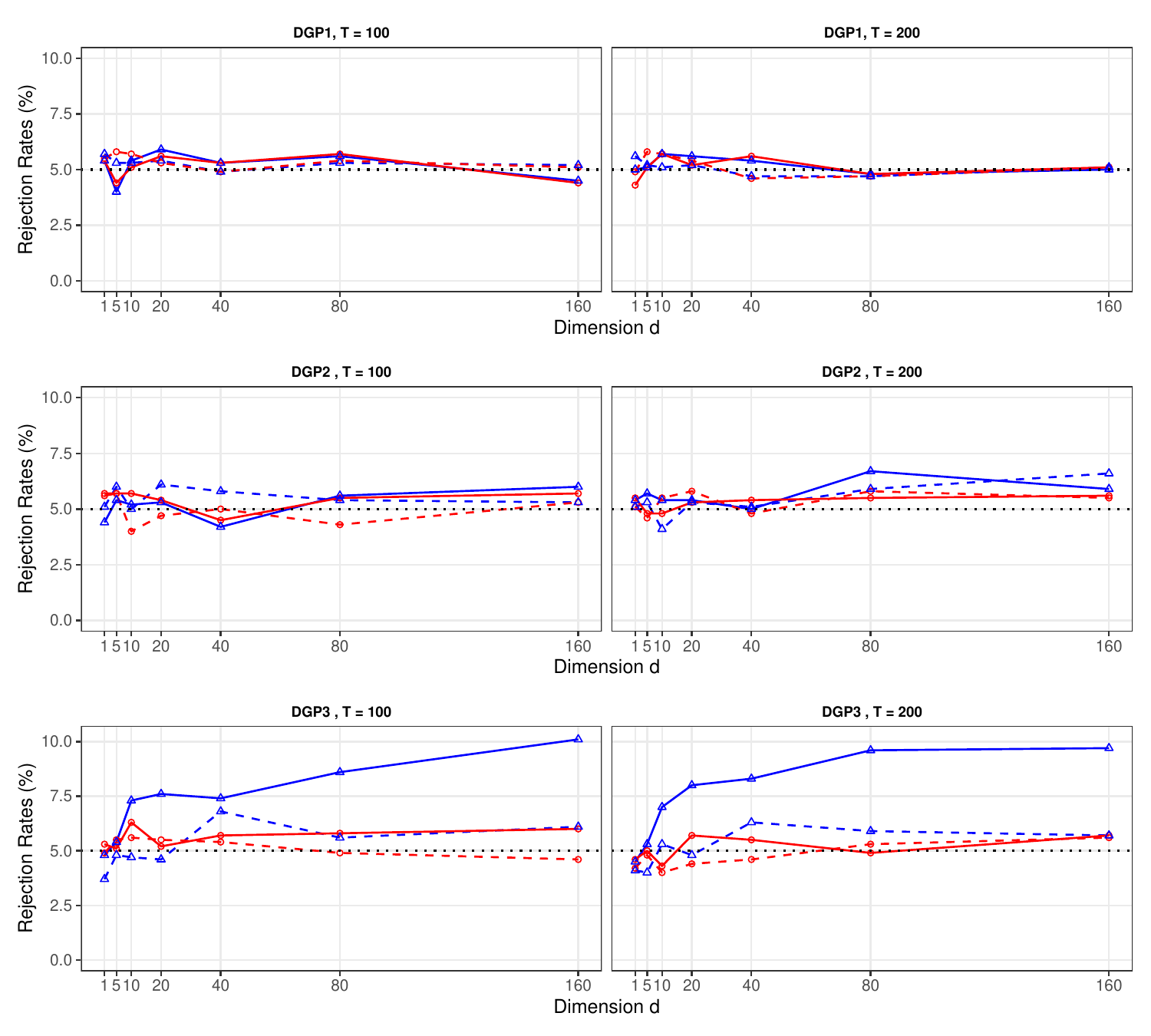}

\caption{\footnotesize Empirical sizes for DGPs~1--3 at the $5\%$ nominal level, plotted against the dimension $d$. Red lines denote tests based on the Gaussian kernel and blue lines denote tests based on the Brownian distance kernel; solid and dashed lines correspond to $V_{T,1}$, $P_{T,3}$, respectively.}\label{figure_rejections_size}
\end{figure}

\begin{figure}[!htbp]
\centering
\includegraphics[width=0.68\linewidth]{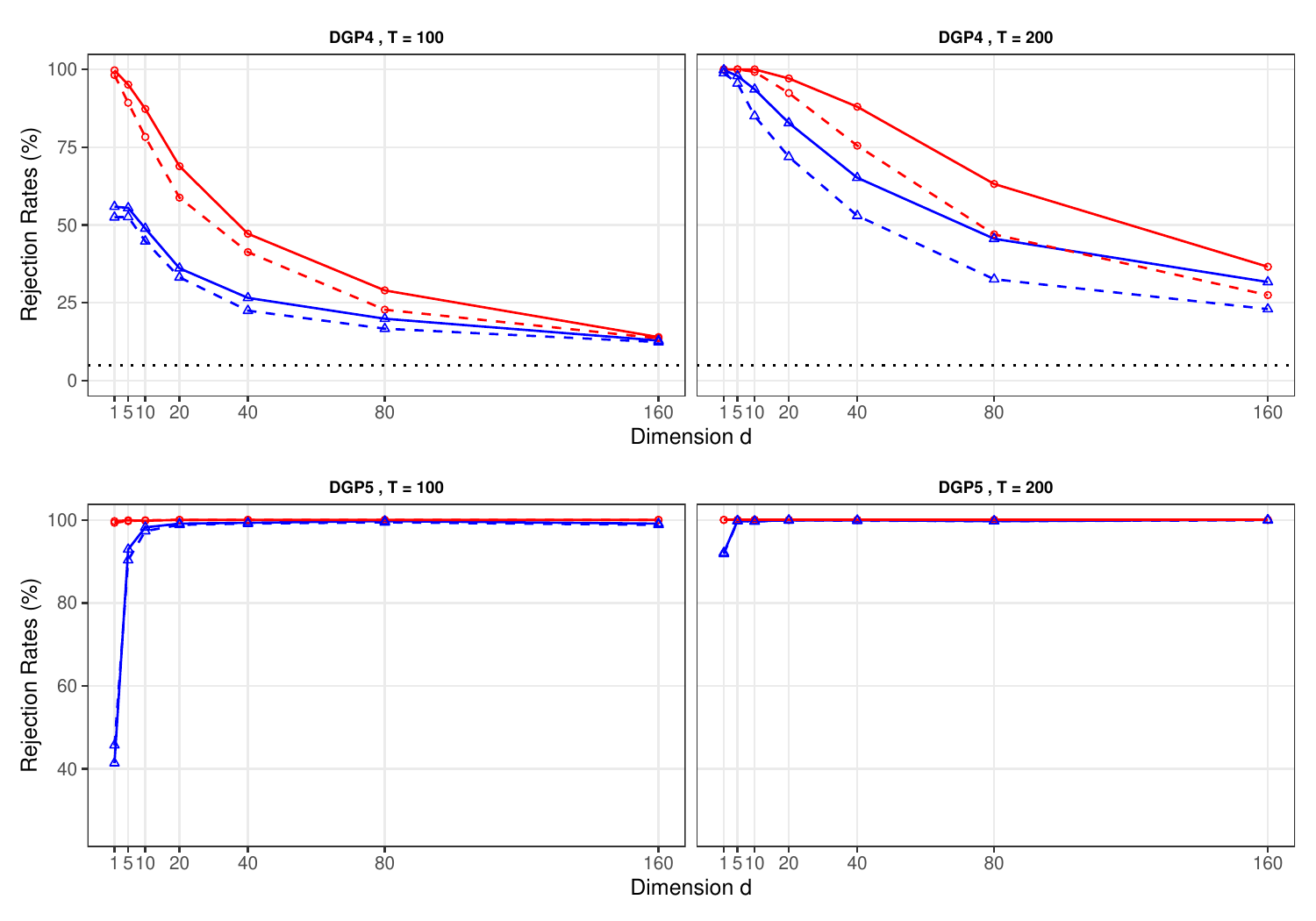}

\caption{\footnotesize Empirical powers for DGPs~4--5 at the $5\%$ nominal level, plotted against the dimension $d$. Red lines denote tests based on the Gaussian kernel and blue lines denote tests based on the Brownian distance kernel; solid and dashed lines correspond to $V_{T,1}$, $P_{T,3}$, respectively.}\label{figure_rejections_power}
\end{figure}

\subsection{Functional Time Series}

We next consider functional time series with observations in \(\mathcal X=L_2[0,1]\), equipped with \( \Vert x\Vert _{\mathcal X}= \{\int_0^1 |x(\tau)|^2\,{\rm d}\tau \}^{1/2} \).
This illustrates the use of AutoHSIC beyond finite-dimensional Euclidean settings.

Let $\mathcal{B}=\{B(\tau):\tau\in[0,1]\}$ denote standard Brownian motion, and let \( \eta_t(\tau) \stackrel{i.i.d.}{\sim} \mathcal{B} \). Following the simulation design of \cite{jiang2024testing}, we generate the functional time series $\{X_t(\tau):\tau\in[0,1]\}_{t=1}^T$ from the following data-generating processes.

\noindent DGP~6: $X_t(\tau)=\eta_t(\tau)$.

\noindent DGP~7:
$ X_t(\tau)=\sigma_t(\tau)\eta_t(\tau),
\quad \sigma_t^2(\tau) = \tau + \int_0^1 0.6\exp\{(\tau^2+\tau_1^2)/2\} X_{t-1}^2(\tau_1)\,{\rm d}\tau_1$.

\noindent DGP~8: $X_t(\tau)=\eta_t(\tau)\eta_{t-1}(\tau)$.

DGP~6 is an i.i.d. Brownian motion benchmark for size evaluation. DGP~7 is a functional autoregressive conditional heteroscedasticity process with conditional second-order serial dependence. DGP~8 is a functional nonlinear moving-average process with nonlinear lag dependence.

Each functional observation is first generated on a grid of 1000 equally spaced
points on $[0,1]$. The discretised curves are then smoothed using cubic B-splines,
with 20 basis functions, before applying the tests.

Table~\ref{table_functional_0.05} reports empirical rejection rates for DGPs~6--8 at the \(5\%\) level. The tests have reasonable size under the i.i.d. Brownian motion benchmark (DGP~6). 
Under DGPs~7 and 8, the bounded-kernel AutoHSIC tests outperform the ADCV-based test, indicating good power against functional dependence arising from volatility or multiplicative lag structures.

\begin{table*}[h]
\caption{Empirical rejection rates (\%) for DGPs~6--8 at the $5\%$ nominal level.} 
\label{table_functional_0.05}
\footnotesize
\tabcolsep=0pt
\begin{tabular*}{\textwidth}{@{\extracolsep{\fill}}ll rrrr r rrrr r rrrr@{\extracolsep{\fill}}}
\toprule%
& & \multicolumn{4}{c}{LK} && \multicolumn{4}{c}{GK} && \multicolumn{4}{c}{BDK} \\

\cline{3-6} \cline{8-11} \cline{12-16}

&$T$ & $V_{T,1}$ & $V_{T,3}$ & $P_{T,3}$ & $P_{T,6}$ && $V_{T,1}$ & $V_{T,3}$ & $P_{T,3}$ & $P_{T,6}$ && $V_{T,1}$ & $V_{T,3}$ & $P_{T,3}$ & $P_{T,6}$ \\

\cline{3-6} \cline{8-11} \cline{12-16}

DGP6 &100 &5.0 &4.9 &5.7 &6.0 && 5.0 & 4.2 & 6.3 &5.3  &&6.5 &4.7 &5.9 &5.0   \\
&200 &6.1 &4.2 & 5.1 & 5.6 & &5.5&4.1&4.9&5.8 &&6.1 &4.5 &4.7 &5.9   \\

\vspace{-0.3cm}\\

DGP7 &100 &45.3 &5.8 &33.8 &27.2 &&35.9&5.5&29.3&23.3&&10.1 &5.4 &9.9 &9.4   \\
&200 &84.3 &5.4 &65.6 &53.3 & &73.2&5.9&57.5 &46.1&&20.5 &4.4 &18.1 &15.8   \\

\vspace{-0.3cm}\\

DGP8 &100 &96.6 &9.8 &93.2 &84.8 &&85.1&7.6& 73.9 & 62.5 &&20.0 &5.6 &17.3 &17.8   \\
&200 &99.9 &10.8 &100.0 &99.7 &&99.7&7.7 &99.1&95.5 &&63.5 &5.2 &48.1 &39.3  \\

\midrule
\end{tabular*}
\end{table*}

\subsection{Matrix Time Series}\label{subsec_matrix_ts}

Finally, we consider matrix-valued observations \(X_t\in\mathbb R^{d\times d}\), with \(\mathcal X=\mathbb R^{d\times d}\) equipped with the Frobenius norm. 
The data are generated from the matrix generalised autoregressive conditional heteroscedasticity
(matrix GARCH) model of \citet{yu2025matrix},
\( X_t=U_t^{1/2}\eta_t V_t^{1/2} \),
where
{\small
\begin{align}
U_t &=\frac{S_{1t}}{\operatorname{tr}(S_{1t})}y_t,
\quad
V_t=\frac{S_{2t}}{\operatorname{tr}(S_{2t})}, \quad
y_t = \omega+\alpha\operatorname{tr}(X_{t-1}X_{t-1}^\top)+\beta y_{t-1},
\nonumber\\
S_{1t} &= A_0A_0^\top+A_1X_{t-1}X_{t-1}^\top A_1^\top
       +A_2S_{1,t-1}A_2^\top, \nonumber\\
S_{2t} &= B_0B_0^\top+B_1X_{t-1}^\top X_{t-1}B_1^\top
       +B_2S_{2,t-1}B_2^\top. \nonumber
\end{align}}
We set $\omega=0.4$, $\alpha=c$, $\beta=0.6$, $A_1=B_1=cI_d$, $A_2=B_2=0.6I_d$.  
The matrices $A_0$ and $B_0$ are equal to the lower triangular matrix with first diagonal element one and all other non-zero lower-triangular entries equal to $0.4$.
The innovations $\{\eta_t\}$ are i.i.d. random matrices with independent standard
normal entries.

Table~\ref{table_matrix_0.05} reports empirical rejection rates for matrix time series at the \(5\%\) level with \(T=200\). The case \(c=0\) assesses size, while \(c=0.2\) and \(c=0.3\) assess power. 
All tests have rejection rates close to \(5\%\) when \(c=0\). 
For \(c>0\), the process is linearly uncorrelated but exhibits nonlinear serial dependence through its conditional row and column covariance matrices.
In this setting, the AutoHSIC tests with bounded kernels show higher empirical power than the ADCV-based tests. The difference is particularly visible for the smaller signal $c=0.2$, suggesting that bounded kernels can be more sensitive to weak nonlinear dependence in matrix time series.

\begin{table*}[t]
\caption{Empirical sizes and powers ($\%$) for matrix time series when $T=200$.}
\label{table_matrix_0.05}
\footnotesize
\tabcolsep=0pt
\begin{tabular*}{\textwidth}{@{\extracolsep{\fill}}ll rrrr r rrrr r rrrr@{\extracolsep{\fill}}}

\toprule
         &  & \multicolumn{4}{c}{LK} && \multicolumn{4}{c}{GK} && \multicolumn{4}{c}{BDK} \\

         \cline{3-6} \cline{8-11} \cline{13-16}

         $c$ & $d$ & $V_{T,1}$ & $V_{T,3}$ & $P_{T,3}$ & $P_{T,6}$
         && $V_{T,1}$ & $V_{T,3}$ & $P_{T,3}$ & $P_{T,6}$ && $V_{T,1}$ & $V_{T,3}$ & $P_{T,3}$ & $P_{T,6}$ \\

        \cline{3-6} \cline{8-11} \cline{13-16} 

         0 & 2 & 4.2 &	5.4 &	4.7 &	5.6 &&	4.2 &	5.3 &	4.6 &	5.9 &&	4.3 &	4.9 &	5.2 &	5.4 \\
         & 5 & 4.8 &	4.6 &	5.3 &	5.9 &&	5.3 &	4.9 &	5.7 &	6.1 &&	5.0 &	5.3 &	5.9 &	5.9 \\
         & 8 & 3.8 &	6.1 &	4.0 &	4.9 &&	4.4 &	5.6 &	3.7 &	4.2 &&	4.9 &	5.1 &	3.5 &	4.0 \\
         
         \vspace{-0.3cm}\\

         0.2 & 2 & 18.3 &	9.3 &	23.1 &	22.7 &&	16.4 &	8.0 &	20.4 &	19.7 &&	5.8 &	5.2 &	7.8& 	8.5 \\
         & 5 & 20.4 &	9.6 &	27.9 &	26.6 &&	14.1 &	7.5 &	19.6 &	17.9 &&	7.2 &	5.5 &	8.9 &	8.7 \\
         & 8 & 40.0 &	19.4 &	51.2 &	51.1 &&	31.8 &	16.4 &	42.7 &	41.1 &&	10.7 &	7.2 &	14.4 &	16.2 \\
         \vspace{-0.3cm}\\

         0.3 & 2 & 55.5 &	34.5 &	66.1 &	66.5 &&	49.4 &	29.1 &	61.1 &	61.7 &&	14.4 &	11.7 &	25.4& 	27.9 \\
         & 5 & 63.6 &	43.3 &	75.5 &	76.2 &&	44.6 &	29.5 &	57.6 &	58.3 &&	14.5 &	12.6 &	25.2 &	27.3 \\
         & 8 & 88.6 	&68.8 &	94.1 &	93.2 &&	74.8 &	55.2 &	84.9 &	84.6 &&	30.5 &	20.9 &	48.6 &	51.9 \\

\midrule
\end{tabular*}
\begin{tablenotes}
\footnotesize
\item Note: Results for $T=100$ are reported in Table~11 of the supplementary material.
\end{tablenotes}
\end{table*}

\section{An Extension to Model Diagnostic Checking}\label{sec_modelcheck}

We next consider model diagnostic checking for multivariate time series. The aim is to assess whether a fitted parametric model has adequately removed serial dependence from the data.
In this setting, AutoHSIC is applied to residuals rather than observations, so parameter estimation may contribute non-negligibly to the limiting null distribution of test statistic.

Let $\mathcal{X}=\mathbb{R}^d$. Following \cite{wang2021new} and
\cite{wan2022goodness}, consider a causal parametric model
\begin{equation}\label{model_f}
X_t=f(I_{t-1},\theta_0,\eta_t),
\end{equation}
where $I_t=(X_t^\top,X_{t-1}^\top,\ldots)^\top$ is the information set, $\theta_0\in\mathbb{R}^p$ is the parameter vector, and $\{\eta_t\in\mathbb{R}^d\}$ is an i.i.d. innovation sequence independent of $\mathcal{F}_{t-1}$, with $\mathcal{F}_t=\sigma(I_t)$. The function $f:\mathbb{R}^{\infty}\times\mathbb{R}^p\times\mathbb{R}^d\to\mathbb{R}^d$ is assumed to be known. This formulation covers a broad class of causal nonlinear time series models, including multivariate volatility models such as BEKK, dynamic conditional correlation \citep{engle1995multivariate,tse2002multivariate}, as well as matrix GARCH models \citep{yu2025matrix}.

Suppose that the model admits the invertible representation
\( \eta_t=g(X_t,I_{t-1},\theta_0) \),
where $g:\mathbb{R}^d\times\mathbb{R}^{\infty}\times\mathbb{R}^p\to\mathbb{R}^d$
is measurable. Let
\( \widehat{\eta}_t = g(X_t,\widehat I_{t-1},\widehat\theta_T) \)
denote the residual, where $\widehat\theta_T$ is an estimator of $\theta_0$ and $\widehat I_{t-1}$ is the observed information set up to time $t-1$, including the initial values.

The diagnostic null hypothesis is
\begin{equation}\nonumber
H_0:\ \{\eta_t\}\ \text{is i.i.d.},
\quad
H_{a,M}:\ V_m^\eta>0\ \text{for some }m\in\{1,\ldots,M\},
\end{equation}
where \( V_m^\eta=\operatorname{HSIC}(\eta_t,\eta_{t-m}) \).
Under correct specification, the innovations should be serially independent.
Since \(\{\eta_t\}\) is unobserved, the tests are constructed from the residuals \(\{\widehat{\eta}_t\}\). Let \(\widehat V_{T,m}\) denote the residual-based AutoHSIC statistic obtained from \eqref{V_T(m)} after replacing \(X_t\) by \(\widehat{\eta}_t\).
Similarly, define the portmanteau statistic
\[
\widehat P_{T,M} =
\sum_{m=1}^M \widehat V_{T,m}.
\]
These statistics target remaining pairwise serial dependence in the innovation sequence up to lag $M$.


\subsection{Asymptotics of $\widehat V_{T,m}$ and $\widehat P_{T,M}$}

We first study the effect of parameter estimation on the null distribution of the residual-based statistics. 
The following assumptions control the smoothness of the innovation map, the asymptotic linearity of the estimator, the effect of initial values, and the regularity of the kernels.

\begin{assumption}\label{ass5.2_g}
Let $g_t(\theta)=g(X_t,I_{t-1},\theta)$. For any
$i,j,q\in\{1,\ldots,p\}$,
{\small
\[
\mathbb{E}\left\{\sup_{\theta\in\Theta}
\Vert \frac{\partial g_t(\theta)}{\partial\theta_i}\Vert\right\}^2<\infty,
\quad
\mathbb{E}\left\{\sup_{\theta\in\Theta}
\Vert \frac{\partial^2 g_t(\theta)}
{\partial\theta_i\partial\theta_j}\Vert \right\}^2<\infty,
\quad
\mathbb{E}\left\{\sup_{\theta\in\Theta}
\Vert \frac{\partial^3 g_t(\theta)}
{\partial\theta_i\partial\theta_j\partial\theta_q}\Vert \right\}^2<\infty,
\]}
where $\Theta\subset\mathbb{R}^p$ is compact.
\end{assumption}

\begin{assumption}\label{ass5.3_esti}
The estimator $\widehat\theta_T$ satisfies the asymptotic linear expansion
\begin{equation}\nonumber
\sqrt{T}(\widehat\theta_T-\theta_0)
=
\frac{1}{\sqrt{T}}\sum_{t=1}^T \pi(X_t,I_{t-1},\theta_0)+o_p(1)
=:
\frac{1}{\sqrt{T}}\sum_{t=1}^T \pi_t+o_p(1),
\end{equation}
where $\pi:\mathbb{R}^d\times\mathbb{R}^{\infty}\times\mathbb{R}^p
\to\mathbb{R}^p$ is measurable,
$\mathbb{E}(\pi_t\mid\mathcal{F}_{t-1})=0$,
$\mathbb{E}\Vert \pi_t \Vert^2<\infty$, and
$T^{-1/2}\sum_{t=1}^T\pi_t\xrightarrow{\mathrm d}\mathcal W$.
\end{assumption}

\begin{assumption}\label{ass5.4_trucation}
Let
\(
\widehat R_t(\theta)=\widehat g_t(\theta)-g_t(\theta) \)
and
\( \widehat g_t(\theta)=g(X_t,\widehat I_{t-1},\theta) \).
Then
\(
\sum_{t=1}^T \sup_{\theta\in\Theta}\Vert \widehat R_t(\theta)\Vert ^3=O_p(1).
\)
\end{assumption}

\begin{assumption}\label{ass5.5_kernel}
The kernels $k$ and $l$, together with their partial derivatives up to second order, are uniformly bounded and Lipschitz continuous. That is, for each
\[
p\in\{k,k_x,k_y,k_{xx},k_{xy},k_{yy},
l,l_x,l_y,l_{xx},l_{xy},l_{yy}\},
\]
there exists a constant $C<\infty$ such that
\(\sup_{x,y}\Vert p(x,y)\Vert \leq C \), 
and \( \Vert p(x_1,y_1)-p(x_2,y_2)\Vert  \leq C\Vert (x_1,y_1)-(x_2,y_2)\Vert 
\).
\end{assumption}

Assumption~\ref{ass5.2_g} is a smoothness and moment condition on the innovation map. 
Assumption~\ref{ass5.3_esti} is the usual asymptotic linearity condition for the estimator and is satisfied by many standard estimators, including maximum likelihood, quasi-maximum likelihood, least squares and nonlinear least squares estimators under regularity conditions. 
Assumption~\ref{ass5.4_trucation} controls the effect of replacing the infinite information set by its observed, initialised version. 
Assumption~\ref{ass5.5_kernel} ensures that Taylor expansions of the residual-based kernel terms are valid. It is satisfied by commonly used smooth bounded kernels such as the Gaussian kernel. 

The following theorem gives the residual-based null distribution. Compared with Theorem~\ref{thm2_nulldis}, the limit contains an additional term due to parameter estimation. 

\begin{theorem}\label{thm5.1_nulldis_esti}
Under $H_0$, suppose that Assumptions~\ref{ass5.1_stationary} and \ref{ass5.2_g}--\ref{ass5.5_kernel} hold. Then, for any fixed maximum lag $M\in\mathbb{Z}^+$,
\begin{equation}\nonumber
T\left(\widehat V_{T,1},\widehat V_{T,2},\ldots,\widehat V_{T,M}\right)^\top
\xrightarrow{\mathrm d}
\left(\xi^{\rm res}_1,\xi^{\rm res}_2,\ldots,\xi^{\rm res}_M\right)^\top,
\qquad T\to\infty,
\end{equation}
where
\(
\xi^{\rm res}_m =
\xi_m+\mathcal W^\top\Lambda^{(23)}\mathcal W, \quad m=1,\ldots,M.
\)
Here $\xi_m$ is defined as in Theorem~\ref{thm2_nulldis}, with $X_t$ replaced by $\eta_t$, $\mathcal W$ is defined in Assumption~\ref{ass5.3_esti}, and $\Lambda^{(23)}$ is a $p\times p$ constant matrix determined by the parameter estimation effect. Its explicit form is given in the supplementary material.
\end{theorem}

Theorem~\ref{thm5.1_nulldis_esti} shows that the residual-based statistic has a different null distribution from the infeasible statistic based on the true innovations. The additional quadratic term reflects the interaction between the estimation error in $\widehat\theta_T$ and the non-degenerate kernel structure of AutoHSIC. Even in the simple AR model $X_t=\theta X_{t-1}+\eta_t$, this effect is generally non-zero; for example, the corresponding \( \Lambda^{(23)}\) reduces to a term,
\( (1-\theta)^2 \mathbb{E}\{\eta_1 k_x(\eta_2,\eta_1)\} \mathbb{E}\{\eta_1 l_x(\eta_2,\eta_1)\} \).
Thus the residual-based limiting null distribution contains nuisance quantities from both the kernel operator and parameter estimation. This motivates the residual bootstrap introduced below.

We next consider fixed alternatives. The following result shows that the diagnostic statistics diverge when the fitted model leaves serial dependence in the innovation sequence.

\begin{theorem}\label{thm5.2_power}
Under $H_{a,M}$, suppose that Assumptions~\ref{ass5.1_stationary}, \ref{assumption1_beta} and \ref{ass5.2_g}--\ref{ass5.5_kernel} hold. Then, for any lag $m$ such that the corresponding residual AutoHSIC population quantity is positive,
\(
T\widehat V_{T,m}\xrightarrow{\mathrm p}\infty
\), as \(T \to \infty\).
Consequently, 
\(
T\widehat P_{T,M}\xrightarrow{\mathrm p}\infty
\), as \(T \to \infty\).
\end{theorem}

\subsection{Residual Bootstrap}\label{subsec_boot}

The null limit in Theorem~\ref{thm5.1_nulldis_esti} is non-pivotal, depending on the estimation effect as well as unknown features of the innovation distribution and kernels. 
We therefore use a residual bootstrap to approximate critical values.
The bootstrap refits the model in each replication, so that the variability due to estimating $\theta_0$ is reflected in the bootstrap residuals.
Such methods are widely used in time series model checking \citep{politis2003impact, wang2021new}. The procedure is summarised in Algorithm~\ref{algo2_res}.

\begin{algorithm}[!t] \caption{Residual bootstrap.}\label{algo2_res} 
\small
\begin{algorithmic}

\STATE {\bf Step~1}: Given the maximum lag $M$, the original sample $\{X_t\}_{t=1}^T$, and the number of bootstrap replications $B$. 

\STATE {\bf Step~2}: Estimate model \eqref{model_f}, obtain the estimator $\widehat{\theta}_T$ and residuals $\{\widehat{\eta}_t\}_{t=1}^T$ based on $\{X_t\}_{t=1}^T$.

\STATE {\bf Step~3}: 
\FOR{$b=1,\ldots,B$}
\STATE \parbox[t]{0.95\linewidth}{%
{\bf Step~3.1}: Generate bootstrap innovations $\{\widehat{\eta}_t^{*(b)}\}_{t=1}^T$ (after standardisation) by resampling with replacement from $\{\widehat{\eta}_t\}_{t=1}^T$. Then generate bootstrap data $\{X_t^{*(b)}\}_{t=1}^T$ according to \eqref{model_f}, based on $\widehat{\theta}_T$ and $\{\widehat{\eta}_t^{*(b)}\}_{t=1}^T$. }

\STATE \parbox[t]{0.95\linewidth}{%
{\bf Step~3.2}: Compute $\widehat{\theta}_T^{*(b)}$ based on $\{X_t^{*(b)}\}_{t=1}^T$, and then calculate the bootstrap residuals $\{ \widehat{\eta}_t^{**(b)}\}_{t=1}^T$ with $\widehat{\eta}_t^{* *(b)} =g (X_t^{*(b)}, \widehat{I}_{t-1}^{*(b)}, \widehat{\theta}_T^{*(b)} )$, where $\widehat{I}_{t-1}^{*(b)}$ is the bootstrap information set up to time $t-1$. }

\STATE \parbox[t]{0.95\linewidth}{%
{\bf Step~3.3}: Calculate the bootstrap test statistics $\widehat{V}_{T,m}^{*(b)}$ in the same way as $\widehat{V}_{T,m}$ and $\widehat{P}_{T,M}^{*(b)} =\sum_{m=1}^M \widehat{V}_{T,m}^{*(b)}$, where $\widehat{\eta}_t^{* *(b)}$ replaces $\widehat{\eta}_t$. }
\ENDFOR

\STATE {\bf Step~4}:
Let $\widehat{V}_{T,m;1-\alpha}^*$ and $\widehat{P}_{T,M;1-\alpha}^*$ denote the empirical 
$1-\alpha$ quantiles of
$\{\widehat{V}_{T,m}^{*(b)}:b=1,\ldots,B\}$ and
$\{\widehat{P}_{T,M}^{*(b)}:b=1,\ldots,B\}$, respectively.
Reject $H_0$ in the single-lag test at lag $m$ if
$T\widehat{V}_{T,m}> T\widehat{V}_{T,m;1-\alpha}^*$, and reject $H_0$ in the portmanteau test if
$T\widehat{P}_{T,M}> T\widehat{P}_{T,M;1-\alpha}^*$.
\end{algorithmic} \end{algorithm}

To state the bootstrap validity result, let $\mathbb{E}^*$ denote conditional expectation given $\{X_t\}_{t=1}^T$. We impose the following conditions similar to Assumption~A7 of \cite{escanciano2006}, which require the bootstrap to reproduce the covariance structure of the estimator and its interaction with the residual AutoHSIC kernel.

\begin{assumption}\label{ass_boot1} 
    The bootstrap estimator $\widehat{\theta}_T^*$ satisfies that
    \begin{equation}\nonumber
\sqrt{T} (\widehat{\theta}_T^*-\widehat{\theta}_T ) 
 =\frac{1}{\sqrt{T}} \sum_{t=1}^T \pi (X_t^*, \widehat{I}_{t-1}^*, \widehat{\theta}_T )+o_p^*(1) =: \frac{1}{\sqrt{T}} \sum_{t=1}^T \pi_t^*+o_p^*(1),
\end{equation}
where $\pi$ is defined in Assumption~7 and $\mathbb{E}^*(\pi_t^* \mid \widehat{I}_{t-1}^*)=0$.
\end{assumption}

\begin{assumption}\label{ass_boot2}
The following convergence results hold:

\noindent $(a)$ \( \frac{1}{T}\sum_{t=1}^{T}
\mathbb{E}^*(\pi_t^*\pi_t^{*\top}) \xrightarrow{\mathrm p}
\mathbb{E}(\pi_t\pi_t^\top) \);

\noindent $(b)$ \( \frac{1}{T-m}\sum_{t=m+1}^{T}
\mathbb{E}^*\{\Phi_\ell^*(\widehat\eta_{t,m}^*)\pi_t^*\} \xrightarrow{\mathrm p} \mathbb{E}\{\Phi_\ell(\eta_{t,m})\pi_t\} \),\\
for each \(\ell \geq 1\) and \( m= 1,\ldots, M\).
Here
\( \eta_{t,m} = (\eta_t,\eta_{t-m}) \), \( \widehat\eta_{t,m}^* = (\widehat\eta_t^*,\widehat\eta_{t-m}^*)\),
and $\{\Phi_\ell^*(\cdot)\}_{\ell=1}^{\infty}$ are the orthonormal eigenfunctions of
\( h_2^{(0*)}(z_1,z_2) = \mathbb{E}^* \{h(z_1,z_2,\widehat\eta_{t,m}^*,\widehat\eta_{t,m}^{*\prime})\},
\)
where $\widehat\eta_{t,m}^{*\prime}$ is an independent copy of $\widehat\eta_{t,m}^*$.
\end{assumption}

The following theorem establishes bootstrap validity. Under the null, the residual bootstrap reproduces the limiting distribution of the residual-based statistics, including the contribution from parameter estimation. Under fixed alternatives, the bootstrap statistics remain bounded at the null scale, whereas the original residual-based statistics diverge.

\begin{theorem}\label{thm5.3_boot}
Suppose that Assumptions~\ref{ass5.1_stationary} and \ref{ass5.2_g}--\ref{ass_boot2} hold. Then the following statements hold.

\noindent $(a)$ Under $H_0$,
\(
T\left(\widehat V_{T,1}^*,\ldots,\widehat V_{T,M}^*\right)^\top
\xrightarrow{\mathrm d^*}
\left(\xi^{\rm res}_1,\ldots,\xi^{\rm res}_M\right)^\top
\) in probability, as $T \to \infty$,
where $\left(\xi^{\rm res}_1,\ldots,\xi^{\rm res}_M\right)^\top$ is defined in
Theorem~\ref{thm5.1_nulldis_esti}.

\noindent $(b)$ Under $H_{a,M}$ and Assumption~\ref{assumption1_beta},
\(
T\widehat V_{T,m}^*=O_p^*(1)
\) in probability
for each fixed $m\in\{1,\ldots,M\}$, and
\(
T\widehat P_{T,M}^*=O_p^*(1)
\) in probability.
\end{theorem}

Theorem~\ref{thm5.3_boot} justifies the use of the residual bootstrap for diagnostic testing. It provides data-driven critical values for the non-pivotal residual-based null distribution. Together with Theorem~\ref{thm5.2_power}, it also shows that the bootstrap tests have asymptotic power equal to one against fixed alternatives.

\section{Applications}\label{sec_realdata}

We illustrate the proposed tests in two empirical settings. The first application concerns functional time series of cumulative intraday returns. The second application considers a $3\times 3$ matrix-valued time series of credit default swap returns. In both cases, empirical $p$-values are computed using 1000 bootstrap replications.

\subsection{Cumulative Intraday Returns}

Functional data methods provide a natural framework for analysing intraday financial returns over a trading day; see, for example, \cite{Hormann2013farch} and \cite{cerovecki2019functional}. We use this application to illustrate two roles of AutoHSIC: testing for serial independence before modelling, and checking whether a fitted functional volatility model has captured the serial dependence adequately.

Following \cite{gabrys2010tests}, the cumulative intraday return (CIDR) curve on day \(t\) is defined as
\(
X_t(s)=100\log\{P_t(s)/P_t(0)\},\ s\in[0,1] \),
where \(P_t(s)\) is the stock price at rescaled trading time \(s\) and \(P_t(0)\) is the opening price. 
We analyse CIDR curves for the China Securities Index 300 (CSI 300) and Meituan in 2024, using one-minute frequency data from Wind (\url{https://www.wind.com.cn/}). 
The CSI 300 sample contains \(T=241\) trading
days, with 242 observations per day over 9:30--11:30 and 13:00--15:00. 
The Meituan sample contains \(T=245\) trading days, with 332 observations per day over 9:30--11:30 and 13:00--16:00 Hong Kong time. 
The two series provide complementary examples of a broad market index and an individual stock.
The discrete observations are smoothed using 20 cubic B-spline basis functions implemented in the \textsf{R} package \texttt{fda}.

We first test serial independence of the CIDR curves. 
We apply AutoHSIC with the Gaussian kernel (GK) and the Brownian distance kernel (BDK), the latter corresponding to ADCV. For comparison, we also report the Cramér--von Mises test based on ADCV of \cite{jiang2024testing}, denoted by \(\operatorname{CvM}_{DC}\), 
and the martingale difference hypothesis test based on martingale difference divergence of \cite{hong2026mdh}, denoted by MDH.

Table~\ref{table_pvalue_cidr_hsi} reports the empirical \(p\)-values. 
At the \(5\%\) level, the MDH test does not reject for either series, whereas the GK-based AutoHSIC tests provide evidence of serial dependence in the CIDR curves.
This suggests dependence beyond the conditional mean.
The ADCV-based procedures, including the BDK version of AutoHSIC and \(\operatorname{CvM}_{DC}\),
yield larger \(p\)-values in this example.

Motivated by this evidence, we fit the univariate functional ARCH(1), or FARCH(1), model of \cite{Hormann2013farch} separately to the CSI 300 and Meituan CIDR curves:
\begin{equation}\nonumber
X_t=\eta_t\sigma_t,
\qquad
\sigma_t^2=\delta+\beta(X_{t-1}^2),
\end{equation}
where $\{\eta_t\}$ is a sequence of i.i.d. random functions in $\mathcal{X}=L_2[0,1]$, $\delta\in\mathcal{X}^+$, and $\beta:\mathcal{X}^+\to\mathcal{X}^+$ is an operator. Here $\mathcal{X}^+$ denotes the set of nonnegative functions in $\mathcal{X}$. The model is estimated using the functional Yule--Walker approach, yielding residuals \(\{ \widehat\eta_t=X_t/\widehat\sigma_t \}\).

We then apply the proposed tests to the residual curves.
Table~\ref{table_pvalue_farch_hsi} reports the corresponding empirical $p$-values.
For Meituan, the relatively large $p$-values provide no evidence of remaining serial dependence, suggesting that the FARCH(1) model captures the main serial dependence in the CIDR curves. 
For the CSI 300, however, some residual dependence remains. This indicates that a first-order functional volatility specification may be insufficient for the index series, and that more flexible dynamics or additional higher-order dependence structures may be needed.

\begin{table*}
\caption{Empirical $p$-values for serial independence tests for CIDR curves in 2024. }
\label{table_pvalue_cidr_hsi}
\footnotesize
\tabcolsep=0pt
\begin{tabular*}{\textwidth}{@{\extracolsep{\fill}}l cccc c cccc cc c cc c@{\extracolsep{\fill}}}
\toprule%
&  \multicolumn{4}{c}{GK} && \multicolumn{4}{c}{BDK} \\
\cline{2-5} \cline{7-10} 
 & $V_{T,1}$ & $V_{T,2}$ & $P_{T,2}$ & $P_{T,4}$ 
&& $V_{T,1}$ & $V_{T,2}$ & $P_{T,2}$ & $P_{T,4}$ 
&&& $\operatorname{CvM}_{DC}$ &&& MDH  \\
\cline{2-5} \cline{7-10} \cline{12-13} \cline{15-16} 

CSI 300 & 0.018 &0.010 & 0.001 & 0.002 &&0.021  &0.011  &0.001 &0.001 &&& 0.433 &&&0.539 \\

Meituan & 0.039  &0.685 &0.112 &0.323 && 0.082  &0.733  &0.230 &0.353 &&& 0.090 &&& 0.206   \\
\midrule

\end{tabular*}
\end{table*}

\begin{table*}
\caption{Empirical $p$-values for residual diagnostic checks of the FARCH(1) model.}
\label{table_pvalue_farch_hsi}
\footnotesize
\tabcolsep=0pt
\begin{tabular*}{\textwidth}{@{\extracolsep{\fill}}l cccc c cccc@{\extracolsep{\fill}}}
\toprule%
&  \multicolumn{4}{c}{GK} && \multicolumn{4}{c}{BDK} \\
\cline{2-5} \cline{7-10} \\[-1em]
 & $\widehat{V}_{T,1}$ & $\widehat{V}_{T,2}$ & $\widehat{P}_{T,2}$ & $\widehat{P}_{T,4}$ 
&& $\widehat{V}_{T,1}$ & $\widehat{V}_{T,2}$ & $\widehat{P}_{T,2}$ & $\widehat{P}_{T,4}$ \\
\cline{2-5} \cline{7-10} 

CSI 300 &0.006 &0.140 &0.004 &0.023 &&0.044 &0.120 &0.046 &0.114   \\

Meituan &0.057 &0.445 &0.211 &0.520 &&0.139 &0.430 &0.259 &0.296 \\
\midrule

\end{tabular*}
\end{table*}

\subsection{Credit Default Swap Returns}

Recently, \cite{yu2025matrix} introduced a matrix GARCH model for modelling conditional heteroskedasticity in matrix-valued time series. In the empirical analysis, they studied a $3\times 3$ matrix-valued series of credit default swap (CDS) returns and used the proposed matrix GARCH model to characterise the conditional covariance structures across tenors and financial institutions.

Our analysis is complementary to theirs. 
For the CDS return matrices, the entrywise sample autocorrelations are weak, suggesting that a dynamic conditional mean model is not the main modelling issue. However, the absence of linear autocorrelation does not imply serial independence.
The proposed AutoHSIC-based tests provide a direct tool for this purpose. Unlike entrywise autocorrelation checks, our tests examine serial independence of the matrix-valued observations themselves and are therefore able to detect nonlinear or higher-order temporal dependence. 
Applied to the CDS return matrices, the tests reject serial independence at the \(5\%\) level, indicating dynamic dependence despite weak autocorrelations. 
Together with the matrix-valued volatility patterns documented in \cite{yu2025matrix}, this provides a natural motivation for fitting a matrix GARCH-type model.

The dataset\footnote{The data are available at \url{https://www.tandfonline.com/doi/full/10.1080/01621459.2024.2415719\#supplemental-material-section}.} consists of a $3 \times 3$ matrix time series of daily log-returns of CDS for three financial institutions: Barclays Bank, Deutsche Bank, and Bank of America (columns), and three tenors of 3, 5, and 7 years (rows). The sample used in \cite{yu2025matrix} spans February 2015 to October 2018, yielding $T=962$ observations.

The empirical $p$-values for testing serial independence of the CDS return series are below 0.05, providing evidence of serial dependence beyond linear autocorrelation. This finding motivates the use of matrix GARCH-type models, as in \cite{yu2025matrix}.

We then fit the matrix GARCH model in Section~\ref{subsec_matrix_ts} by the quasi-maximum likelihood method of \cite{yu2025matrix}, and apply the proposed tests to the residual matrices. 
The corresponding $p$-values are reported in Table~\ref{table_pvalue_matrix_residual}.
Here, \(\widehat V_{T,m}\) and \(\widehat P_{T,M}\)
are residual AutoHSIC tests, while \(Q_{T,M}\) is the autocorrelation-based portmanteau test applied to the squared Frobenius norms of residual matrices up to lag \(M\), as in Section~4 of \cite{yu2025matrix}.

The results highlight the distinction between checking for remaining second-order dependence and independence in the residuals. 
The portmanteau statistics \(Q_{T,2}\) and \(Q_{T,4}\) do not reject the null at the \(5\%\) level, suggesting that the fitted matrix GARCH model removes autocorrelation in the squared Frobenius-norm residuals. 
By contrast, the small \(p\)-values of the AutoHSIC-based tests indicate remaining nonlinear or higher-order temporal dependence in the standardised residual matrices. 
In this sense, the proposed method complements the matrix GARCH framework by detecting forms of residual temporal dependence that may not be visible through autocorrelation-based portmanteau diagnostics.

\begin{table*}
\caption{Empirical $p$-values for residual diagnostic checks of the matrix GARCH.}
\label{table_pvalue_matrix_residual}
\footnotesize
\tabcolsep=0pt
\begin{tabular*}{\textwidth}{@{\extracolsep{\fill}}l cccc c cccc c cc@{\extracolsep{\fill}}}
\toprule%
&  \multicolumn{4}{c}{GK} && \multicolumn{4}{c}{BDK} \\
\cline{2-5} \cline{7-10} \\[-1em]
 & $\widehat{V}_{T,1}$ & $\widehat{V}_{T,2}$ & $\widehat{P}_{T,2}$ & $\widehat{P}_{T,4}$ 
&& $\widehat{V}_{T,1}$ & $\widehat{V}_{T,2}$ & $\widehat{P}_{T,2}$ & $\widehat{P}_{T,4}$
&& $Q_{T,2}$ & $Q_{T,4}$ \\
\cline{2-5} \cline{7-10} \cline{12-13} 
$p$-value & 0.000 & 0.008 & 0.002 & 0.003 && 0.003 & 0.012 & 0.003 & 0.006 && 0.253 & 0.281 \\
[-0.2em]
\midrule

\end{tabular*}
\end{table*}

\section{Conclusion}\label{sec_conclusion}

This paper develops an HSIC-based framework for testing serial independence in a single stationary time series. The proposed AutoHSIC statistic measures lag-specific dependence and can detect nonlinear serial dependence beyond autocorrelation.
The main inferential challenge is that the statistic is constructed from overlapping
lagged pairs and therefore has a U-statistic structure under temporal dependence, with degeneracy playing a central role under the null.

We establish the asymptotic behaviour of the single-lag and portmanteau statistics under the null and alternatives, and develop bootstrap procedures for the non-pivotal null distributions. 
We also extend the framework to residual-based diagnostics and show how parameter estimation affects the null limit.
Simulations and applications show that AutoHSIC can complement classical autocorrelation-based diagnostics by detecting nonlinear forms of remaining serial dependence.

Several questions remain for future work. 
First, the fixed-lag portmanteau statistic may have limited sensitivity to dependence beyond the chosen lag range; frequency-domain or spectral approaches may offer a broader alternative \citep{hong1999hypothesis, jiang2024testing}. 
Second, our treatment of parameter estimation in residual diagnostics is developed for Euclidean-valued residuals.
Extending it to functional, matrix-valued or other non-Euclidean time series would require additional theory for estimated objects in structured spaces.
Third, the RKHS embedding perspective suggests further possibilities for distributional inference in dependent data; recent developments in RKHS-based inference for time series provide one possible
direction \citep{Ghoshal2026resampling}.

\newpage
\appendix
\section*{Supplementary Material} 

The supplementary material is organised as follows.
Section~\ref{sm_simulation} presents additional simulation results.
Section~\ref{sm_application} provides additional empirical results, including an application to a 30-dimensional realised variance time series and supplementary results for CDS
returns.
Section~\ref{sm_proof} contains the proofs of the theoretical results.

\section{Additional Simulation Results}\label{sm_simulation}

\subsection{Additional Simulation Results in Section~5}

Tables~\ref{additional_table_dgp1}, \ref{additional_table_dgp2-3} and \ref{additional_table_dgp4-5} report the full results for DGPs~1--5 in Section~5.1, based on the Laplacian kernel (LK), the Gaussian kernel (GK), and the Brownian distance kernel (BDK).

Additionally, we examine empirical power under two other DGPs at the \(5\%\) significance level: a multivariate GARCH process in Table~\ref{additional_table_garch_0.05} and a vector autoregressive (VAR) process in Table~\ref{additional_table_dgp_var}. 
The former further confirms that the AutoHSIC-based tests have high power against nonlinear serial dependence, whereas the latter shows that the proposed tests are also sensitive to linear serial dependence.

Moreover, Table~\ref{add_matrix_100} presents the simulation results for matrix time series when $T = 100$.

\begin{table*}
\caption{ Empirical sizes ($\%$) with DGP~1: $X_t \stackrel{i.i.d.}{\sim} N_d (\mathbf{0}, \Sigma)$, $\Sigma_{ij} = 0.5^{|i-j|} $, at the significance level $5\%$. }
\label{additional_table_dgp1}
\footnotesize
\tabcolsep=0pt
\begin{tabular*}{\textwidth}{@{\extracolsep{\fill}}l cccc c cccc c cccc@{\extracolsep{\fill}}}
\toprule%
& \multicolumn{4}{c}{LK} && \multicolumn{4}{c}{GK} && \multicolumn{4}{c}{BDK} \\

         \cline{2-5} \cline{7-10} \cline{12-15}

          $d$ & $V_{T,1}$ & $V_{T,3}$ & $P_{T,3}$ & $P_{T,6}$
         && $V_{T,1}$ & $V_{T,3}$ & $P_{T,3}$ & $P_{T,6}$ && $V_{T,1}$ & $V_{T,3}$ & $P_{T,3}$ & $P_{T,6}$ \\

        \cline{2-5} \cline{7-10} \cline{12-15}

        & \multicolumn{14}{c}{$T=100$} \\
        1 & 6.1 & 3.6 & 4.8 & 4.9 && 5.4 & 3.8 & 5.5 & 5.2  && 5.4 & 5.3 & 5.7 & 5.3 \\
        5 & 5.1 & 5.2 & 4.7 & 5.0 && 4.4 & 4.6 & 5.8 & 5.7  && 4.0 & 5.4 & 5.3 & 5.3 \\
        10 & 4.9 & 5.2 & 5.6 & 5.3 && 5.1 & 5.5 & 5.7 & 5.3  && 5.4 & 4.9 & 5.3 & 5.4 \\
        20 & 4.8 & 5.4 & 5.1 & 5.6 && 5.6 & 6.0 & 5.3 & 5.2  && 5.9 & 5.9 & 5.4 & 5.2 \\
        40 & 5.1 & 5.7 & 4.5 & 5.0 && 5.3 & 5.2 & 4.9 & 4.6  && 5.3 & 5.5 & 4.9 & 4.5 \\
        80 & 5.8 & 5.4 & 5.2 & 5.0 && 5.7 & 5.8 & 5.4 & 5.2  && 5.6 & 5.9 & 5.3 & 5.1 \\
        160 & 5.1 & 5.4 & 5.1 & 5.3 && 4.4 & 4.3 & 5.1 &5.5   && 4.5 & 4.3 & 5.2 & 5.4 \\
         & \\
         & \multicolumn{14}{c}{$T=200$} \\
         
         1 & 6.4 & 4.8 & 5.2 & 5.1 && 4.3 & 5.2 & 4.9 & 4.5  && 5.0 & 4.8 & 5.6 & 4.2 \\
         5 & 5.0 & 5.3 & 5.4 & 5.3 && 5.1 & 4.8 & 5.8 & 5.4  && 5.1 & 5.3 & 5.2 & 5.1 \\
10 & 5.9 & 5.5 & 5.7 & 5.8 && 5.7 & 4.8 & 5.7 &  5.9 && 5.7 & 5.0 & 5.1 &5.5  \\
20 & 5.4 & 5.3 & 5.3 & 5.6 && 5.2 & 5.2 & 5.4 & 5.5  && 5.6 & 4.9 & 5.2 & 5.3 \\
40 & 5.4 & 5.2 & 4.8 & 5.3 && 5.6 & 5.0 & 4.6 & 4.8  && 5.4 & 5.2 & 4.7 & 4.7 \\
80 & 5.4 & 4.9 & 4.9 & 5.0 && 4.8 & 5.2 & 4.7 & 5.1  && 4.8 & 5.3 & 4.7 & 5.2 \\
160 & 5.1 & 4.8 &5.0  & 6.1 && 5.1 & 4.7 & 5.1 & 5.6  && 5.0 & 4.7 & 5.1 & 5.6 \\
\midrule
\end{tabular*}
\end{table*}

\begin{table*}
\caption{ Empirical sizes ($\%$) with DGPs~2--3: $X_t \stackrel{i.i.d.}{\sim} t_\nu (\Sigma)$, at the significance level $5\%$. }
\label{additional_table_dgp2-3}
\footnotesize
\tabcolsep=0pt
\begin{tabular*}{\textwidth}{@{\extracolsep{\fill}}ll cccc c cccc c cccc@{\extracolsep{\fill}}}
\toprule%
&& \multicolumn{4}{c}{LK} && \multicolumn{4}{c}{GK} && \multicolumn{4}{c}{BDK} \\

         \cline{3-6} \cline{8-11} \cline{13-16}

        $\nu$ & $d$ & $V_{T,1}$ & $V_{T,3}$ & $P_{T,3}$ & $P_{T,6}$
         && $V_{T,1}$ & $V_{T,3}$ & $P_{T,3}$ & $P_{T,6}$ && $V_{T,1}$ & $V_{T,3}$ & $P_{T,3}$ & $P_{T,6}$ \\

        \cline{3-6} \cline{8-11} \cline{13-16}

        && \multicolumn{14}{c}{$T=100$} \\
        
        2 & 1 & 4.5 & 5.2 & 5.1 & 5.7 && 5.6 & 5.1 & 5.7 & 5.4 && 4.4 & 5.7 & 5.1 &5.2  \\
&5 & 5.6 & 6.0 & 5.1 & 5.7 && 5.7 & 6.1 & 5.7 & 5.5 && 5.4 &5.5  & 6.0 & 5.1 \\
&10& 5.3 & 5.2 & 3.5 & 5.1 && 5.7 & 6.1 & 4.0 & 5.2 && 5.2 & 6.0 & 5.0 & 5.2 \\
&20& 6.0 & 5.7 & 5.8 & 4.8 && 5.4 & 4.7 & 4.7 & 4.9 && 5.3 & 5.3 & 6.1 & 5.0 \\
&40& 5.1 & 6.5 & 4.4 & 5.0 && 4.5 & 6.0 & 5.0 & 5.0 && 4.2 & 5.9 & 5.8 & 5.4 \\
&80& 4.6 & 5.8 & 5.4 & 5.6 && 5.5 & 6.8 & 4.3 & 5.0 && 5.6 & 5.1 & 5.4 & 5.7 \\
&160& 5.3 & 5.2 & 3.5 & 5.1 && 5.7 & 6.1 & 5.3 & 5.2 && 6.0 & 6.4 & 5.3 & 5.5 \\

& \\

1 & 1 & 5.8 & 4.5 & 5.7 & 4.8 && 5.3 & 4.4 & 4.9 & 5.2 && 4.8 & 4.9 & 3.7 & 3.8 \\
&5 & 5.7 & 5.6 & 5.5 & 6.7 && 5.1 & 6.3 & 5.5 & 5.5 && 5.4 & 6.5 & 4.8 & 4.9 \\
&10& 5.8 & 5.8 & 6.1 & 7.4 && 6.3 & 5.3 & 5.6 & 5.3 && 7.3 & 8.0 & 4.7 & 5.0 \\
&20& 4.3 & 4.8 & 5.2 & 5.3 && 5.2 & 5.6 & 5.5 & 5.3 && 7.6 & 6.6 & 4.6 & 4.7 \\
&40& 6.0 & 6.3 & 6.2 & 6.5 && 5.7 & 6.0 & 5.4 & 6.2 && 7.4 & 8.8 & 6.8 & 5.4 \\
&80& 5.6 & 5.4 & 4.4 & 4.2 && 5.8 &5.8  & 4.9 & 4.9 && 8.6 & 9.6 & 5.6 & 5.3 \\
&160& 5.6 & 5.8 & 5.5 & 5.5 && 6.0 & 5.4 & 4.6 & 4.9 && 10.1 & 8.9 & 6.1 & 6.3 \\

         & \\
         && \multicolumn{14}{c}{$T=200$} \\
         
2 & 1 & 4.8 & 5.5 & 4.5 & 5.0 && 5.5 & 5.1 & 5.1 & 5.5 && 5.4 & 4.5 & 5.1 & 6.2 \\
&5 & 5.1 & 4.3 & 5.3 & 6.2 && 4.8 & 4.0 & 4.6 &5.5  && 5.7 & 4.4 & 5.3 & 5.7 \\
&10& 4.7 & 5.3 & 5.6 & 5.1 && 4.8 & 5.6 & 5.5 & 4.8 && 5.4 & 4.6 & 4.1 & 5.6 \\
&20& 5.8 & 5.2 & 6.3 & 5.9 && 5.3 & 5.4 & 5.8 & 5.6 && 5.4 &4.7  & 5.3 & 5.3 \\
&40& 4.5 & 6.3 & 5.4 & 6.4 && 5.4 & 5.6 & 4.8 & 5.9 && 5.0 & 5.6 & 5.1 & 5.3 \\
&80& 5.1 & 6.3 & 6.6 &6.0  && 5.5 & 5.9 & 5.8 & 5.9 && 6.7 & 5.6 & 5.9 & 5.3 \\
&160& 6.0 & 4.2 & 6.1 & 5.6 && 5.6 & 5.8 & 5.5 & 5.6 && 5.9 & 7.0 & 6.6 &5.1  \\
& \\

1 & 1 & 5.7 & 5.1 & 5.5 & 4.7 && 4.6 & 5.1 & 4.2 & 4.6 && 4.5 & 4.2 & 4.1 & 4.5 \\
&5 & 5.1 & 4.8 & 4.8 & 6.0 && 5.0 & 4.6 & 4.8 & 5.7 && 5.3 & 6.1 & 4.0 &  5.2\\
&10& 4.0 & 5.2 &4.1  & 5.0 && 4.3 & 4.2 & 4.0 & 4.9 && 7.0 & 7.6 & 5.3 & 3.8 \\
&20& 5.3 & 5.8 & 5.1 & 4.3 && 5.7 & 5.3 & 4.4 & 5.2 && 8.0 & 7.1 & 4.8 & 5.1 \\
&40& 6.4 & 5.6 & 5.8 & 5.2 && 5.5 & 5.1 & 4.6 & 4.8 && 8.3 & 10.8 & 6.3 & 6.0 \\
&80& 6.0 & 4.5 & 6.1 & 5.3 && 4.9 & 5.8 & 5.3 &4.5  &&9.6  & 10.8 & 5.9 &4.7  \\
&160& 4.9 & 5.9 & 5.9 & 5.3 && 5.7 & 6.3 & 5.6 & 5.3 && 9.7 & 10.8 & 5.7 & 4.5 \\
         
\midrule
\end{tabular*}
\end{table*}

\begin{table*}
\caption{ Empirical powers ($\%$) with DGPs~4--5: $X_t=\eta_t \odot \eta_{t-1} \odot \eta_{t-2}$, at the significance level $5\%$. }
\label{additional_table_dgp4-5}
\footnotesize
\tabcolsep=0pt
\begin{tabular*}{\textwidth}{@{\extracolsep{\fill}}ll cccc c cccc c cccc@{\extracolsep{\fill}}}
\toprule%
&& \multicolumn{4}{c}{LK} && \multicolumn{4}{c}{GK} && \multicolumn{4}{c}{BDK} \\

         \cline{3-6} \cline{8-11} \cline{13-16}

        $\eta_t$& $d$ & $V_{T,1}$ & $V_{T,3}$ & $P_{T,3}$ & $P_{T,6}$
         && $V_{T,1}$ & $V_{T,3}$ & $P_{T,3}$ & $P_{T,6}$ && $V_{T,1}$ & $V_{T,3}$ & $P_{T,3}$ & $P_{T,6}$ \\

        \cline{3-6} \cline{8-11} \cline{13-16}

         && \multicolumn{14}{c}{$T=100$} \\
        
$N(\mathbf{0},\Sigma)$ & 1 & 100.0 & 12.7 & 99.9 & 99.1 && 99.7 & 10.0 & 98.2  & 94.5  &&55.9  & 7.4 & 52.5 & 49.4 \\
& 5 & 99.5 & 15.0 & 98.8 & 96.3 && 95.1 & 11.7 & 89.3  & 83.5  && 55.5 & 7.6 & 52.6 &  47.6\\
& 10 & 98.5 & 15.1 & 97.0 & 94.8 && 87.3 & 10.9 & 78.3  & 74.0  && 48.9 & 7.8 & 44.8 & 41.0 \\
& 20 & 95.7 & 13.1 & 90.3 & 86.6 && 68.9 & 9.2 &  58.8 & 54.2  && 36.1 & 6.5 & 33.2 & 30.9 \\
& 40 & 87.4 & 13.8 & 80.2 & 76.9 && 47.2 & 8.7 & 41.3  & 36.1  && 26.6 & 7.7 & 22.5 & 21.4 \\
& 80 & 80.1 &10.4  & 65.3 &60.3  && 29.0 & 6.6 &  22.8 & 21.6  && 19.9 & 7.0 & 16.7 & 16.1 \\
& 160 & 65.1 & 7.7 & 48.3 &40.1  && 14.0 & 6.6 &  13.6 & 13.1  && 12.9 &7.0  & 12.4 & 12.3 \\

& \\

$t_2(\Sigma)$ & 1 & 99.9 & 14.9 & 99.9 & 99.5 && 99.7 & 11.9 & 99.3  &  96.5 && 41.4 & 10.3 & 45.7 & 43.2 \\
& 5 & 100.0 & 18.9 & 100.0 & 99.9 && 99.9 & 16.7 & 99.7  & 98.6  && 92.9 & 14.9 & 90.3 & 86.4 \\
& 10 & 100.0 & 20.4 & 100.0 & 100.0 && 99.8 & 18.7 & 99.9  & 99.5  && 98.2 & 18.5 & 97.3 & 95.8 \\
& 20 & 100.0 & 18.3 & 100.0 & 100.0 && 100.0 & 16.5 & 100.0  & 99.8  && 99.1 & 17.8 & 98.8 & 98.7 \\
& 40 & 100.0 & 19.3 & 100.0 & 100.0 && 100.0 & 19.0 & 100.0  & 100.0  &&99.3  & 20.8 & 99.1 & 99.2 \\
& 80 & 100.0 & 20.5 & 100.0 & 100.0 && 100.0 & 21.1 & 100.0  & 100.0  && 99.7 & 23.0 & 99.4 & 99.5 \\
& 160 & 100.0 &  21.1& 100.0 &100.0  && 100.0 &18.6  & 100.0  & 100.0  && 99.1 & 22.1 & 98.8 & 99.2 \\
        
         & \\
         && \multicolumn{14}{c}{$T=200$} \\
         
$N(\mathbf{0},\Sigma)$ & 1 & 100.0 & 14.9 & 100.0 & 100.0 && 100.0 & 13.6 & 100.0  & 100.0  && 99.8 & 6.9 & 98.8 & 94.7 \\
& 5 & 100.0 & 17.2 & 100.0 & 100.0 && 100.0 & 13.3 & 100.0  &  99.8 && 97.9 & 9.3 & 95.5 & 89.8 \\
& 10 & 100.0 & 15.4 & 100.0 & 100.0 && 100.0 & 10.3 & 99.2  &  97.3 && 93.6 & 8.5 & 85.0 &77.0  \\
& 20 &100.0  & 13.1 & 99.9 &99.5  && 97.1 & 8.6 &  92.4 & 86.5  && 82.8 &6.8  & 71.9 & 62.5 \\
& 40 & 99.3 &11.9  & 98.7 & 96.7 &&88.0  & 7.5 & 75.5  & 68.0  && 65.2 & 7.4 & 53.0 & 43.2 \\
& 80 & 98.1 & 9.6 & 93.9 & 87.9 && 63.1 & 7.1 & 47.0  &  38.0 && 45.6 & 7.0 & 32.6 & 26.5 \\
& 160 & 93.7 & 8.0 & 81.2 & 73.8 && 36.6 & 5.5 &  27.5 & 23.2  && 31.7 &5.6  & 23.0 & 19.1 \\
& \\

$t_2(\Sigma)$ & 1  & 100.0 & 13.1 & 100.0 & 100.0 && 100.0 & 11.6 & 100.0  & 100.0  && 92.1 & 10.1 & 91.8 & 89.3 \\
& 5 & 100.0 & 19.7 & 100.0 &100.0  && 100.0 & 18.6 &  100.0 & 100.0  && 99.8 & 17.2 & 99.7 & 99.7 \\
& 10 & 100.0 & 20.6 & 100.0 & 100.0 && 100.0 & 17.6 &  100.0 &  100.0 && 99.7 & 18.5 & 99.6 & 99.7 \\
& 20 & 100.0 & 16.2 & 100.0 & 100.0 && 100.0 & 15.5 & 100.0  &  100.0 && 99.9 & 19.1 & 99.9 & 99.9 \\
& 40 &  100.0& 19.2 & 100.0 &100.0  && 100.0 &17.4  & 100.0  &  100.0 && 99.9 & 21.2 & 99.8 & 99.9 \\
& 80 & 100.0 & 21.0 & 100.0 & 100.0 && 100.0 & 17.4 & 100.0  & 100.0  && 99.7 & 21.8 & 99.7 & 99.7 \\
& 160 & 100.0 & 18.1 & 100.0 & 100.0 && 100.0 & 16.6 & 100.0  & 100.0  && 100.0 & 22.8 & 99.9 & 100.0 \\
      
\midrule
\end{tabular*}
\end{table*}

\begin{table*}
\caption{ Empirical powers ($\%$) with DGP (GARCH): $X_t=\eta_t \odot \sigma_t$, $\sigma_t \odot \sigma_t=0.1+ 0.6 X_{t-1} \odot X_{t-1}+0.2 \sigma_{t-1} \odot \sigma_{t-1}$, at the significance level $5\%$. }
\label{additional_table_garch_0.05}
\footnotesize
\tabcolsep=0pt
\begin{tabular*}{\textwidth}{@{\extracolsep{\fill}}ll cccc c cccc c cccc@{\extracolsep{\fill}}}
\toprule%
&& \multicolumn{4}{c}{LK} && \multicolumn{4}{c}{GK} && \multicolumn{4}{c}{BDK} \\

         \cline{3-6} \cline{8-11} \cline{13-16}

        $\eta_t$ & $d$ & $V_{T,1}$ & $V_{T,3}$ & $P_{T,3}$ & $P_{T,6}$
         && $V_{T,1}$ & $V_{T,3}$ & $P_{T,3}$ & $P_{T,6}$ && $V_{T,1}$ & $V_{T,3}$ & $P_{T,3}$ & $P_{T,6}$ \\

        \cline{3-6} \cline{8-11} \cline{13-16}

        && \multicolumn{14}{c}{$T=100$} \\
        
$N(\mathbf{0},\Sigma)$ & 1  &45.9 &13.3 &41.9 &37.8 &&50.0&14.1&46.7&41.3  &&12.2 &6.4 &18.6 &18.0\\
& 5  &62.7 &22.6 &61.0 &56.8 &&48.8&19.9 &50.6&46.7  &&14.9 &7.2 &20.3 &20.6\\
& 10  &60.1 &24.9 &60.7 &57.2 &&43.0&20.7 &47.5&44.3  &&11.8 &8.6 &18.0 &19.5\\
& 20  &54.7 &24.3 &59.6 &56.8 &&36.8&18.8 &40.4&40.4  &&12.5 &8.8 &16.6 &18.5\\
& 40  &48.6 &21.9 &53.1 &51.8 &&27.8&14.5&33.2&34.0  &&8.6 &6.6 &13.1 &14.5\\
& 80  &37.5 &18.8 &44.7 &44.5 &&21.5&12.1&27.0&27.9  &&9.0 &6.7 &12.0 &14.5\\
& 160  &30.1 &13.5 &35.2 &36.3 &&15.1&8.8&18.7&20.1  &&6.7 &5.4 &8.2 &9.4\\

& \\
$t_2(\Sigma)$ & 1  &97.4 &83.0 &97.3 &95.8 &&97.0&83.1&97.0
&95.9  &&65.5 &41.8 &83.6 &85.0\\
& 5  &100.0 &98.1 &100.0 &99.9 &&99.9&95.8&99.9 &99.9  &&91.3 &74.2 &96.6 &96.9\\
& 10  &100.0 &99.2 &100.0 &100.0 &&99.9&97.9&99.9
&99.8  &&92.7 &80.4 &98.4 &98.7\\
& 20  &100.0 &99.9 &100.0 &100.0 &&100.0&98.9&100.0
&100.0  &&95.4 &86.6 &99.3 &99.5\\
& 40  &100.0 &99.7 &100.0 &100.0 &&99.9&99.0&100.0
&100.0  &&95.3 &87.9 &99.1 &99.4\\
& 80  &100.0 &100.0 &100.0 &100.0 &&100.0&99.4&100.0
&100.0 &&96.2 &90.3 &99.3 &99.7\\
& 160  &100.0 &100.0 &100.0 &100.0 &&100.0&99.5&100.0 &100.0  &&97.2 &92.7 &99.7 &99.8\\
        
         & \\
         && \multicolumn{14}{c}{$T=200$} \\
         
$N(\mathbf{0},\Sigma)$ & 1 &79.8 &20.4 &74.4 &64.7 & &83.7&22.9&79.7&71.3 &&36.2 &9.5 &41.2 &36.4\\
& 5 &93.3 &41.2 &91.6 &86.3 &&85.1&37.4&84.2&77.9  &&36.0 &12.7 &43.5 &43.5\\
& 10 &93.2 &45.4 &92.7 &88.9 &&79.9&36.6&81.4&76.2  &&32.0 &13.2 &42.8 &41.8\\
& 20 &90.4 &46.5 &90.3 &86.6 & &70.0&34.5&73.7&68.2  &&27.1 &12.4 &37.7 &38.0\\
& 40 &85.1 &42.6 &87.3 &82.8 & &58.3&28.9&62.1&61.3 &&18.2 &9.6 &29.6 &31.4\\
& 80 &78.2 &38.0 &80.7 &77.5 & &47.7&23.3&53.0&51.3 &&15.2 &7.6 &21.7 &24.0\\
& 160 &68.7 &33.2 &72.5 &69.7 & &34.7 &19.3&43.2&43.1  &&11.1 &8.4 &17.2 &19.2\\

& \\

$t_2(\Sigma)$ & 1 &100.0 &98.8 &100.0 &100.0 &&100.0&99.0&100.0&100.0  &&98.0 &89.6 &99.6 &99.7\\
& 5 &100.0 &100.0 &100.0 &100.0 &&100.0&100.0&100.0&100.0  &&99.6 &99.1 &100.0 &100.0\\
& 10 &100.0 &100.0 &100.0 &100.0 &&100.0&100.0&100.0&100.0  &&99.6 &99.5 &100.0 &100.0\\
& 20 &100.0 &100.0 &100.0 &100.0 &&100.0&100.0 &100.0&100.0   &&99.5 &99.3 &99.9 &100.0\\
& 40 &100.0 &100.0 &100.0 &100.0 &&100.0&100.0&100.0&100.0   &&99.6 &99.2 &99.9 &100.0\\
& 80 &100.0 &100.0 &100.0 &100.0 &&100.0&100.0&100.0&100.0 &&99.8 &99.2 &100.0 &100.0\\
& 160 &100.0 &100.0 &100.0 &100.0 &&100.0&100.0&100.0 &100.0  &&99.6 &99.3 &100.0 &100.0\\
      
\midrule
\end{tabular*}
\end{table*}

\begin{table*}
\caption{Empirical powers ($\%$) with DGP (VAR): $X_t = 0.3 X_{t-1} + \eta_t $ , at the significance level $5\%$.}
\label{additional_table_dgp_var}
\footnotesize
\tabcolsep=0pt
\begin{tabular*}{\textwidth}{@{\extracolsep{\fill}}ll cccc c cccc c cccc@{\extracolsep{\fill}}}
\toprule%
&& \multicolumn{4}{c}{LK} && \multicolumn{4}{c}{GK} && \multicolumn{4}{c}{BDK} \\

         \cline{3-6} \cline{8-11} \cline{13-16}

        $\eta_t$ &$d$ & $V_{T,1}$ & $V_{T,3}$ & $P_{T,3}$ & $P_{T,6}$
         && $V_{T,1}$ & $V_{T,3}$ & $P_{T,3}$ & $P_{T,6}$ && $V_{T,1}$ & $V_{T,3}$ & $P_{T,3}$ & $P_{T,6}$ \\

        \cline{3-6} \cline{8-11} \cline{13-16}

         && \multicolumn{14}{c}{$T=100$} \\
        
$N(\mathbf{0},\Sigma)$ & 1   & 46.2 & 6.4  & 32.8 & 27.0 && 56.7 & 6.4  & 41.7 & 34.2 && 75.0 & 6.6  & 59.6 & 50.6 \\
& 5   & 89.0 & 11.5 & 76.6 & 71.9 && 93.1 & 12.0 & 82.6 & 77.8 && 94.2 & 12.0 & 84.9 & 80.6 \\
& 10  & 96.8 & 25.1 & 94.4 & 94.4 && 98.0 & 23.7 & 95.6 & 94.6 && 98.1 & 23.7 & 95.9 & 94.9 \\
& 20  & 99.8 & 58.5 & 99.7 & 99.9 && 99.9 & 55.7 & 99.6 & 99.8 && 99.9 & 52.9 & 99.6 & 99.8 \\
& 40  & 100.0 & 97.3 & 100.0 & 100.0 && 100.0 & 94.3 & 100.0 & 100.0 && 100.0 & 94.5 & 100.0 & 100.0 \\
& 80  & 100.0 & 100.0 & 100.0 & 100.0 && 100.0 & 100.0 & 100.0 & 100.0 && 100.0 & 100.0 & 100.0 & 100.0 \\
& 160 & 100.0 & 100.0 & 100.0 & 100.0 && 100.0 & 100.0 & 100.0 & 100.0 && 100.0 & 100.0 & 100.0 & 100.0 \\

& \\
$t_2(\Sigma)$ & 1   & 67.5 & 6.1  & 48.3 & 38.6 && 73.8 & 6.2  & 57.9 & 47.2 && 91.0 & 8.5  & 81.8 & 73.5 \\
& 5   & 97.6 & 18.6 & 91.5 & 88.9 && 97.6 & 17.4 & 91.0 & 89.1 && 99.2 & 22.8 & 97.6 & 97.0 \\
& 10  & 99.5 & 32.3 & 98.2 & 98.7 && 99.3 & 29.0 & 97.8 & 98.1 && 99.7 & 40.1 & 99.8 & 99.8 \\
& 20  & 99.8 & 54.4 & 99.9 & 100.0 && 99.6 & 47.8 & 99.7 & 99.9 && 99.2 & 71.3 & 100.0 & 100.0 \\
& 40  & 99.9 & 81.2 & 100.0 & 100.0 && 99.8 & 71.9 & 99.9 & 100.0 && 99.0 & 95.7 & 99.9 & 100.0 \\
& 80  & 100.0 & 92.5 & 100.0 & 100.0 && 99.6 & 86.6 & 100.0 & 100.0 && 98.8 & 99.0 & 99.9 & 100.0 \\
& 160 & 100.0 & 96.2 & 100.0 & 100.0 && 99.4 & 90.2 & 99.9 & 100.0 && 98.6 & 99.7 & 100.0 & 100.0 \\

         & \\
         && \multicolumn{14}{c}{$T=200$} \\
         
$N(\mathbf{0},\Sigma)$ & 1   & 81.3 & 6.3  & 64.4 & 51.8 && 88.2 & 6.1  & 75.4 & 65.1 && 97.8 & 6.8  & 90.9 & 83.9 \\
& 5   & 99.9 & 13.5 & 98.3 & 96.0 && 99.9 & 13.1 & 99.2 & 97.5 && 100.0 & 13.7 & 99.5 & 98.0 \\
& 10  & 100.0 & 26.7 & 99.9 & 99.9 && 100.0 & 25.7 & 99.9 & 99.9 && 100.0 & 25.3 & 99.9 & 100.0 \\
& 20  & 100.0 & 64.8 & 100.0 & 100.0 && 100.0 & 58.9 & 100.0 & 100.0 && 100.0 & 58.8 & 100.0 & 100.0 \\
& 40  & 100.0 & 98.9 & 100.0 & 100.0 && 100.0 & 96.6 & 100.0 & 100.0 && 100.0 & 96.8 & 100.0 & 100.0 \\
& 80  & 100.0 & 100.0 & 100.0 & 100.0 && 100.0 & 100.0 & 100.0 & 100.0 && 100.0 & 100.0 & 100.0 & 100.0 \\
& 160 & 100.0 & 100.0 & 100.0 & 100.0 && 100.0 & 100.0 & 100.0 & 100.0 && 100.0 & 100.0 & 100.0 & 100.0 \\

& \\
$t_2(\Sigma)$& 1   & 96.1 & 6.1  & 85.4 & 74.6 && 97.9 & 7.1  & 91.4 & 83.2 && 99.7 & 9.6  & 98.8 & 97.4 \\
& 5   & 100.0 & 18.6 & 100.0 & 99.8 && 100.0 & 16.0 & 99.9 & 99.7 && 100.0 & 23.6 & 100.0 & 100.0 \\
& 10  & 100.0 & 33.2 & 100.0 & 100.0 && 100.0 & 30.0 & 100.0 & 100.0 && 99.9 & 46.7 & 100.0 & 100.0 \\
& 20  & 100.0 & 59.3 & 100.0 & 100.0 && 100.0 & 51.8 & 100.0 & 100.0 && 100.0 & 80.6 & 100.0 & 100.0 \\
& 40  & 100.0 & 83.8 & 100.0 & 100.0 && 100.0 & 75.8 & 100.0 & 100.0 && 99.9 & 98.0 & 100.0 & 100.0 \\
& 80  & 100.0 & 95.8 & 100.0 & 100.0 && 100.0 & 88.7 & 100.0 & 100.0 && 99.9 & 99.9 & 100.0 & 100.0 \\
& 160 & 100.0 & 98.4 & 100.0 & 100.0 && 100.0 & 94.2 & 100.0 & 100.0 && 99.6 & 99.8 & 100.0 & 100.0 \\

\midrule
\end{tabular*}
\end{table*}

\begin{table*}[t]
\caption{Empirical sizes and powers ($\%$) for matrix time series when $T=100$. }
\label{add_matrix_100}
\footnotesize
\tabcolsep=0pt
\begin{tabular*}{\textwidth}{@{\extracolsep{\fill}}ll rrrr r rrrr r rrrr@{\extracolsep{\fill}}}

\toprule
         &  & \multicolumn{4}{c}{LK} && \multicolumn{4}{c}{GK} && \multicolumn{4}{c}{BDK} \\

         \cline{3-6} \cline{8-11} \cline{13-16}

         $c$ & $d$ & $V_{T,1}$ & $V_{T,3}$ & $P_{T,3}$ & $P_{T,6}$
         && $V_{T,1}$ & $V_{T,3}$ & $P_{T,3}$ & $P_{T,6}$ && $V_{T,1}$ & $V_{T,3}$ & $P_{T,3}$ & $P_{T,6}$ \\

        \cline{3-6} \cline{8-11} \cline{13-16}

        0 & 2 & 4.4 &	4.7 &	4.4 &	5.5 &&	5.0 &	5.0 &	5.1 &	5.1 &&	5.7 &	4.3 &	4.8 &	4.8  \\
        & 5 & 4.2 &	6.3 &	5.8 &	5.4 &&	4.5 &	6.9 &	5.7 &	4.9 &&	4.0 &	6.8 &	6.2 &	6.2  \\
        & 8 & 4.4 &	5.0 &	4.5 &	5.3 &&	4.6 &	4.9 &	4.7 &	5.1 &&	4.9 &	5.1 &	4.5 &	5.8  \\
        &\\
        
        0.2 & 2 & 9.5 &	6.8 &	11.7 &	11.8 &&	9.7 &	6.3 &	12.0 &	11.3 &&	6.7 &	4.5 &	6.3 &	6.8  \\
        & 5 &  13.4 &	9.8 &	18.6 &	18.2 &&	9.3 &	9.0 &	13.8 &	14.5 &&	5.6 &	7.4 &	8.0 &	8.2 \\
        & 8 & 34.8 &	18.6 &	45.7 &	45.6 &&	30.4 &	16.3 &	41.1 &	39.6 &&	9.1 &	6.8 &	14.0 &	15.1 \\
        &\\
        
        0.3 & 2 & 24.9 &	16.8 &	31.2 &	31.9 &&	20.9 &	14.9 &	28.2 &	29.0 &&	7.8 &	6.9 &	10.6 &	12.7  \\
        & 5 & 36.6 &	23.6 &	46.1 &	46.5 &&	23.1 &	16.8 &	31.5 &	33.8 &&	10.2 &	7.2 &	13.5 &	16.5  \\
        & 8 & 68.4 &	48.1 &	78.2 &	78.7 &&	51.4 &	36.8 &	65.5 &	66.0 &&	17.0 &	14.1 &	30.4 &	34.4  \\
\midrule
\end{tabular*}
\end{table*}

\subsection{Simulations of Model Diagnostic Checking}

In this subsection, we illustrate the use of our methodology for model diagnostic checking of the GARCH model. The null model is the GARCH(1,1) model with independent innovations, that is, the time series $\{X_t\}$ is generated as follows:
\begin{align}
    X_t &=\sigma_t \eta_t ,\nonumber\\
    \sigma_t^2 &= \omega +\alpha X_{t-1}^2+ \beta \sigma_{t-1}^2, \nonumber
\end{align}
where $\{\eta_t\}$ is a sequence of i.i.d. innovations.

To evaluate the empirical size of the model diagnostic checking test, we consider the following error-generating processes (EGPs):

EGP~1: $\eta_t \stackrel{i.i.d.}{\sim} N(0,1)$.

EGP~2: $\eta_t \stackrel{i.i.d.}{\sim} t_4$.

Moreover, to evaluate the empirical power of the tests, consider the following EGP:

EGP~3: $\eta_t=\varepsilon_t \varepsilon_{t-1} \varepsilon_{t-2}$, $\varepsilon_t \stackrel{i . i . d .}{\sim} N(0,1)$.

Let $\omega = 0.2$, $\alpha = 0.1$, and $\beta = 0.5$. We consider sample sizes $T = 200, 400$. We then apply the residual bootstrap procedure in Section~6.2 to approximate the critical values. The bootstrap procedure is replicated 500 times.

Based on 1000 replications, the simulated empirical sizes and powers are reported in Table \ref{table_check_garch_0.05}. At the $5\%$ significance level, all methods yield empirical sizes close to the nominal level. Both the AutoHSIC-based and ADCV-based tests are able to capture such nonlinear dependence, with the AutoHSIC-based test showing higher power.

\begin{table*}[t]
\caption{ Empirical sizes and powers ($\%$) with EGPs~1--3. }
\label{table_check_garch_0.05}
\footnotesize
\tabcolsep=0pt
\begin{tabular*}{\textwidth}{@{\extracolsep{\fill}}l rrrr r rrrr r rrrr@{\extracolsep{\fill}}}
\toprule%
& \multicolumn{4}{c}{LK} && \multicolumn{4}{c}{GK} && \multicolumn{4}{c}{BDK} \\

\cline{2-5} \cline{7-10} \cline{11-15} \\[-1em]

$T$ & $\widehat{V}_{T,1}$ & $\widehat{V}_{T,3}$ & $\widehat{P}_{T,3}$ & $\widehat{P}_{T,6}$
         && $\widehat{V}_{T,1}$ & $\widehat{V}_{T,3}$ & $\widehat{P}_{T,3}$ & $\widehat{P}_{T,6}$ && $\widehat{V}_{T,1}$ & $\widehat{V}_{T,3}$ & $\widehat{P}_{T,3}$ & $\widehat{P}_{T,6}$ \\

\cline{2-5} \cline{7-10} \cline{11-15}

& \multicolumn{14}{c}{EGP~1} \\
200 & 5.2  & 5.2 & 5.2 & 6.2 && 6.0& 5.4& 5.6& 5.9 && 5.9 & 4.1 & 6.0 & 5.9 \\
400 & 5.8 & 3.6 & 5.0 & 5.7 & & 5.6& 4.4& 5.5& 5.8 && 5.7 & 4.7 & 6.0 & 5.4 \\

&\\
         & \multicolumn{14}{c}{EGP~2} \\
200 & 5.9 & 5.7 & 6.3 & 5.8 & & 5.6& 5.9& 5.6& 6.0 && 4.9 & 4.4 & 4.4 & 5.1  \\
400 & 5.6 & 6.1 & 5.0 & 5.1 & & 5.5&5.5& 5.3& 6.1 && 6.0 & 5.6 & 5.8 & 5.9  \\

&\\

& \multicolumn{14}{c}{EGP~3} \\
200 & 97.2 & 21.5 & 98.8 & 97.2 && 73.5& 23.2 & 80.6& 78.4 && 43.6 & 7.2 & 46.9 & 35.4 \\
400 & 100.0 & 39.5 & 100.0 & 100.0 && 95.8& 42.5& 99.4 & 97.6 && 72.5 & 23.9 & 84.4 & 71.2  \\

\midrule
\end{tabular*}
\end{table*}

\section{Additional Empirical Results}\label{sm_application}

\subsection{High-dimensional Realised Variance Series}

This section presents an additional application to a 30-dimensional realised variance time series.

In recent years, high-dimensional time series analysis has attracted considerable attention. Consequently, assessing serial independence prior to model construction has become an important preliminary step.
\citet{Hecq2021Granger} developed a test for Granger causality in high-dimensional VAR models based on penalised least squares estimation.
As an illustrative example, we apply the proposed method to assess the serial independence of the daily realised variances of 30 U.S.\ assets, for which \citet{Hecq2021Granger} directly constructed a vector autoregressive (VAR) model.

Specifically, given the realised variance series, \citet{Hecq2021Granger} employed a multivariate heterogeneous autoregressive (VHAR) model, originally proposed by \citet{corsi2009rv}, to capture their joint dynamics. However, the preliminary step of testing for serial independence of the original series was not considered. To address this gap, we apply the proposed method.

Define the 10-minute realised variance of asset $i$ on day $t$ as
\begin{equation}\nonumber
\operatorname{RV}_{it} = \sum_{j=1}^J r_{it,j}^2, \quad r_{it,j}=\ln P_{it,j}-\ln P_{it,j-1},
\end{equation}
where $P_{it,j}$ denotes the intraday price of asset $i$ at the $j$-th 10-minute interval on day $t$, with $i=1,\ldots,30$ and $t=1,\ldots,T$. Let
\begin{equation}\nonumber
    \mathbf{RV}_t = \left(\operatorname{RV}_{1t}, \ldots, \operatorname{RV}_{30 t}\right)^\top,
\end{equation}
which forms a 30-dimensional time series. The list of stocks is provided in Table~\ref{add_table_rv_stock name}, and the dataset is available at \url{https://github.com/Marga8}.

Following \citet{Hecq2021Granger}, we consider two sample periods: a longer period from March 2008 to February 2017 ($T=2236$) and a shorter period covering the year 2013 ($T=250$). The latter is expected to exhibit more stable dependence structures owing to its shorter time span, as it excludes two major events associated with substantial market instability, namely the 2008 global financial crisis and the 2011 U.S.\ debt-ceiling crisis.

Applying the proposed test of serial independence to $\{\mathbf{RV}_t\}$, we reject the null hypothesis at the $5\%$ significance level in all cases. The empirical $p$-values for the full sample are all zero. 
Table~\ref{table_pvalue_rv2013} reports the results for 2013, which likewise show uniformly small $p$-values. 
These findings indicate strong serial dependence in the 30-dimensional realised variance series, providing empirical support for directly modelling such data, as is common in the literature.

\begin{sidewaystable}
\caption{Stocks used in the application to high-dimensional time series of realised variances.}
\label{add_table_rv_stock name}
\tabcolsep=0pt%
\begin{tabular*}{\textwidth}{@{\extracolsep{\fill}}ccc ccc@{\extracolsep{\fill}}}
\toprule%
N. & Symbol & Issue name & N. & Symbol & Issue name \\
\midrule
1  & AAPL & APPLE INC & 16 & KO  & COCA-COLA CO \\
2  & AXP  & AMERICAN EXPRESS CO & 17 & MCD & MCDONALD’S CORP \\
3  & BA   & BOEING CO & 18 & MMM & 3M \\
4  & CAT  & CATERPILLAR & 19 & MRK & MERCK \& CO \\
5  & CSCO & CISCO SYSTEMS & 20 & MSFT & MICROSOFT CORPORATION \\
6  & CVX  & CHEVRON CORP & 21 & NKE & NIKE INC \\
7  & DD   & DOW CHEMICAL COMPANY & 22 & PFE & PFIZER INC \\
8  & DIS  & WALT DISNEY CO & 23 & PG  & PROCTER \& GAMBLE CO \\
9  & GE   & GENERAL ELEC & 24 & TRV & TRAVELERS COMPANIES INC \\
10 & GS   & GOLDMAN SACHS GROUP INC & 25 & UNH & UNITEDHEALTH GROUP INC \\
11 & HD   & HOME DEPOT INC & 26 & UTX & UNITED TECHNOLOGIES CORPORATION \\
12 & IBM  & INTL BUS MACHINE & 27 & V   & VISA INC \\
13 & INTC & INTEL CORP & 28 & VZ  & VERIZON COMMUNICATIONS INC \\
14 & JNJ  & JOHNSON \& JOHNSON & 29 & WMT & WALMART INC \\
15 & JPM  & JPMORGAN CHASE \& CO & 30 & XOM & EXXON MOBIL CORPORATION \\
\midrule
\end{tabular*}
\end{sidewaystable}

\begin{table*}[!t]
\caption{Empirical $p$-values for serial independence tests of realised variances in 2013.}
\label{table_pvalue_rv2013}
\footnotesize
\tabcolsep=0pt
\begin{tabular*}{\textwidth}{@{\extracolsep{\fill}}l rrrr r rrrr@{\extracolsep{\fill}}}
\toprule%
         &  \multicolumn{4}{c}{GK} && \multicolumn{4}{c}{BDK} \\

         \cline{2-5} \cline{7-10} 

          & $V_{T,1}$ & $V_{T,2}$ & $P_{T,2}$ & $P_{T,4}$ 
         && $V_{T,1}$ & $V_{T,2}$ & $P_{T,2}$ & $P_{T,4}$ \\

         \cline{2-5} \cline{7-10}
         
         $p$-value & 0.001 &0.006 &0.001 &0.000 && 0.000 &0.005 &0.000 &0.000 \\
\midrule
\end{tabular*}
\end{table*}

\subsection{Additional Results for CDS Returns}

Table~\ref{table_pvalue_matrix_cds} reports additional empirical $p$-values for the serial independence tests of the CDS return matrices considered in Section~7.2 of the main text.
The results complement the main application by reporting the pre-modelling serial independence tests.

\begin{table*}
\caption{Empirical $p$-values for serial independence tests of CDS returns over 2015--2018.}
\label{table_pvalue_matrix_cds}
\footnotesize
\tabcolsep=0pt
\begin{tabular*}{\textwidth}{@{\extracolsep{\fill}}l cccc c cccc@{\extracolsep{\fill}}}
\toprule%
&  \multicolumn{4}{c}{GK} && \multicolumn{4}{c}{BDK} \\
\cline{2-5} \cline{7-10} 
 & $V_{T,1}$ & $V_{T,2}$ & $P_{T,2}$ & $P_{T,4}$ 
&& $V_{T,1}$ & $V_{T,2}$ & $P_{T,2}$ & $P_{T,4}$ \\
\cline{2-5} \cline{7-10} 

$p$-value & 0.000 & 0.000 & 0.000 & 0.000 && 0.000 & 0.002 & 0.000 & 0.000 \\
\midrule

\end{tabular*}
\end{table*}

\section{Proofs}\label{sm_proof}

This section contains the proofs of the main theoretical results.

By the definition of the AutoHSIC between $X_t$ and $X_{t-m}$, we have
\begin{align}
    V_m &= \operatorname{HSIC}(X_t , X_{t-m})\nonumber \\
    &= \mathbb{E} \left[ k(X_t, X_t^\prime) \times \{ l(X_{t-m},X_{t-m}^\prime ) - 2 l(X_{t-m},X_{t-m}^{\prime \prime}) + l(X_{t-m}^{\prime \prime},X_{t-m}^{\prime \prime \prime} ) \} \right] \nonumber \\
    &= \mathbb{E} \{ d_k (X_t, X_t^{\prime}) d_l (X_{t-m}, X_{t-m}^{\prime}) \} =: \mathbb{E}
    \{ \mathcal{K}(Z_{t,m}, Z_{t,m}^{\prime})\},\label{Vm_expectation}
\end{align}
where $Z_{t,m}^{\prime} = (X_t^{\prime}, X_{t-m}^{\prime}) $, $Z_{t,m}^{\prime \prime} = (X_t^{\prime \prime}, X_{t-m}^{\prime \prime}) $ and $Z_{t,m}^{\prime \prime \prime} = (X_t^{\prime \prime \prime}, X_{t-m}^{\prime \prime \prime})$ are i.i.d. copies of  $Z_{t,m} = (X_t, X_{t-m}) $, 
$\mathcal{K}(z, z^\prime) = d_k (x, x^{\prime}) d_l (y, y^{\prime})$, $z=(x,y) \in \mathcal{X}^2$, $z^\prime =(x^\prime,y^\prime) \in \mathcal{X}^2$, and 
\begin{align}
    d_k(x, x^\prime) &= k(x, x^\prime) - \mathbb{E}_{X^\prime}\{k( x, X^\prime )\} - \mathbb{E}_{X}\{k( X, x^\prime )\} + \mathbb{E}_{X,X^\prime}\{k( X, X^\prime )\},\label{d_k d_l} \\
    d_l(y, y^\prime) &= l(y, y^\prime) - \mathbb{E}_{Y^\prime}\{l( y, Y^\prime )\} - \mathbb{E}_{Y}\{l( Y, y^\prime )\} + \mathbb{E}_{Y,Y^\prime}\{l( Y, Y^\prime )\},\nonumber
\end{align}
$(X^\prime, Y^\prime)$ is an i.i.d. copy of $(X,Y)$.

Note that in \eqref{Vm_expectation}, $V_m$ is an expectation functional. Therefore, it can be naturally estimated by a U-statistic, which is given by
\begin{equation}\label{VTm_fourth order U}
V_{T,m}=\binom{T-m}{4}^{-1} \sum_{m+1 \leq i<j<q<r \leq T} h\left(Z_{i,m}, Z_{j,m}, Z_{q,m}, Z_{r,m}\right),
\end{equation}
where the kernel function corresponding to this fourth-order U-statistic is
\begin{equation}\label{kernel function in U}
h\left(z_i, z_j, z_q, z_r\right)=\frac{1}{4!} \sum_{\left(i_1, i_2, i_3, i_4\right)}^{(i, j, q, r)} k\left(x_{i_1}, x_{i_2}\right)\left\{ l\left(y_{i_1}, y_{i_2}\right)-2 l\left(y_{i_1}, y_{i_3}\right) + l\left(y_{i_3}, y_{i_4}\right) \right\},
\end{equation}
$z_t=(x_t,y_t) \in \mathcal{X}^2$ for $t=1, \ldots, T$.

\subsection{Proof of Theorem~1}

Our analysis of the asymptotic properties of HSIC relies on the Cauchy--Schwarz inequality in the RKHS, in contrast to the triangle inequality typically used for distance covariance in \citet{jiang2024testing}.

\begin{lemma}[Cauchy-Schwarz inequality in RKHS]\label{lemma1_CSinequality}
    Let $k:\mathcal{X}\times\mathcal{X}\to\mathbb{R}$ be a symmetric and positive-definite kernel. Then, for any $x,y \in \mathcal{X}$, 
    \begin{equation}
       \left| k(x,y) \right| \leq \left| k(x,x) k(y,y) \right|^{1/2}.
    \end{equation}
\end{lemma}

\begin{proof}[Proof of Lemma~\ref{lemma1_CSinequality}.]

By the reproducing property of $k$, we have
$k(x,y) = \langle k(x,\cdot), k(\cdot,y) \rangle_{\mathcal{H}_k}$, so
$\left| k(x,y) \right| = \left| \langle k(x,\cdot),  k(\cdot,y) \rangle \right|$.
By the Cauchy--Schwarz inequality in the Hilbert space $\mathcal{H}$,
\begin{equation}\nonumber
\left|\langle u, v\rangle_{\mathcal{H}}\right| \leq \Vert u \Vert_{\mathcal{H}} \Vert v \Vert_{\mathcal{H}}, \quad \forall u, v \in \mathcal{H},
\end{equation}
we have $\left| k(x, y)\right| \leq \Vert k(x,\cdot) \Vert_{\mathcal{H}} \Vert k(\cdot, y) \Vert_{\mathcal{H}} = \left| k(x,x) k(y,y) \right|^{1/2}$.
\end{proof}

\begin{proof}[\bf Proof of Theorem~1.]

First, by Assumption~3, and using the $\beta$-mixing inequality in Lemma 1 of \citet{yoshihara1976limiting}, we obtain $\mathbb{E}(V_{T,m}) = V_m + o(1)$.

Since $V_{T,m}$ is a U-statistic based on the $\beta$-mixing sample $\{ Z_{t,m}\}_{t=m+1}^T$, we apply the law of large numbers for U-statistics under the $\beta$-mixing condition from \cite{arcones1998law}. In view of Theorem 1 (iii) in \cite{arcones1998law}, it suffices to show that for any $i,j,q,r$, for some $0 <\delta \leq 1$, 
\begin{equation}\nonumber
    \sup _{1 \leq i, j, q, r<\infty} \mathbb{E} \{ \left|h\left(Z_{i,m}, Z_{j,m}, Z_{q,m}, Z_{r,m}\right)\right|\left(\log ^{+}\left|h\left(Z_{i,m}, Z_{j,m}, Z_{q,m}, Z_{r,m}\right)\right|\right)^{(1+\delta)}\}<\infty ,
\end{equation}
where $\log ^{+}(x)=\max \{\log x, 0\}$.

By the same argument based on elementary inequalities as in the proof of Theorem~3.3 in \citet{jiang2024testing}, it suffices to show that, for some $\delta>0$,
\begin{equation}\label{moment condition in LLN}
    \sup _{m+1 \leq i <j <q <r< T} \mathbb{E}\{\left|h\left(Z_{i,m}, Z_{j,m}, Z_{q,m}, Z_{r,m}\right)\right|^{1+\delta}\}<\infty.
\end{equation}
We show that \eqref{moment condition in LLN} holds as follows.
\begin{align}
    & \sup _{m+1 \leq i <j <q <r< T} \mathbb{E}\{ \left|h\left(Z_{i,m}, Z_{j,m}, Z_{q,m}, Z_{r,m}\right)\right|^{1+\delta}\}\nonumber \\
    & \leq C \sup _{(i_1, i_2, i_3, i_4)}\left( \mathbb{E}\left|k_{i_1 i_2} l_{i_1 i_2}\right|^{1+\delta}+2 \mathbb{E}\left|k_{i_1 i_2} l_{i_1 i_3}\right|^{1+\delta} + \mathbb{E}\left|k_{i_1 i_2} l_{i_3 i_4}\right|^{1+\delta}\right) \nonumber \\
    & \leq C \sup _{(i_1, i_2, i_3, i_4)} \left( \mathbb{E} \left| k_{i_1 i_2} \right|^{2+2\delta} \mathbb{E} \left| l_{i_3 i_4} \right|^{2+2\delta} \right)^{1/2} \nonumber \\
    & \leq C \sup _{(i_1, i_2, i_3, i_4)} \left( \mathbb{E} \left| k_{i_1 i_1} k_{i_2 i_2} \right|^{1+\delta} \mathbb{E} \left| l_{i_3 i_3} l_{i_4 i_4} \right|^{1+\delta} \right)^{1/2} \nonumber \\
    & \leq C \sup _{(i_1, i_2, i_3, i_4)} \left(  \mathbb{E} \left| k_{i_1 i_1} \right|^{2+2\delta} \mathbb{E} \left| k_{i_2 i_2} \right|^{2+2\delta}   \mathbb{E} \left| l_{i_3 i_3} \right|^{2+2\delta} \mathbb{E} \left| l_{i_4 i_4} \right|^{2+2\delta}  \right)^{1/4} ,\nonumber
\end{align}
where $k_{i_1 i_2} = k(X_{i_1}, X_{i_2})$, $l_{i_1 i_2} = l(X_{i_1-m}, X_{i_2-m})$.
The first inequality follows from Minkowski's inequality, the second and fourth from H\"older's inequality, and the third from the Cauchy--Schwarz inequality in the RKHS (Lemma~\ref{lemma1_CSinequality}).
Then, by Assumption~2 (choosing any $0<\delta\leq r/2$), \eqref{moment condition in LLN} holds. 
\end{proof}

\subsection{Proof of Theorem~2 and 3}

Before proving Theorem~2, we present the following five lemmas.

\begin{lemma}[Hoeffding decomposition of a U-statistic] \label{lemma_Hoeffding}
For the kernel $h$ in \eqref{kernel function in U} of the U-statistic $V_{T,m}$, let $z_i=\left(x_i, y_i\right)$, $i=1,2,3,4$, and
    \begin{equation}\nonumber
         h_c \left( z_1, \ldots, z_c \right) = \mathbb{E} \left[h \left( z_1, \ldots, z_c, Z_{c+1}, \ldots, Z_4 \right) \right]
    \end{equation}
    for $c=1,2,3,4$, where $\{Z_t\}_{t =1}^4 $ are i.i.d. copies of $Z_{t,m} =(X_t,X_{t-m})$.
    Then,
\begin{align}
    h_1(z_1)= & \frac{1}{2}[\mathbb{E}\{d_k(x_1, X_t) d_l(y_1, X_{t-m})\}+ V_m ],\nonumber \\
h_2\left(z_1, z_2\right)= & \frac{1}{6}\left[d_k\left(x_1, x_2\right) d_l\left(y_1, y_2\right)+2 \mathbb{E}\{ d_k\left(x_1, X_t\right) d_l\left(y_1, X_{t-m}\right)\}\right. \nonumber\\ & \left. +2 \mathbb{E}\{ d_k\left(x_2, X_t\right) d_l\left(y_2, X_{t-m}\right)\} +V_m -\mathbb{E}\{ d_k\left(x_1, X_t\right) d_l\left(y_2, X_{t-m}\right)\}\right. \nonumber\\ &\left. -\mathbb{E}\{d_k \left(x_2, X_t\right) d_l\left(y_1, X_{t-m}\right) \}\right], \nonumber\\
h_3\left(z_1, z_2, z_3\right)= & \frac{1}{12} \left[2 \sum_{1 \leq i_1<i_2 \leq 3} d_k\left(x_{i_1}, x_{i_2}\right) d_l\left(y_{i_1}, y_{i_2}\right) \right. \nonumber\\
&\left. +2 \sum_{1 \leq i \leq 3} \mathbb{E}\{d_k\left(x_i, X_t\right) d_l\left(y_i, X_{t-m}\right)\} -\sum_{\left(i_1, i_2, i_3\right)}^{(1,2,3)} d_k\left(x_{i_1}, x_{i_2}\right) d_l\left(y_{i_1}, y_{i_3}\right) \right. \nonumber\\
& \left.-\sum_{1 \leq i_1 \neq i_2 \leq 3} \mathbb{E}\{d_k\left(x_{i_1}, X_t\right) d_l\left(y_{i_2}, X_{t-m} \right)\} \right], \nonumber\\
h_4\left(z_1, z_2, z_3, z_4\right)= & \frac{1}{4!} \sum_{\left(i_1, i_2, i_3, i_4\right)}^{(1,2,3,4)} d_k\left(x_{i_1}, x_{i_2}\right) \{d_l\left(y_{i_3}, y_{i_4}\right)+d_l\left(y_{i_1}, y_{i_2}\right)-2 d_l\left(y_{i_1}, y_{i_3}\right)\} .\nonumber
\end{align}

Moreover, define $h^{(1)}(z)=h_1(z)-V_m$ and $h^{(c)}=h_c\left(z_1, \ldots, z_c\right)-\sum_{j=1}^{c-1} \sum_{(c, j)} h^{(j)}\left(z_{i_1}, \ldots, z_{i_j}\right)- V_m$, where $\sum_{(c, j)}$ denotes the summation over all $j$-subsets $\left\{\left(i_1, \ldots, i_j\right)\right\}$ of $\{1, \ldots, c\}$. Then, 
\begin{equation}\label{decomposition VTm}
V_{T,m}=V_m +\sum_{c=1}^4\binom{4}{c} U_{T,m}\left(h^{(c)}\right),
\end{equation}
where $U_{T,m}\left(h^{(c)}\right)$ is the U-statistic based on kernel $h^{(c)}$.
\end{lemma}

\begin{proof}[Proof of Lemma~\ref{lemma_Hoeffding}.]
    The decomposition of $V_{T,m}$ in \eqref{decomposition VTm} follows directly from the Hoeffding decomposition of a U-statistic (Section 1.6 of \citet{lee1990u}). Moreover, straightforward calculations yield the explicit form of $h_c$; see, for example, Section 1.2 of the supplement of \citet{zhang2018conditional} or Lemma B.1 of the supplement of \citet{jiang2024testing}.
\end{proof}

\begin{lemma}[Mercer theorem for RKHS]\label{lemma_mercer for psd}
    Suppose that Assumption~4 holds, then,
    \begin{equation}\nonumber
\mathcal{K}\left(z, z^{\prime}\right):=d_k\left(x, x^{\prime}\right) d_l\left(y, y^{\prime}\right)=\sum_{\ell=1}^{\infty} \lambda_{\ell} \Phi_{\ell}(z) \Phi_{\ell}\left(z^{\prime}\right),
\end{equation}
where the series converges absolutely and uniformly on $\left(z, z^{\prime}\right) \in \mathcal{X}^2 \times \mathcal{X}^2$.
$\left\{\lambda_{\ell}\right\}_{\ell=1}^{\infty}$ and $\left\{\Phi_{\ell}\right\}_{\ell=1}^{\infty}$ are the non-zero eigenvalues and orthonormal eigenfunctions corresponding to $\mathcal{K}(\cdot, \cdot)$, i.e.,
$\mathbb{E} [\mathcal{K} (z, Z_{t,m} ) \Phi_{\ell} (Z_{t,m} ) ]=\lambda_\ell \Phi_\ell (z) $.
\end{lemma}

\begin{proof}[Proof of Lemma~\ref{lemma_mercer for psd}.]
This follows from Theorem 2 of \cite{sun2005mercer}. To apply the theorem, we need to verify two conditions:
$(a)$ $\mathcal{K}\left(z, z^{\prime}\right)$ is positive-definite (i.e., a Mercer kernel as defined in \citealp{sun2005mercer}). 
$(b)$ $\mathbb{E}\mathcal{K}^2\left(Z, Z^{\prime}\right) < \infty$, where $Z^\prime = (X^\prime, Y^\prime)$ is an i.i.d. copy of $Z=(X,Y) \in \mathcal{X} \times \mathcal{X}$ with $X$ and $Y$ i.i.d. copies of $X_t$. 

For $(a)$, we first prove that $d_k(x,x^\prime)$ is positive-definite. In the RKHS $\mathcal{H}_k$, by the reproducing property, we have
\begin{equation}\nonumber
    d_k(x,x^\prime) = \langle k(x,\cdot)-\mathbb{E}k(X,\cdot), k(\cdot, x^\prime)-\mathbb{E}k(X,\cdot) \rangle_{\mathcal{H}_k}.
\end{equation}
Therefore, for $\forall n$, $\forall \{x_1,\ldots,x_n\} \in \mathcal{X}$, $\forall \{c_1,\ldots,c_n\} \in \mathbb{R}$, we have
\begin{equation}\nonumber
    \sum_{i, j=1}^n c_i c_j d_k\left(x_i, x_j\right)=\Vert  \sum_{i=1}^n c_i\left( k(x_i,\cdot) -\mathbb{E}k(X,\cdot) \right) \Vert _{\mathcal{H}_k}^2 \geq 0,
\end{equation}
which means $d_k(x,x^\prime)$ is a positive-definite kernel. 
Then, condition $(a)$ holds, since the product of two positive-definite kernels is also positive-definite.

For $(b)$, by H\"older's inequality,
\begin{equation}\nonumber
    \mathbb{E}\{ \mathcal{K}^2\left(Z, Z^{\prime}\right)\} =  \mathbb{E}\{ d_k^2(X,X^\prime) d_l^2(Y,Y^\prime) \} 
    \leq \sqrt{ \mathbb{E}\{ d_k^4(X,X^\prime) \} \mathbb{E}\{ d_l^4(Y,Y^\prime) \} },
\end{equation}
so we only need to show that $\mathbb{E}\left[ d_k^4(X,X^\prime) \right] < \infty$. 
\begin{align}
    \mathbb{E}\{ d_k^4(X,X^\prime) \} &\leq C \mathbb{E}k^4(X,X^\prime)
    \leq C \mathbb{E}[k^2(X,X) k^2(X^\prime,X^\prime)] \nonumber\\
    &= C \mathbb{E}k^2(X,X)  \mathbb{E}k^2(X^\prime,X^\prime) <\infty,\nonumber
\end{align}
where the second inequality holds by Lemma \ref{lemma1_CSinequality}. Then, condition $(b)$ holds. 
\end{proof}

\begin{lemma}\label{lemma_mds_proof}
    Define $\mathcal{F}_t=\sigma\left(X_t, X_{t-1}, \ldots\right)$. Under $H_0$, for any fixed $M$ and $L$, we have\\
$(a)$   The sequence
\( \left\{ \left(\Phi_{\ell}(Z_{t,m})\right)_{\ell=1,\ldots,L;\,m=1,\ldots,M}
\right\}_{t=1}^T
\)
forms a \(LM\)-dimensional martingale difference vectors with respect to
\(\{\mathcal{F}_t\}\).\\
    $(b)$ $\mathbb{E}\{ \Phi_{\ell_1}\left(Z_{t,m_1}\right) \Phi_{\ell_2}\left(Z_{t,m_2}\right) \}=\mathbf{I}\left(\ell_1=\ell_2, m_1=m_2\right) $.\\
    $(c)$ For $\forall \ell_1, \ell_2$, $i_1<j_1, i_2<j_2$, $\mathbb{E}\left\{\Phi_{\ell_1}\left(Z_{i_1,m}\right) \Phi_{\ell_1}\left(Z_{j_1,m}\right) \Phi_{\ell_2}\left(Z_{i_2,m}\right) \Phi_{\ell_2}\left(Z_{j_2,m}\right)\right\}=0$, except for the case when the index set $\{i_1, j_1\}$ is identical to $\{i_2, j_2\}$. 
\end{lemma}

\begin{proof}[Proof of Lemma~\ref{lemma_mds_proof}.]
    This result is stated as Lemma~B.3 in the Supplement of \cite{jiang2024testing}.
\end{proof}

\begin{lemma}\label{lemma4_E}
    Under $H_0$, $P \in \mathcal{M}_k^2(\mathcal{X}) $ and $P \in \mathcal{M}_l^2(\mathcal{X}) $, then $\mathbb{E} V_{T,m} =O\left((T-m)^{-2}\right)$.
\end{lemma}

\begin{lemma}\label{lemma5_Var}
    Under $H_0$, suppose that Assumption~4 holds. Then, for $c=3,4$, 
    \begin{equation*}
\operatorname{Var}\left(U_{T,m}\left(h^{(c)}\right)\right) \leq C(T-m)^{-3},
\end{equation*}
for some uniform constant $C > 0$.
\end{lemma}

Lemmas~\ref{lemma4_E} and~\ref{lemma5_Var} are derived from properties of $U$-statistics for $m$-dependent samples under $H_0$, as established by \cite{janson2023asymptotic}. They correspond to Lemmas~B.4 and~B.5 in the Supplement of \cite{jiang2024testing}.

\begin{proof}[\bf Proof of Theorem~2.]
The proof proceeds in three steps: $(a)$ the leading term; $(b)$ approximation of $\mathcal{K}(z,z^\prime)$; and $(c)$ joint convergence.

$(a)$ Leading term.
Under $H_0$, we have $V_m=0$, $h^{(1)}(z)=0$, and $h^{(2)}\left(z, z^{\prime}\right)=h_2\left(z, z^{\prime}\right)=d_k\left(x, x^{\prime}\right) d_l\left(y, y^{\prime}\right) / 6 = \mathcal{K}(z, z^\prime)/6$.
Therefore, by Hoeffding decomposition in Lemma \ref{lemma_Hoeffding},
\begin{equation}\nonumber
V_{T,m}=\binom{T-m}{2}^{-1} \sum_{i=m+1}^T \sum_{j=i+1}^T \mathcal{K}\left(Z_{i,m}, Z_{j,m}\right)+\sum_{c=3}^4\binom{4}{c} U_{T,m}\left(h^{(c)}\right) .
\end{equation}
From Lemma \ref{lemma4_E} and Lemma \ref{lemma5_Var}, we know that $\mathbb{E}\{U_{T-m}\left(h^{(c)}\right)\}^2=O\left((T-m)^{-3}\right)$ for $c=3,4$, hence by Markov inequality, 
\begin{equation}\nonumber
V_{T,m}=\binom{T-m}{2}^{-1} \sum_{i=m+1}^T \sum_{j=i+1}^T \mathcal{K}\left(Z_{i,m}, Z_{j,m}\right)+o_p\left(T^{-1}\right):=\tilde{V}_{T,m}+o_p\left(T^{-1}\right) ,
\end{equation}
where
\begin{equation}\nonumber
\tilde{V}_{T,m} :=\binom{T-m}{2}^{-1} \sum_{i=m+1}^T \sum_{j=i+1}^T \mathcal{K}\left(Z_{i,m}, Z_{j,m}\right),
\end{equation}
denotes the leading term that determines the asymptotic distribution of $V_{T,m}$.

$(b)$ Approximation of $\mathcal{K}\left(z, z^{\prime}\right)$.
By Lemma \ref{lemma_mercer for psd}, let
\begin{equation}\nonumber
\mathcal{K}^{(L)}\left(z, z^{\prime}\right)=\sum_{\ell=1}^L \lambda_{\ell} \Phi_{\ell}(z) \Phi_{\ell}\left(z^{\prime}\right) \to \mathcal{K}\left(z, z^{\prime}\right), \quad \text { as } L \to \infty .
\end{equation}
Furthermore, let
\begin{align}
\tilde{V}_{T,m}^{(L)} & =\binom{T-m}{2}^{-1} \sum_{i=m+1}^T \sum_{j=i+1}^T \mathcal{K}^{(L)}\left(Z_{i,m}, Z_{j,m}\right) \nonumber \\
& =\binom{T-m}{2}^{-1} \sum_{i=m+1}^T \sum_{j=i+1}^T \sum_{\ell=1}^L \lambda_{\ell} \Phi_{\ell}\left(Z_{i,m}\right) \Phi_{\ell}\left(Z_{j,m}\right) .\label{VTm_trunction}
\end{align}

Then, we have 
\begin{align}
    & T^2 \mathbb{E}\{\tilde{V}_{T,m}-\tilde{V}_{T,m}^{(L)}\}^2 \nonumber \\
= & T^2\binom{T-m}{2}^{-2} \mathbb{E}\{ \sum_{i=m+1}^T \sum_{j=i+1}^T \sum_{\ell=L+1}^{\infty} \lambda_{\ell} \Phi_{\ell}\left(Z_{i,m}\right) \Phi_{\ell}\left(Z_{j,m}\right)\}^2 \nonumber\\
= & T^2\binom{T-m}{2}^{-2} \mathbb{E}\{\sum_{i=m+1}^T \sum_{j=i+1}^T \sum_{\ell_1, \ell_2=L+1}^{\infty} \lambda_{\ell_1} \lambda_{\ell_2} \Phi_{\ell_1}\left(Z_{i,m}\right) \Phi_{\ell_2}\left(Z_{i,m}\right) \Phi_{\ell_1}\left(Z_{j,m}\right) \Phi_{\ell_2}\left(Z_{j,m}\right) \} \nonumber\\
\leq & C \sum_{\ell_1, \ell_2=L+1}^{\infty} \lambda_{\ell_1} \lambda_{\ell_2} \to 0, \quad \text { as } L \to \infty,\label{trucation_0}
\end{align}
where the second equality holds by Lemma \ref{lemma_mds_proof} $(c)$, 
and the inequality holds by H\"older's inequality 
\begin{align}
&\mathbb{E}\{ \Phi_{\ell_1}\left(Z_{i,m}\right) \Phi_{\ell_2}\left(Z_{i,m}\right) \Phi_{\ell_1}\left(Z_{j,m}\right) \Phi_{\ell_2}\left(Z_{j,m}\right) \} \nonumber \\ 
    & \leq \left( \mathbb{E} \Phi_{\ell_1}^4\left(Z_{i,m}\right) \right)^{1/4}
   \left( \mathbb{E} \Phi_{\ell_2}^4\left(Z_{i,m}\right) \right)^{1/4}
    \left( \mathbb{E} \Phi_{\ell_1}^4\left(Z_{j,m}\right) \right)^{1/4}
    \left( \mathbb{E} \Phi_{\ell_2}^4\left(Z_{j,m}\right) \right)^{1/4} \nonumber\\
    & = \left( \mathbb{E} \Phi_{\ell_1}^4\left(Z\right) \right)^{1/2}
    \left( \mathbb{E} \Phi_{\ell_2}^4\left(Z\right) \right)^{1/2} < \infty ,\nonumber
\end{align}
which is implied by $\mathbb{E}\left(\mathcal{K}^2(Z, Z)\right)<\infty$.
The final convergence holds since $\sum_{\ell=1}^{\infty} \lambda_{\ell}=\mathbb{E} \mathcal{K}(Z, Z)<\infty$ and
$\sum_{\ell_1, \ell_2 = L+1}^{\infty} \lambda_{\ell_1} \lambda_{\ell_2}=\left(\sum_{\ell =L+1}^{\infty} \lambda_\ell \right)^2$.

$(c)$ Joint convergence of $T ( \tilde{V}_{T,1}^{(L)}, \tilde{V}_{T,2}^{(L)}, \ldots, \tilde{V}_{T,M}^{(L)})^\top$.
By (\ref{VTm_trunction}), we have 
\begin{equation}\label{VTm_L_expansion}
\begin{aligned}
    T \tilde{V}_{T,m}^{(L)}=&\frac{T}{T-m-1} \sum_{\ell=1}^L \lambda_{\ell}\left[ \{\frac{1}{\sqrt{T-m}} \sum_{i=m+1}^T \Phi_{\ell}\left(Z_{i,m}\right)\}^2 \right. \\ &\left. -\frac{1}{T-m} \sum_{i=m+1}^T \Phi_{\ell}^2\left(Z_{i,m}\right)\right].
\end{aligned}
\end{equation}

According to Lemma \ref{lemma_mds_proof} $(a)$, we know that for any fixed $M$ and $L$, 
\begin{equation*}
     \left\{ \left(\Phi_{\ell}(Z_{t,m})\right)_{\ell=1,\ldots,L;\,m=1,\ldots,M}
\right\}_{t=1}^T
\end{equation*}
forms a sequence of martingale difference vectors.
Hence, for any fixed $M$, $L$, and $1 \leq m \leq M $, $1 \leq \ell \leq L $, we have
\begin{equation}\nonumber
    \frac{1}{T-m} \sum_{i=m+1}^T \Phi_{\ell}^2\left(Z_{i,m}\right) \xrightarrow{\mathrm{p}} \mathbb{E} \Phi_{\ell}^2(Z)=1
\end{equation}
by the weak law of large numbers, and
\begin{equation}\label{joint_clt}
\left( \frac{1}{\sqrt{T-m}} \sum_{i=m+1}^T \Phi_{\ell}\left(Z_{i,m}\right) \right)_{\ell = 1, \ldots,L; m=1, \ldots,M} \xrightarrow{\mathrm{d}} 
\left( G_{\ell,m} \right)_{\ell = 1, \ldots,L; m=1, \ldots,M} 
\end{equation}
by the central limit theorem for martingale difference sequences and Cram\'e--Wold device.
Here, $\left( G_{\ell,m} \right)_{\ell = 1, \ldots, L;\, m = 1, \ldots, M}$ is a random vector of dimension $LM$, whose components are i.i.d. standard normal random variables, in view of Lemma~\ref{lemma_mds_proof} $(b)$.

The continuous mapping theorem implies that
\begin{equation}\nonumber
\begin{aligned}
T & \left( \tilde{V}_{T,1}^{(L)}, \tilde{V}_{T,2}^{(L)}, \ldots, \tilde{V}_{T,M}^{(L)} \right )^\top \xrightarrow{\mathrm{d}} \\& \left(\sum_{\ell=1}^L \lambda_{\ell}(G_{\ell,1}^2-1), \sum_{\ell=1}^L \lambda_{\ell}(G_{\ell,2}^2-1), \ldots,\sum_{\ell=1}^L \lambda_{\ell}(G_{\ell,M}^2-1) \right)^\top.
\end{aligned}
\end{equation}

By (\ref{trucation_0}), letting $L \to \infty$ yields the final result.  
\end{proof}

\begin{proof}[\bf Proof of Theorem~3.] 
Since $\sigma^2 > 0$, the conclusion follows by applying the central limit theorem for non-degenerate U-statistics of stationary $\beta$-mixing processes (Theorem~1 of \cite{yoshihara1976limiting}).
\end{proof}

\subsection{Proof of Theorem~4}

For notational simplicity, we omit the index $m$ from $a_{ij,m}$ and $b_{ij,m}$.
To describe the asymptotic properties of the bootstrap test statistics, let $\operatorname{Pr}^*$ and $\mathbb{E}^*$ denote the conditional probability and expectation given $\{X_t\}_{t=1}^T$, respectively. We write $o_p^*(1)$ and $O_p^*(1)$ for sequences of random variables that converge to zero and are bounded in probability, respectively, conditional on $\{X_t\}_{t=1}^T$.

\begin{proof}[\bf Proof of Theorem~4.]
$(a)$.
(i) Leading term by 
\begin{equation}\nonumber
  \tilde{V}_{T,m}^*=\frac{1}{(T-m) (T-m-3)} \sum_{m+1 \leq i \neq j \leq T} \mathcal{K}\left(Z_{i,m}, Z_{j,m} \right) w_{i,m} w_{j,m}.  
\end{equation}

We need to show that $ (T-m) V_{T,m}^* - (T-m) \tilde{V}_{T,m}^* \xrightarrow{\mathrm{d}^*} 0$ in probability.
By Chebyshev's inequality, it suffices to verify that
\begin{equation}\label{equation_Estar}
\mathbb{E}^*\left[\frac{1}{T-m-3} \sum_{m+1 \leq i \neq j \leq T}\left\{a_{i j} b_{i j}-\mathcal{K}\left(Z_{i,m}, Z_{j,m} \right)\right\} w_{i,m} w_{j,m}\right]=0,
\end{equation}
and
\begin{align}
    & \operatorname{Var}^*\left[\frac{1}{T-m-3} \sum_{m+1 \leq i \neq j \leq T}\left\{a_{i j} b_{i j}-\mathcal{K}\left(Z_{i,m}, Z_{j,m} \right)\right\} w_{i,m} w_{j,m}\right] \nonumber\\
&=  \frac{1}{(T-m-3)^2} \sum_{m+1 \leq i \neq j \leq T}\left\{a_{i j} b_{i j}-\mathcal{K}\left(Z_{i,m}, Z_{j,m} \right)\right\}^2 \xrightarrow{\mathrm{p}} 0,\label{equation_Varstar}
\end{align}
where the equality follows from the independence and unit variance of the bootstrap weights.

Recall that $\mathcal{K}\left(Z_{i,m}, Z_{j,m} \right)=d_k \left(X_i, X_j\right) d_l \left(X_{i-m}, X_{j-m}\right)$. 
Equation (\ref{equation_Estar}) holds immediately. We now proceed to establish (\ref{equation_Varstar}).

By H\"older's inequality, we have
\begin{equation*}
\begin{aligned}
& \quad \sum_{m+1 \leq i \neq j \leq T}\left\{a_{i j} b_{i j}-\mathcal{K}\left(Z_{i,m}, Z_{j,m}\right)\right\}^2 \\
& \leq 2 \sum_{m+1 \leq i \neq j \leq T}\{a_{i j}-d_k\left(X_i, X_j\right)\}^2 b_{i j}^2+2 \sum_{m+1 \leq i \neq j \leq T}\{b_{i j}-d_l \left(X_{i-m}, X_{j-m}\right)\}^2 d_k^2\left(X_i, X_j\right) \\
& \leq 4 \sum_{m+1 \leq i \neq j \leq T}\{a_{i j}-d_k\left(X_i, X_j\right)\}^2\{b_{i j}-d_l\left(X_{i-m}, X_{j-m}\right)\}^2 \\
& \quad+4 \sum_{m+1 \leq i \neq j \leq T}\{a_{i j}-d_k \left(X_i, X_j\right)\}^2 d_l^2\left(X_{i-m}, X_{j-m}\right) \\
& \quad+2 \sum_{m+1 \leq i \neq j \leq T}\{b_{i j}-d_l\left(X_{i-m}, X_{j-m}\right)\}^2 d_k^2\left(X_i, X_j\right) \\
& \leq 4[\sum_{m+1 \leq i \neq j \leq T}\{a_{i j}-d_k \left(X_i, X_j\right)\}^4]^{1 / 2}[\sum_{m+1 \leq i \neq j \leq T}\left[b_{i j}-d_l \left(X_{i-m}, X_{j-m}\right)\right]^4]^{1 / 2} \\
& \quad+4[\sum_{m+1 \leq i \neq j \leq T}\left[a_{i j}-d_k \left(X_i, X_j\right)\right]^4]^{1 / 2}[\sum_{m+1 \leq i \neq j \leq T} d_l^4\left(X_{i-m}, X_{j-m}\right)]^{1 / 2}   \\
& \quad+2[\sum_{m+1 \leq i \neq j \leq T}\left[b_{i j}-d_l \left(X_{i-m}, X_{j-m}\right)\right]^4]^{1 / 2}
[ \sum_{m+1 \leq i \neq j \leq T} d_k^4\left(X_i, X_j\right) ]^{1 / 2} .
\end{aligned}
\end{equation*}

Under Assumption~4, it is clear that 
\begin{equation}\nonumber
    \frac{1}{(T-m-3)^2} \sum_{m+1 \leq i \neq j \leq T} d_k^4\left(X_i, X_j\right)
\end{equation}
and 
\begin{equation}\nonumber
    \frac{1}{(T-m-3)^2} \sum_{m+1 \leq i \neq j \leq T} d_l^4\left(X_{i-m}, X_{j-m}\right)
\end{equation}
are bounded in probability. So we only need to prove that 
\begin{equation}\label{equation_inboot1_null}
\frac{1}{(T-m-3)^2} \sum_{m+1 \leq i \neq j \leq T}\{ a_{i j}-d_k \left(X_i, X_j\right)\}^4 \xrightarrow{\mathrm{p}} 0,
\end{equation}
as
\begin{equation}\nonumber
\frac{1}{(T-m-3)^2} \sum_{m+1 \leq i \neq j \leq T}\left[b_{i j}-d_l \left(X_{i-m}, X_{j-m}\right)\right]^4 \xrightarrow{\mathrm{p}} 0, \; \text{is similar}.
\end{equation}

By the definitions of $a_{ij}$ and $d_k(X_i,X_j)$, and the inequality that $\mid a+b\mid^p \leq 2^{p-1} ( \mid a\mid^p + \mid b\mid^p ) $, for $a, b \in \mathbb{R}$ and $p>1$, we have
\begin{align}
    \mathbb{E}\{a_{i j}-d_k \left(X_i, X_j\right)\}^4
\leq& C \left[ \mathbb{E}\{\frac{\sum_{t=m+1}^T k\left(X_t, X_i\right)}{T-m-2}- \mathbb{E} \left( k \left( X_t, X_i  \right) \mid X_i \right) \}^4 \right. \nonumber\\
&+ \left. \mathbb{E}\{\frac{\sum_{m+1 \leq t \neq t^{\prime} \leq T} k\left(X_t, X_{t^{\prime}}\right)}{(T-m-2)^2}- \mathbb{E}\left( k\left( X_i, X_j \right) \right) \}^4 \right] .\nonumber
\end{align}

For the first term, under $H_0$, we note that
\begin{align}
    &\mathbb{E} \prod_{s=1}^4\{k \left(X_{i_s}, X_i\right)- \mathbb{E} \left( k \left( X_t, X_i  \right) \mid X_i \right) \} \nonumber\\
    =&\mathbb{E}\left[\mathbb{E}\left[\prod_{s=1}^4\{k\left(X_{i_s}, X_i\right)- \mathbb{E} \left( k \left( X_t, X_i  \right) \mid X_i \right) \} \mid X_i\right]\right]=0 \nonumber
\end{align}
for distinct 4-tuples $\left(i_1, i_2, i_3, i_4\right)$.
It follows that the first term is at most of order $O\left(T^{-1}\right)$.
Similarly, the second term is at most of order $O\left(T^{-1}\right)$. Therefore, 
$\mathbb{E}\{ a_{i j}-d_k \left(X_i, X_j\right)\}^4 \to 0$, and by Markov's inequality, \eqref{equation_inboot1_null} holds. Thus, \eqref{equation_Varstar} holds.

(ii) Approximation for $\mathcal{K}(z, z^\prime)$.

Similar to the proof of Theorem~2, let
\begin{equation}\nonumber
\mathcal{K}^{(L)}\left(z, z^{\prime}\right)=\sum_{\ell=1}^L \lambda_{\ell} \Phi_{\ell}(z) \Phi_{\ell}\left(z^{\prime}\right) \to \mathcal{K}\left(z, z^{\prime}\right), \quad \text { as } L \to \infty
\end{equation}
Furthermore, denote
\(
\tilde{V}_{T,m}^{(L)*}=\binom{T-m}{2}^{-1} \sum_{i=m+1}^T \sum_{j=i+1}^T \mathcal{K}^{(L)}\left(Z_{i,m}, Z_{j,m}\right)w_{i, m} w_{j, m}
\).

We want to show that as $L \to \infty$,
\begin{equation}\label{proof_boot_H0_appro}
\mathbb{E}^*\left(T\left|\tilde{V}_{T,m}^{(L) *}-\tilde{V}_{T,m}^*\right|\right)^2 \xrightarrow{\mathrm{p}} 0.
\end{equation}
Recall that $\mathcal{K} \left(z, z^{\prime}\right)=\sum_{\ell=1}^\infty \lambda_{\ell} \Phi_{\ell}(z) \Phi_{\ell}\left(z^{\prime}\right) $, we have
{\small
\begin{align}
& \mathbb{E}^*\left[\frac{T}{(T-m)(T-m-3)} \sum_{m+1 \leq i\neq j \leq T}\left\{\mathcal{K}^{(L)}\left(Z_{i,m}, Z_{j,m} \right)-\mathcal{K}\left(Z_{i,m}, Z_{j,m} \right)\right\} w_{i,m} w_{j,m}\right]^2 \nonumber \\
= & \frac{T^2}{(T-m)^2(T-m-3)^2} \sum_{m+1 \leq i\neq j \leq T}\{\sum_{\ell=L+1}^{\infty} \lambda_{\ell} \Phi_{\ell}\left(Z_{i,m}\right) \Phi_{\ell}\left(Z_{j,m}\right)\}^2 \nonumber \\
= & \frac{T^2}{(T-m)^2(T-m-3)^2} \sum_{\substack{m+1 \leq i\neq j \leq T \nonumber \\
\left| j-i \right| >m}}\{\sum_{\ell=L+1}^{\infty} \lambda_{\ell} \Phi_{\ell}\left(Z_{i,m} \right) \Phi_{\ell}\left(Z_{j,m} \right)\}^2 \nonumber \\
& +\frac{T^2}{(T-m)^2(T-m-3)^2}  \sum_{\substack{m+1 \leq i \neq j \leq T \nonumber \\
\left| j-i \right| \leq m}}\{\sum_{\ell=L+1}^{\infty} \lambda_{\ell} \Phi_{\ell}\left(Z_{i,m} \right) \Phi_{\ell}\left(Z_{j,m} \right)\}^2 \nonumber \\
\xrightarrow{\mathrm{p}} & \mathbb{E}\{\sum_{\ell=L+1}^{\infty} \lambda_{\ell} \Phi_{\ell}(Z) \Phi_{\ell}\left(Z^{\prime}\right)\}^2 =\sum_{\ell=L+1}^{\infty} \lambda_{\ell}^2 \nonumber
\end{align}}
where there are at most $2m(T-m)$ terms in the summation of $\sum_{\substack{m+1 \leq i \neq j \leq T \\
\left| j-i \right| \leq m}}$,
and $Z_{i,m}$ and $Z_{j,m}$ are independent for $\left| j-i \right| > m$.
By Assumption~4, 
$ \sum_{\ell=1}^\infty \lambda_\ell^2 = \mathbb{E} \mathcal{K}^2(Z, Z^\prime) < \infty $. 
Hence, \eqref{proof_boot_H0_appro} follows by $\sum_{\ell=L+1}^{\infty} \lambda_{\ell}^2 \to 0$, as $L \to \infty$.

(iii) Joint convergence of
\begin{equation}\nonumber
   T  \left( V_{T,1}^{(L)*}, V_{T,2}^{(L)*}, \ldots, V_{T,M}^{(L)*} \right)^\top  \xrightarrow{\mathrm{d}^*} \left(\xi_1, \xi_2, \ldots, \xi_M \right)^\top , \quad \text{in probability}.
\end{equation}

Similar to \eqref{VTm_L_expansion}, we have
\begin{align}
    T \tilde{V}_{T,m}^{(L)*}= & \frac{T}{T-m-3} \sum_{\ell=1}^L \lambda_{\ell}\left[\{\frac{1}{\sqrt{T-m}} \sum_{i=m+1}^T w_{i,m} \Phi_{\ell}\left(Z_{i,m}\right)\}^2\right. \nonumber \\
& \left.-\frac{1}{T-m} \sum_{i=m+1}^T w_{i,m}^2  \Phi_{\ell}^2\left(Z_{i,m}\right)\right] .\nonumber
\end{align}
By continuous mapping theorem, it suffices to show that, in probability,
\begin{align}
    \frac{1}{T-m} \sum_{i=m+1}^T w_{i,m}^2 \Phi_{\ell}^2\left(Z_{i,m} \right) &\xrightarrow{\mathrm{p}^*} 1, \nonumber \\
    \left( \frac{1}{\sqrt{T-m}} \sum_{i=m+1}^T w_{i,m} \Phi_{\ell}\left(Z_{i,m} \right)\right)_{\ell=1, \ldots, L ; m=1, \ldots, M} &\xrightarrow{\mathrm{d}^*}\left(G_{\ell,m}\right)_{\ell=1, \ldots, L ; m=1, \ldots, M}.\nonumber
\end{align}

Note that 
$\mathbb{E}^*\{\frac{1}{T-m} \sum_{i=m+1}^T w_{i,m}^2 \Phi_{\ell}^2\left(Z_{i,m} \right)\}=\frac{1}{T-m} \sum_{i=m+1}^T \Phi_{\ell}^2\left(Z_{i,m} \right) \xrightarrow{\mathrm{p}} 1$,
and
\begin{equation}\nonumber
    \operatorname{Var}^*\{ \frac{1}{T-m} \sum_{i=m+1}^T w_{i,m}^2 \Phi_{\ell}^2\left(Z_{i,m} \right)\} =\frac{1}{(T-m)^2} \sum_{i=m+1}^T \Phi_{\ell}^4\left(Z_{i,m} \right) \operatorname{Var}\left(w_{i,m}^2\right) \xrightarrow{\mathrm{p}} 0
\end{equation}
by weak law of large numbers for $m$-dependent sequences and Slutsky's theorem. Then, by Chebyshev's inequality,
\begin{equation}\nonumber
\frac{1}{T-m} \sum_{i=m+1}^T w_{i,m}^2 \Phi_{\ell}^2\left(Z_{i,m} \right) \xrightarrow{\mathrm{p}^*} 1, \quad \text { in probability. }
\end{equation}

Next, by the weak law of large numbers for $m$-dependent sequences,  for any fixed $\ell$ and $m$,
$\sum_{i=m+1}^T \{(T-m)^{-1 / 2} \Phi_{\ell}\left(Z_{i,m} \right)\}^2 \xrightarrow{\mathrm{p}} 1$. Then, 
\begin{equation}\nonumber
\frac{\sum_{i=m+1}^T \{(T-m)^{-1 / 2} \Phi_{\ell}\left(Z_{i,m} \right)\}^4 \mathbb{E}\left| w_{i,m} \right|^4 }{\left[\sum_{i=m+1}^T \{(T-m)^{-1 / 2} \Phi_{\ell}\left(Z_{i,m} \right)\}^2\right]^2} \xrightarrow{\mathrm{p}} 0,
\end{equation}
which implies that Lyapunov central limit theorem holds in probability.

Moreover, for any fixed $\ell_1$, $\ell_2$ and $m_1 \leq m_2$, using the independence between $\mathbf{w}(m_1)$ and $\mathbf{w}(m_2)$ for $m_1 \neq m_2$, we have
\begin{align}
    & \operatorname{Cov}^*\left[\frac{1}{\sqrt{T-m_1}} \sum_{i=m_1+1}^T w_i\left(m_1\right) \Phi_{\ell_1}\left(Z_i(m_1) \right), \frac{1}{\sqrt{T-m_2}} \sum_{i=m_2+1}^T w_i\left(m_2\right) \Phi_{\ell_2}\left(Z_i(m_2) \right)\right] \nonumber\\
= & \frac{1}{T-m_2} \sum_{i=m_2+1}^T \Phi_{\ell_1}\left(Z_i(m_1) \right) \Phi_{\ell_2}\left(Z_i(m_2) \right) \mathbf{I}\left(m_1=m_2\right) \xrightarrow{\mathrm{p}} \mathbf{I}\left(\ell_1=\ell_2, m_1=m_2\right),\nonumber
\end{align}
in view of Lemma \ref{lemma_mds_proof} $(b)$.

Therefore, by Cram\'er--Wold device, we obtain the joint convergence. Finally, the result follows from \eqref{equation_Varstar} and \eqref{proof_boot_H0_appro}.

$(b)$.
It is straightforward to verify that $\mathbb{E}^*\left( T V_{T,m}^*\right)=0$.
By Chebyshev's inequality, 
$ \operatorname{Pr}^*\left( \left| TV_{T,m}^* \right| >C \right) < \frac{\mathbb{E}^*( TV_{T,m}^* )^2 }{C^2} $.
Hence, we need to show that $\mathbb{E}^*( TV_{T,m}^* )^2 = O_p(1)$, and by Markov's inequality, it remains to show that $\mathbb{E}\left\{ \mathbb{E}^*( TV_{T,m}^* )^2 \right\} = O(1)$. By direct calculation, we have
\begin{align}
    \mathbb{E}\left\{ \mathbb{E}^*( TV_{T,m}^* )^2 \right\} 
    &=\mathbb{E}\left[ \mathbb{E}^*\{ \frac{T^2}{(T-m)^2 (T-m-3)^2}\left( \sum_{m+1 \leq i \neq j \leq T}w_{i,m} a_{ij} b_{ij} w_{j,m} \right)^2 \} \right] \nonumber\\
       & = \frac{T^2}{(T-m)^2 (T-m-3)^2} \sum_{m+1 \leq i \neq j \leq T} \mathbb{E} \left( a_{ij}^2 b_{ij}^2 \right) \nonumber\\
       & \leq \frac{T^2}{(T-m)^2 (T-m-3)^2} \sum_{m+1 \leq i \neq j \leq T} \{\mathbb{E} \left( a_{ij}^4 \right)\}^{1/2} \{\mathbb{E} \left( b_{ij}^4 \right)\}^{1/2} < \infty,\nonumber
\end{align}
where the second-to-last inequality follows from H\"older's inequality and the last inequality holds by Assumption~4. Hence we have $\mathbb{E}\left\{ \mathbb{E}^*( TV_{T,m}^* )^2 \right\} = O(1)$.
Similarly, we can prove that $T P_{T,M}^* = O_p^*(1) $ in probability.
\end{proof}

\subsection{Proof of Theorem~5--7}

Now, we study the asymptotic behaviour of $\widehat{V}_{T,m}$ in Section~6.
In Section~6, $V_{T,m}$ denotes the sample AutoHSIC between $\eta_t$ and $\eta_{t-m}$, where $X_t$ in Sections~2 and~3 is replaced by $\eta_t$.
Also, we denote $\eta_{t,m} := (\eta_t, \eta_{t-m})$.
Before proving Theorem~5, we introduce some notation and lemmas, following \cite{wang2021new}.
Let $k_{ij}=k(\eta_i, \eta_j) $, $l_{qr} = l(\eta_{q-m}, \eta_{r-m}) $, and 
$$
\bar{k}_{i j}=\frac{\partial g_i\left(\theta_0\right)}{\partial \theta} k_x\left(\eta_i, \eta_j\right)+\frac{\partial g_j\left(\theta_0\right)}{\partial \theta} k_y\left(\eta_i, \eta_j\right) ,
$$
$$
\bar{l}_{q r}=\frac{\partial g_{q-m}\left(\theta_0\right)}{\partial \theta} l_x\left(\eta_{q-m}, \eta_{r-m}\right)+\frac{\partial g_{r-m}\left(\theta_0\right)}{\partial \theta} l_y\left(\eta_{q-m}, \eta_{r-m}\right),
$$
$$
\breve{k}_{i j}=  \left(\frac{\partial g_i\left(\theta_0\right)}{\partial \theta}, \frac{\partial g_j\left(\theta_0\right)}{\partial \theta}\right)\left(\begin{array}{cc}
k_{x x}\left(\eta_i, \eta_j\right) & k_{x y}\left(\eta_i, \eta_j\right) \\
k_{x y}\left(\eta_i, \eta_j\right) & k_{y y}\left(\eta_i, \eta_j\right)
\end{array}\right) 
 \left(\frac{\partial g_i\left(\theta_0\right)}{\partial \theta}, \frac{\partial g_j\left(\theta_0\right)}{\partial \theta}\right)^{\top} ,
$$
$$
\begin{aligned}
\breve{l}_{q r}= \left(\frac{\partial g_{q-m}\left(\theta_0\right)}{\partial \theta}, \frac{\partial g_{r-m}\left(\theta_0\right)}{\partial \theta}\right)
\left(\begin{array}{cc}
l_{x x}\left(\eta_{q-m}, \eta_{r-m}\right) & l_{x y}\left(\eta_{q-m}, \eta_{r-m}\right) \\
l_{x y}\left(\eta_{q-m}, \eta_{r-m}\right) & l_{y y}\left(\eta_{q-m}, \eta_{r-m}\right)
\end{array}\right) \\ \times
 \left(\frac{\partial g_{q-m}\left(\theta_0\right)}{\partial \theta}, \frac{\partial g_{r-m}\left(\theta_0\right)}{\partial \theta}\right)^{\top},
\end{aligned}
$$
for $i, j, q, r \in\{m+1, m+2, \ldots, T\}$.

With this notation, we define
\begin{equation}\label{V_Tm_ab}
\begin{aligned}
    V_{T,m}^{(a b)} =& \binom{T-m}{2}^{-1} \sum_{m+1 \leq i<j \leq T} k_{i j}^{(a b)} l_{i j}^{(a b)}+\binom{T-m}{4}^{-1} \sum_{m+1 \leq i<j<q<r \leq T} k_{i j}^{(a b)} l_{q r}^{(a b)} \\
 &-2\binom{T-m}{3}^{-1} \sum_{m+1 \leq i<j<q \leq T} k_{i j}^{(a b)} l_{i q}^{(a b)},
\end{aligned}
\end{equation}
for $a \in\{1,2\}$ and $b \in\{1, \ldots, a+1\}$, where 
$k_{i j}^{(11)}=\bar{k}_{i j}, l_{i j}^{(11)}=l_{i j}, k_{i j}^{(12)}=k_{i j}, l_{i j}^{(12)}=\bar{l}_{i j}$,
$k_{i j}^{(21)}=\breve{k}_{i j}, l_{i j}^{(21)}=l_{i j}, k_{i j}^{(22)}=k_{i j}, l_{i j}^{(22)}=\breve{l}_{i j}, k_{i j}^{(23)}=\bar{k}_{i j}$, and $l_{i j}^{(23)}=\bar{l}_{i j}^{\top}$.

Furthermore, they can be rewritten as U-statistics:
\begin{equation}\nonumber
V_{T,m}^{(ab)}=\binom{T-m}{4}^{-1} \sum_{m+1 \leq i<j<q<r \leq T} h^{(ab)} (\varsigma_{i,m}, \varsigma_{j,m}, \varsigma_{q,m}, \varsigma_{r,m} ) ,
\end{equation}
where
$ \varsigma_{t,m} =\left(\eta_t, \frac{\partial g_t\left(\theta_0\right)}{\partial \theta}, \eta_{t-m}, \frac{\partial g_{t-m}\left(\theta_0\right)}{\partial \theta}\right) \in \mathbb{R}^d \times \mathbb{R}^{p \times d} \times \mathbb{R}^d \times \mathbb{R}^{p \times d} $,
$ \varsigma_{t} =\left(\eta_t, \frac{\partial g_t\left(\theta_0\right)}{\partial \theta} \right) \in \mathbb{R}^d \times \mathbb{R}^{p \times d} $, and the kernel function $h^{(ab)}$ is
\begin{equation}\nonumber
    h^{(a b)}\left(\varsigma_{i,m}, \varsigma_{j,m}, \varsigma_{q,m}, \varsigma_{r,m}\right)=\frac{1}{4!} \sum_{(t, u, v, w)}^{(i, j, q, r)} k_{t u}^{(a b)} ( l_{t u}^{(a b)}+ l_{v w}^{(a b)}-2 l_{t v}^{(a b)} ).
\end{equation}

Lemma \ref{lemma_expansion} offers an important expansion of $\widehat{V}_{T,m}$.

\begin{lemma}\label{lemma_expansion}
    $\widehat{V}_{T,m}$ allows for the following expansion:
\begin{equation}\label{VTm_hat_expansion}
    \begin{aligned}
\widehat{V}_{T,m}
=& V_{T,m} + \zeta_T^{\top} V_{T,m}^{(11)} +\zeta_T^{\top} V_{T,m}^{(12)} +\frac{1}{2} \zeta_T^{\top} V_{T,m}^{(21)} \zeta_T \\
 &+\frac{1}{2} \zeta_T^{\top} V_{T,m}^{(22)} \zeta_T+\zeta_T^{\top} V_{T,m}^{(23)} \zeta_T +R_{T,m} ,
\end{aligned}
\end{equation}
where $V_{T,m}$ is defined in the same way as (9) with $X_t$ replaced by $\eta_t$, and $V_{T,m}^{(ab)}$ is defined in \eqref{V_Tm_ab}, $\zeta_T =\widehat{\theta}_T-\theta_0$, and $R_{T,m}$ is the remainder term.
\end{lemma}

\begin{proof}[Proof of Lemma~\ref{lemma_expansion}.]
   By applying Taylor expansions twice, once for the true innovation $\eta_t$ and once for the true parameter $\theta_0$, we can obtain this result after some algebraic manipulations.
For details, see the proof of Lemma~3.1 in the Appendix of \cite{wang2021new}. 
\end{proof}

Lemma \ref{lemma_degenerate} below investigates the degeneracy and non-degeneracy conditions of the corresponding U-statistic, which aids in examining the asymptotic distribution of $V_{T,m}$ and $V_{T,m}^{(ab)}$ under $H_0$.

\begin{lemma}\label{lemma_degenerate}
    Suppose that Assumptions~1, 6, and 9 hold. Then, under $H_0$,
    
    $(a)$ $ \mathbb{E}\{h (z_1, \eta_{t,m}, \eta_{t,m}^{\prime}, \eta_{t,m}^{\prime \prime}) \}=0$, 
 $\forall z_1 \in \mathbb{R}^d \times \mathbb{R}^d$;
 
$(b)$ $\mathbb{E}\{h^{(a b)}\left(s_1, \varsigma_{t,m}, \varsigma_{t,m}^{\prime}, \varsigma_{t,m}^{\prime \prime} \right)\}=0$, $\forall s_1 \in \mathbb{R}^d \times \mathbb{R}^{p \times d} \times \mathbb{R}^d \times \mathbb{R}^{p \times d}$, and each $a,b = 1,2$;

$(c)$ $\mathbb{E}\{h^{(23)}\left(s_1, \varsigma_{t,m}, \varsigma_{t,m}^{\prime}, \varsigma_{t,m}^{\prime \prime} \right)\}=\Upsilon$, 
$\forall s_1 \in \mathbb{R}^d \times \mathbb{R}^{p \times d} \times \mathbb{R}^d \times \mathbb{R}^{p \times d}$, where
\small{
$$
\begin{aligned}
\Upsilon= & 4 \mathbb{E}\{\frac{\partial g_2 \left(\theta_0 \right)}{\partial \theta} k_x\left(\eta_2, \eta_1 \right)\} \mathbb{E}\{\frac{\partial g_{2-m}\left(\theta_0 \right)}{\partial \theta} l_x\left(\eta_{2-m}, \eta_{1-m}\right)-\frac{\partial g_{3-m}\left(\theta_0\right)}{\partial \theta} l_x\left(\eta_{3-m}, \eta_{1-m}\right)\} \\
& +4 \mathbb{E}\{ \frac{\partial g_3\left(\theta_0\right)}{\partial \theta} k_x\left(\eta_3, \eta_1\right)\} \mathbb{E}\{\frac{\partial g_{3-m}\left(\theta_0\right)}{\partial \theta} l_x\left(\eta_{3-m}, \eta_{1-m}\right)-\frac{\partial g_{2-m}\left(\theta_0\right)}{\partial \theta} l_x\left(\eta_{2-m}, \eta_{1-m}\right)\} .
\end{aligned}
$$ }
Here, $\eta_{t,m}^{\prime}$ and $\eta_{t,m}^{\prime \prime}$ are i.i.d. copies of $\eta_{t,m}$, and $\varsigma_{t,m}^{\prime}$ and $\varsigma_{t,m}^{\prime \prime}$  are i.i.d. copies of $\varsigma_{t,m}$.
\end{lemma}

\begin{proof}[Proof of Lemma~\ref{lemma_degenerate}.]
    The result follows from straightforward calculations; see the proof of Lemma~3.2 in the Appendix of \cite{wang2021new} for details.
\end{proof}

\begin{remark}
From Lemma~\ref{lemma_expansion}, $\widehat{V}_{T,m}$ admits a decomposition into three components. 
The first is $V_{T,m}$, analysed in detail in Theorem~2. 
The second,
$\zeta_T^{\top} V_{T,m}^{(11)}+\zeta_T^{\top} V_{T,m}^{(12)} +\frac{1}{2} \zeta_T^{\top} V_{T,m}^{(21)} \zeta_T +\frac{1}{2} \zeta_T^{\top} V_{T,m}^{(22)} \zeta_T+\zeta_T^{\top} V_{T,m}^{(23)} \zeta_T$
captures the estimation effect. 
The third is a remainder term $R_{T,m}$, whose asymptotic negligibility is established in subsequent lemmas.

We next consider $V_{T,m}^{(ab)}$, which contributes to the estimation effect. 
Under $H_0$, for $a,b=1,2$, $V_{T,m}^{(ab)}$ is a degenerate U-statistic of order 1 by Lemma~\ref{lemma_degenerate}$(b)$. 
Using standard results for U-statistics of $m$-dependent stationary sequences (see Chapter~2.4.1 of \cite{lee1990u}), we have $T V_{T,m}^{(ab)} = O_p(1)$. 
By Assumption~7, the corresponding estimation effect is asymptotically negligible, that is, $\zeta_T^{\top} T V_{T,m}^{(ab)} = o_p(1)$, for $a,b=1,2$.

However, under $H_0$, the contribution associated with $V_{T,m}^{(23)}$ is not negligible when $\Upsilon \neq 0$. 
In this case, by the law of large numbers for U-statistics, $V_{T,m}^{(23)} = O_p(1)$, and hence $T \zeta_T^{\top} V_{T,m}^{(23)} \zeta_T = O_p(1)$, implying that the associated estimation effect is not negligible.
\end{remark}

A central limit theorem for martingale difference sequences is given in Lemma~\ref{lemma_clt}.

\begin{lemma}\label{lemma_clt}
    Suppose that Assumptions~1 and 6--9 hold. Then, under $H_0$, 
\begin{equation}\nonumber
\begin{aligned}
    \mathcal{G}_T:= & (\frac{1}{\sqrt{T-m}} \sum_{t=m+1}^{T} \mathcal{G}_{1 t}^{\top}, \frac{1}{\sqrt{T-m}} \sum_{t=m+1}^{T} \mathcal{G}_{2 t}^{\top} )^{\top}  \\ &\xrightarrow{\mathrm{d}} 
\mathcal{G}:= ( (G_{\ell,m} )^\top_{\ell=1,\ldots,L; m=1,\ldots,M}, \mathcal{W}^{\top} )^{\top},
\end{aligned}
\end{equation}
as $T \to \infty$, 
where 
$\mathcal{G}_{1 t}= (\Phi_\ell (\eta_{t,m}))_{\ell=1,\ldots,L; m=1,\ldots,M} \in \mathbb{R}^{LM}$,
$\mathcal{G}_{2 t}= \pi_t \in \mathbb{R}^p$,
$\mathcal{G}$ follows a multivariate normal distribution with mean zero and covariance matrix $\overline{\mathcal{G}}=E\left(\mathcal{G}_t \mathcal{G}_t^{\top} \right)$ with
$\mathcal{G}_t=\left(\mathcal{G}_{1 t}^{\top}, \mathcal{G}_{2 t}^{\top} \right)^{\top} $,
$\{ G_{\ell,m} \}_{\ell=1,\ldots,L; m=1,\ldots,M}$ is a sequence of i.i.d. standard normal random variables, 
$\mathcal{W}$ is defined in Assumption~7.
\end{lemma}

\begin{proof}[Proof of Lemma~\ref{lemma_clt}]
Under $H_0$, similar to Lemma~\ref{lemma_mds_proof}, we know that $\{ \mathcal{G}_{1 t}\}$ is a martingale difference sequence. 
From Assumption~7, we have that $\{ \mathcal{G}_{2 t}\}$ is a martingale difference sequence. It follows that $E(\mathcal{G}_{t} \mid \mathcal{F}_{t-1}) = 0$. Moreover, it is straightforward to know that $\mathbb{E} \Vert \mathcal{G}_t \Vert| ^2<\infty$ by Assumption~7 and 9.
By the central limit theorem for martingale difference sequences (see Chapter 5 in \citealp{white2001asymptotic}), we have that $\mathcal{G}_T \xrightarrow{\mathrm{d}} \mathcal{G}$, as $T \to \infty$.
Additionally, by Lemma~\ref{lemma_mds_proof} $(b)$, we obtain $\{G_{\ell,m} \}_{\ell=1,\ldots,L; m=1,\ldots,M}$ is a sequence of i.i.d. standard normal random variables. 
\end{proof}

The following two lemmas characterise the asymptotic behaviour of the remainder term $R_{T,m}$ in Lemma~\ref{lemma_expansion} under $H_0$ and $H_{a,M}$, respectively. These results are adapted from Lemmas~A.1 and~A.2 of \cite{wang2021new}. The proofs follow similar arguments, with the modification that the samples are $m$-dependent rather than i.i.d. under $H_0$.

\begin{lemma}\label{lemma_a.1}
    Suppose that Assumptions~1 and 6--9 hold. Then, under $H_0$,
    $T \Vert R_{T,m} \Vert =o_p(1)$.
\end{lemma}

\begin{lemma}\label{lemma_a.2}
    Suppose that Assumptions~1, 3 and 6--9 hold. Then, $\sqrt{T} \Vert R_{T,m} \Vert =o_p(1)$.
\end{lemma}

\begin{proof}[\bf Proof of Theorem~5]
By Lemmas~\ref{lemma_expansion} and \ref{lemma_a.1},
$T \widehat{V}_{T,m} = Z_{T,m}+o_p(1)$,
where
\begin{equation}\nonumber
    \begin{aligned}
Z_{T,m} := & T V_{T,m}+\zeta_T ^{\top} T  V_{T,m}^{(11)}+\zeta_T^{\top} T  V_{T,m}^{(12)} \\
& +\frac{1}{2} \zeta_T ^{\top}T  V_{T,m}^{(21)}  \zeta_T+\frac{1}{2} \zeta_T^{\top} T  V_{T,m}^{(22)} \zeta_T \\
& + \sqrt{T} \zeta_T ^{\top} V_{T,m}^{(23)} \sqrt{T} \zeta_T  .
\end{aligned}
\end{equation}

For $a,b = 1,2$, $V_{T,m}^{(ab)}$ is a first-order degenerate U-statistic by Lemma~\ref{lemma_degenerate}$(b)$. Consequently, by the theory of U-statistics with $m$-dependent samples (see Chapter 2.4.1 of \cite{lee1990u}), it follows that $T V_{T,m}^{(a b)}=O_p(1)$.
By Assumption~7, we have
\begin{equation}\nonumber
    \begin{aligned}
Z_{T,m} & =T  V_{T,m}+(\sqrt{T} \zeta_T)^{\top} V_{T,m}^{(23)} (\sqrt{T} \zeta_T ) + o_p(1) \\
& =T  V_{T,m} + (\sqrt{T} \zeta_T)^{\top} \Lambda^{(23)} (\sqrt{T} \zeta_T )+o_p(1),
\end{aligned}
\end{equation}
where
\begin{equation}\label{Lambda23}
    \Lambda^{(23)} = \mathbb{E}\{ h^{(23)} (\varsigma_{t,m}, \varsigma_{t,m}^{\prime}, \varsigma_{t,m}^{\prime \prime}, \varsigma_{t,m}^{\prime \prime \prime} ) \},
\end{equation}
and $\varsigma_{t,m}^{\prime}$, $\varsigma_{t,m}^{\prime \prime}$ and $\varsigma_{t,m}^{\prime \prime \prime}$ are i.i.d. copies of $\varsigma_{t,m}$.
The last equality holds by the law of large numbers for U-statistics with $m$-dependent samples (e.g., see \citealp{janson2023asymptotic}).

Following the same steps in the proof of Theorem~2, we can define the truncation version
\begin{equation}\nonumber
    Z_{T,m}^{(L)} := T  \tilde{V}_{T,m}^{(L)} + (\sqrt{T} \zeta_T)^{\top} \Lambda^{(23)} (\sqrt{T} \zeta_T ),
\end{equation}
where $\tilde{V}_{T,m}^{(L)}$ is defined in the same way as \eqref{VTm_trunction} with $Z_{t,m}$ replaced by $\eta_{t,m}$.
By \eqref{VTm_L_expansion}, we have 
\begin{equation}\nonumber
\begin{aligned}
    Z_{T,m}^{(L)} =& \frac{T}{T-m-1} \sum_{\ell=1}^L \lambda_{\ell} [ \{\frac{1}{\sqrt{T-m}} \sum_{i=m+1}^T \Phi_{\ell} (\eta_{i,m} ) \}^2-\frac{1}{T-m} \sum_{i=m+1}^T \Phi_{\ell}^2 (\eta_{i,m} ) ] \\
    &+ (\frac{1}{\sqrt{T-m}} \sum_{i=m+1}^T \pi_i )^{\top} \Lambda^{(23)} (\frac{1}{\sqrt{T-m}} \sum_{i=m+1}^T \pi_i ).
\end{aligned}
\end{equation}

According to Lemma~\ref{lemma_clt} and Assumption~7, the continuous mapping theorem implies that
\begin{equation}\nonumber
\left( Z_{T,m}^{(L)} \right)_{m=1,\ldots, M} \xrightarrow{\mathrm{d}} \left( \sum_{\ell=1}^L \lambda_{\ell} ( G_{\ell,m}^2-1 ) + \mathcal{W}^\top \Lambda^{(23)} \mathcal{W} \right)_{m=1,\ldots, M},
\end{equation}
as $T \to \infty$.
Let $L \to \infty$, similar to the part $(b)$ in the proof of Theorem~2, we obtain
\begin{equation}\nonumber
    \left( Z_{T,m} \right)_{m=1,\ldots, M} \xrightarrow{\mathrm{d}} \left( \sum_{\ell=1}^\infty \lambda_{\ell} ( G_{\ell,m}^2-1 ) + \mathcal{W}^\top \Lambda^{(23)} \mathcal{W} \right)_{m=1,\ldots ,M}.
\end{equation}

Therefore,
$ T  \left( \widehat{V}_{T,1}, \ldots, \widehat{V}_{T,M} \right)^\top \xrightarrow{\mathrm{d}} \left( \xi^{res}_1, \ldots, \xi^{res}_M \right)^\top $, where $\xi^{res}_m = \sum_{\ell=1}^\infty \lambda_{\ell} ( G_{\ell,m}^2-1 ) + \mathcal{W}^\top \Lambda^{(23)} \mathcal{W}$. 
\end{proof}

\begin{proof}[\bf Proof of Theorem~6.]
    The proof follows the same argument as the proof of Theorem~3.2 in \cite{wang2021new}.
\end{proof}

\begin{remark}
The main difference from \cite{wang2021new} is that, under $H_0$, the sample $\{\eta_{t,m}\}_{t=1}^T$ is $m$-dependent rather than i.i.d..
\end{remark}

To verify the validity of the residual bootstrap procedure, we first introduce some notation.
Let
\begin{equation}\label{h2^star}
h_2^* \left(z_1, z_2\right) =\mathbb{E}^*\{ h \left(z_1, z_2, \widehat{\eta}_{t,m}^{*}, \widehat{\eta}_{t,m}^{*\prime} \right)\}
\end{equation}
\begin{equation}\label{Lambda^23star}
    \Lambda^{(23 *)} =\mathbb{E}^*\{ h^{(23)}\left(\widehat{\varsigma}_{t,m}^{*}, \widehat{\varsigma}_{t,m}^{*\prime}, \widehat{\varsigma}_{t,m}^{*\prime \prime}, \widehat{\varsigma}_{t,m}^{*\prime \prime \prime} \right)\},
\end{equation}
where $\widehat{\eta}_{t,m}^{*\prime}$ is an i.i.d. copy of $\widehat{\eta}_{t,m}^{*}=\left( \widehat{\eta}_t^{*}, \widehat{\eta}_{t-m}^{*} \right)$,
and $\widehat{\varsigma}_{t,m}^{*\prime}$, $\widehat{\varsigma}_{t,m}^{*\prime \prime}$ and $\widehat{\varsigma}_{t,m}^{*\prime \prime \prime}$ are i.i.d. copies of $\widehat{\varsigma}_{t,m}^{*}$.
Moreover, let $\zeta_T^*=\widehat{\theta}_T^*-\widehat{\theta}_T$.

To prove Theorem~7, we use two lemmas, stated without proof, from Lemmas~A.4 and~A.5 in the Appendix of \cite{wang2021new}.

\begin{lemma}\label{lemma a.3 in wang}
    Suppose that Assumptions~1, 6--9 hold. Then, under $H_0$, for $\forall K_0>0$,

    $(a)$
    $ \sup _{\Omega_1} |  \frac{1}{(T-m)^4} \sum_{i,j,q,r} h \left(z_1, z_2,\left(\eta_i, \eta_{j-m}\right),\left(\eta_q, \eta_{r-m}\right)\right) -h_2 \left(z_1, z_2\right) |=o_p(1)$,
    where $\Omega_1=\left\{\left(z_1, z_2\right): \Vert z_s \Vert \leq K_0\right.$ for $\left.s=1,2\right\}$, and $h_2 \left(z_1, z_2\right)$ is defined as in Lemma~\ref{lemma_Hoeffding}.

    $\begin{aligned} (b)  \sup _{\Omega_2} \mid & \frac{1}{(T-m)^4} \sum_{i, j, q, r} h^{(23)}\left(\left(x_1, \varsigma_{i-m}\right),\left(x_2, \varsigma_{j-m}\right),\left(x_3, \varsigma_{q-m}\right),\left(x_4, \varsigma_{r-m}\right)\right) \\ & -\mathbb{E}\{h^{(23)}\left(\left(x_1, \varsigma_{t-m}\right),\left(x_2, \varsigma_{t-m}^{\prime} \right),\left(x_3, \varsigma_{t-m}^{\prime \prime} \right),\left(x_4, \varsigma_{t-m}^{\prime \prime \prime} \right)\right)\} \mid=o_p(1), \end{aligned}$ \\
    where $\Omega_2=\left\{\left(x_1, x_2, x_3, x_4 \right): \Vert x_{s} \Vert  \leq K_0\right.$ for $\left.s=1,2,3,4\right\}$.

    $\begin{aligned} (c) \sup _{\Omega_3} & \mid  \frac{1}{(T-m)^4} \sum_{i, j, q, r} h^{(23)}\left(\left(\varsigma_{i}, y_1 \right),\left(\varsigma_{j}, y_2 \right),\left(\varsigma_{ q}, y_3 \right),\left(\varsigma_{r}, y_4 \right)\right) \\ & -E\{ h^{(23)}\left(\left(\varsigma_t, y_1 \right),\left(\varsigma_t^{\prime}, y_2 \right),\left(\varsigma_t^{\prime \prime}, y_3 \right),\left(\varsigma_t^{\prime \prime \prime}, y_4 \right)\right)\} \mid =o_p(1),\end{aligned}$ \\
    where $\Omega_3=\left\{\left(y_1, y_2, y_3, y_4\right):\Vert  y_{s}\Vert  \leq K_0\right.$ for $\left.s=1,2,3,4\right\}$.
\end{lemma}

\begin{lemma}\label{lemma a.4 in wang}
    Suppose that Assumptions~1, 6--9 hold. Then, under $H_0$,

    $(a)$ $\sup _{\Omega_1}\left|h_2^* \left(z_1, z_2\right)-h_2 \left(z_1, z_2\right)\right|=o_p(1)$,
    where $\Omega_1, h_2 \left(z_1, z_2\right)$ and $h_2^* \left(z_1, z_2\right)$ are defined as in Lemma~\ref{lemma a.3 in wang} $(a)$, Lemma~\ref{lemma_Hoeffding} and \eqref{h2^star}, respectively;

    $(b)$ $\left|\Lambda^{(23 *)}-\Lambda^{(23)}\right|=o_p(1)$, where $\Lambda^{(23)}$ and $\Lambda^{(23 *)}$ are defined as in \eqref{Lambda23} and \eqref{Lambda^23star}, respectively.
\end{lemma}

\begin{proof}[\bf Proof of Theorem~7.]

We first prove the part $(b)$.
By Assumption~10 and 11, $\sqrt{T} \zeta_T^*=O_p^*(1)$ in probability.
Since $\left\{\widehat{\eta}_t^*\right\}_{t=1}^T$ is an i.i.d. sequence conditional on $\{X_t\}_{t=1}^T$, a similar argument as for Lemma~\ref{lemma_expansion} implies that
\begin{equation}\label{expansion boot}
    \begin{aligned}
\widehat{V}_{T,m}^*
=& V_{T,m}^{(0*)}+ \zeta_T^{*\top} V_{T,m}^{(11*)} +\zeta_T^{*\top} V_{T,m}^{(12*)} +\frac{1}{2} \zeta_T^{*\top} V_{T,m}^{(21*)} \zeta_T^* \\
 &+\frac{1}{2} \zeta_T^{*\top} V_{T,m}^{(22*)} \zeta_T^*+\zeta_T^{*\top} V_{T,m}^{(23*)} \zeta_T^* +R_{T,m}^* ,
\end{aligned}
\end{equation}
where 
$V_{T,m}^{(0*)}$, $V_{T,m}^{(ab*)}$ and $R_{T,m}^*$ are defined analogously to $V_{T,m}$, $V_{T,m}^{(ab)}$ and $R_{T,m}$, respectively, with $\eta_{t,m}$ and $\varsigma_{t,m}$ replaced by $\widehat{\eta}_{t,m}^*$ and $\widehat{\varsigma}_{t,m}^*$.

Moreover, by arguments similar to those in Lemma~\ref{lemma_degenerate} $(a)$ and the proof of Theorem~2, we obtain
\begin{equation}\label{VTm_L_expansion boot}
\begin{aligned}
    T V_{T,m}^{(0*)}=&\frac{T}{T-m-1} \sum_{\ell=1}^\infty \lambda_{\ell}^* \left[\{\frac{1}{\sqrt{T-m}} \sum_{i=m+1}^T \Phi_{\ell}^* \left(\widehat{\eta}_{i,m}^* \right)\}^2 \right. \\ 
    &\left. -\frac{1}{T-m} \sum_{i=m+1}^T \Phi_{\ell}^{*2}\left(\widehat{\eta}_{i,m}^* \right)\right] + o_p^*(1),
\end{aligned}
\end{equation}
where $\{\lambda_{\ell}^* \}_{\ell=1}^{\infty}$ and $\{\Phi_{\ell}^*(\cdot)\}_{\ell=1}^{\infty}$ are sequences of non-zero eigenvalues and orthonormal eigenfunctions of $h_2^*(z_1, z_2)$ in \eqref{h2^star}.
Furthermore, by \eqref{expansion boot}, \eqref{VTm_L_expansion boot}, Assumption~11, and similar arguments as for Lemmas~\ref{lemma_degenerate} $(b)$-$(c)$ and \ref{lemma_a.1}, we can show that
\begin{equation}\label{VT_hat star}
\begin{aligned}
T  \widehat{V}_{T,m}^* = &\sum_{\ell=1}^{\infty} \lambda_\ell^* [\{\frac{1}{\sqrt{T-m}} \sum_{i=m=1}^{T} \Phi_\ell ^*\left(\widehat{\eta}_{i,m}^* \right)\}^2 
-\frac{1}{T-m} \sum_{i=m+1}^T \Phi_{\ell}^{*2}\left(\widehat{\eta}_{i,m}^* \right) ]\\
&+(\sqrt{T-m} \zeta_T^{* \top}) \Lambda^{(23 *)}(\sqrt{T-m} \zeta_T^*) +o_p^*(1)\\
=& O_p^*(1), \text{ in probability}.
\end{aligned}
\end{equation}

It follows by a similar argument that $T  P_{T,M}^* = O_p^*(1)$ in probability.

For part $(a)$, let $\mathcal{G}^*_{1 t}= \left( \Phi_\ell^* (\widehat{\eta}_{t,m}^*)\right)_{\ell=1,\ldots,L; m=1,\ldots,M} \in \mathbb{R}^{LM}$,  
$\mathcal{G}_{2 t}^* = \pi_t^*$, and
\begin{equation}\nonumber
    \mathcal{G}_T^* = (\frac{1}{\sqrt{T-m}} \sum_{t=m+1}^{T} \mathcal{G}_{1 t}^{*\top}, \frac{1}{\sqrt{T-m}} \sum_{t=m+1}^{T} \mathcal{G}_{2 t}^{*\top} )^{\top},
\end{equation}
where $\pi_t^*$ is defined in Assumption~10.
Then, let $\mathcal{G}_t^*=\left(\mathcal{G}_{1 t}^{* \top}, \mathcal{G}_{2 t}^{* \top}\right)^\top$. 
As for Lemma \ref{lemma_clt}, by Assumption~10 and 11, we can show the following joint convergence result 
\begin{equation}\nonumber
\mathcal{G}_T^* \xrightarrow{\mathrm{d}^*} \mathcal{G} \;\; \text{in probability},
\end{equation}
as $T \to \infty$, where $\mathcal{G}$ is defined as in Lemma~\ref{lemma_clt}.

Next, by Lemma \ref{lemma a.4 in wang} $(a)$ and Corollary XI.9.4 $(a)$ in \cite{dunford1963}, we obtain
\begin{equation}\label{lambda approximation}
    \left|\lambda_\ell^*-\lambda_\ell \right|=o(1).
\end{equation}
Therefore, the conclusion holds by \eqref{VT_hat star}--\eqref{lambda approximation}, Lemma \ref{lemma a.4 in wang} $(b)$, and the continuous mapping theorem. 
\end{proof}

\newpage
\bibliographystyle{apalike}
\bibliography{ref}

\end{document}